\documentclass[letterpaper,USenglish,numberwithinsect,cleveref,autoref]{lipics-v2021}
\pdfoutput=1
\bibstyle{plainurl}

\usepackage{pgf}
\usepackage{tikz}%% for images

\usetikzlibrary{arrows,shapes,automata,backgrounds,petri,positioning}
\usetikzlibrary{decorations.pathmorphing}
\usetikzlibrary{decorations.shapes}
\usetikzlibrary{decorations.text}
\usetikzlibrary{decorations.fractals}
\usetikzlibrary{decorations.footprints}
\usetikzlibrary{decorations.pathreplacing}

\usetikzlibrary{shadows}
\usetikzlibrary{calc}
\usetikzlibrary{spy}

\tikzset{snake it/.style={decorate, decoration=snake}}
\tikzset{
    position label/.style={
       below = 3pt,
       text height = 1.5ex,
       text depth = 1ex
    },
    position label2/.style={
       above = 3pt,
       text height = 1.5ex,
       text depth = 1ex
    },
   brace/.style={
     decoration={brace, mirror},
     decorate
   }
}

\usepackage{graphicx} % Required for inserting images
\usepackage{microtype}
\usepackage{amsmath,amsfonts}
\usepackage{epsfig,amssymb,amstext,xspace,amsthm,stmaryrd}
\usepackage[shortlabels,inline]{enumitem}
\usepackage{comment}
\usepackage{calc,bbm}
\usepackage[ruled, linesnumbered, vlined, nokwfunc, noalgohanging]{algorithm2e}

\usepackage{cleveref}
\usepackage{ifthen}
\makeatletter
\newcommand{\clonelabel}[2]{\@bsphack
  \expandafter\ifx\csname r@#2\endcsname\relax
  \else\protected@write\@auxout{}{\string\newlabel{#1}%
    {\csname r@#2\endcsname}}%
  \fi
  \expandafter\ifx\csname r@#2@cref\endcsname\relax
  \else\protected@write\@auxout{}{\string\newlabel{#1@cref}%
    {\csname r@#2@cref\endcsname}}%
  \fi
  \@esphack}
\makeatother

\DeclareMathOperator{\Exp}{\mathbb{E}}
\newcommand{\E}[2][{}]{\ensuremath{{\textstyle \Exp_{#1}}\bigl[#2\bigr]}}
\newcommand{\bg}[1]{\medskip\noindent{\it #1}}
 
\newenvironment{proofbox}[1][Proof]{\begin{proof}[#1]}{\end{proof}}

\newenvironment{proofof}[1]{\begin{proof}[Proof of {#1}]}{\end{proof}}

\newcommand{\Ex}{\Exp} %{\ensuremath{\mathbb E}}
\newcommand{\R}{\ensuremath{\mathbb R}} 
\newcommand{\Z}{\ensuremath{\mathbb Z}}
 
\newcommand{\A}{\ensuremath{\mathcal{A}}}
\newcommand{\B}{\ensuremath{\mathcal{B}}}
 
\newcommand{\Ec}{\ensuremath{\mathcal{E}}} 

\newcommand{\I}{\ensuremath{\mathcal I}} 
\newcommand{\J}{\ensuremath{\mathcal J}} 
\newcommand{\K}{\ensuremath{\mathcal K}}

\newcommand{\T}{\ensuremath{\mathcal T}}
\newcommand{\Nc}{\ensuremath{\mathcal N}} 

\newcommand{\Pc}{\ensuremath{\mathcal P}}
 
\newcommand{\Rc}{\ensuremath{\mathcal R}} 
 
\newcommand{\V}{\ensuremath{\mathcal V}} 
\newcommand{\opt}{\ensuremath{\mathsf{opt}}\xspace}
\newcommand{\OPT}{\ensuremath{\mathit{OPT}}}

\newcommand{\argmax}{\operatorname{argmax}}

\newcommand{\sm}{\ensuremath{\setminus}} 
\newcommand{\es}{\ensuremath{\emptyset}}

\newcommand{\ceil}[1]{\ensuremath{\left\lceil#1\right\rceil}}
\newcommand{\floor}[1]{\ensuremath{\left\lfloor#1\right\rfloor}}
\newcommand{\poly}{\operatorname{poly}}

\newcommand{\e}{\ensuremath{\epsilon}} 

\newcommand{\gm}{\ensuremath{\gamma}} 

\newcommand{\sse}{\subseteq}

\newcommand{\vrp}{\ensuremath{\mathsf{VRP}}\xspace}

\newcommand{\csko}{\ensuremath{\mathsf{CorrKO}}\xspace}
\newcommand{\cskocancel}{\ensuremath{\csko\text{-}\mathsf{Cancel}}\xspace}

\newcommand{\tcsko}{\ensuremath{\mathsf{2CorrKO}}\xspace}

\newcommand{\tkdo}{\ensuremath{\mathsf{2KDO}}\xspace}

\newcommand{\cso}{\ensuremath{\mathsf{CorrO}}\xspace}

\newcommand{\knapo}{\ensuremath{\mathsf{KnapOrient}}\xspace}

\newcommand{\into}{\ensuremath{\mathrm{in}}}
\newcommand{\out}{\ensuremath{\mathrm{out}}}

\newcommand{\reg}{\ensuremath{\mathsf{reg}}}

\newcommand{\rt}{\ensuremath{\rho}} %ROOT OF THE GRAPH
\newcommand{\di}{\ensuremath{d}}%distance function
\newcommand{\pathh}{\ensuremath{\mathsf{path}}\xspace}
\newcommand{\parent}{\ensuremath{\mathsf{par}}\xspace}
\newcommand{\rightt}{\ensuremath{\mathsf{rt}}\xspace}
\newcommand{\leftt}{\ensuremath{\mathsf{left}}\xspace}
\newcommand{\level}{\ensuremath{\mathsf{lev}}\xspace}

\newcommand{\rewd}{\ensuremath{\mathsf{rewd}}\xspace}

\newcommand{\Rew}{\ensuremath{\mathcal{R}}} % reward policy
\newcommand{\Rewd}{\ensuremath{\mathsf{Rewd}}\xspace}

\newcommand{\nsize}[1][v]{\ensuremath{S_{{#1}}}}
\newcommand{\nrewd}[1][v]{\ensuremath{R_{{#1}}}}

\newcommand{\Sv}{\nsize} %{\ensuremath{S_v}} %Vertex v
\newcommand{\Rv}{\nrewd} %{\ensuremath{R_v}}
\newcommand{\Sw}{\nsize[w]} %{\ensuremath{S_w}}%Vertex w
\newcommand{\Su}{\nsize[u]} %{\ensuremath{S_u}}%Vertex u
\newcommand{\Ru}{\nrewd[u]} %{\ensuremath{R_u}}

\newcommand{\sonev}{\nsize^{(1)}} %{\ensuremath{S^{(1)}_v}}
\newcommand{\stwov}{\nsize^{(2)}} %{\ensuremath{S^{(2)}_v}}
\newcommand{\sthreev}{\nsize^{(3)}} %{\ensuremath{S^{(3)}_v}}
\newcommand{\sonew}{\nsize[w]^{(1)}} %{\ensuremath{S^{(1)}_w}}
\newcommand{\stwow}{\nsize[w]^{(2)}} %{\ensuremath{S^{(2)}_w}}
\newcommand{\sthreew}{\nsize[w]^{(3)}} %{\ensuremath{S^{(3)}_w}}
\newcommand{\soneu}{\nsize[u]^{(1)}} %{\ensuremath{S^{(1)}_u}}
\newcommand{\stwou}{\nsize[u]^{(2)}} %{\ensuremath{S^{(2)}_u}}
\newcommand{\ronev}{\nrewd^{(1)}} %{\ensuremath{R^{(1)}_v}}
\newcommand{\rtwov}{\nrewd^{(2)}} %{\ensuremath{R^{(2)}_v}}
\newcommand{\rthreev}{\nrewd^{(3)}} %{\ensuremath{R^{(3)}_v}}
\newcommand{\piv}{\ensuremath{\pi_v}}
\newcommand{\Piv}{\ensuremath{\Pi_v}}

\newcommand{\vl}{\ensuremath{v_{\mathsf{last}}}} %{\ensuremath{v_\ell}}%LAST VERTEX VISITED
\newcommand{\Xvj}{\ensuremath{X_v^j}}
\newcommand{\muvj}{\ensuremath{\mu_v^j}}
\newcommand{\muwj}{\ensuremath{\mu_w^j}}

\newcommand{\iv}{\ensuremath{i_v}}
\newcommand{\ld}{\ensuremath{\lambda}} 
\newcommand{\kp}{\ensuremath{\kappa}}
\newcommand{\al}{\ensuremath{\alpha}} 

\newcommand{\tht}{\ensuremath{\theta}}

\newcommand{\sg}{\ensuremath{\sigma}}

\newcommand{\tz}{\ensuremath{\widetilde z}} 
\newcommand{\ts}{\ensuremath{\widetilde s}} 
 
\newcommand{\tgm}{\ensuremath{\widetilde\gm}}

\newcommand{\hpi}{\ensuremath{\widehat\pi}}
\newcommand{\bx}{\ensuremath{\overline x}} 
\newcommand{\by}{\ensuremath{\overline y}}
\newcommand{\bz}{\ensuremath{\overline z}}
\newcommand{\bs}{\ensuremath{\overline s}}

\newcommand{\ptp}{\ensuremath{\mathsf{P2P}}\xspace}
\newcommand{\dead}{\ensuremath{D}}

\newcommand{\lb}{\ensuremath{\mathsf{lb}}}
\newcommand{\LB}{\ensuremath{\mathsf{LB}}}

\newcommand{\UB}{\ensuremath{\mathsf{UB}}}
\newcommand{\bon}{\ensuremath{\mathbbm{1}}}

\newcommand{\diam}{\ensuremath{\Delta}}

\newcommand{\spath}{\ensuremath{Q^*}}
\newcommand{\rpath}{\ensuremath{P}}
\newcommand{\nd}{\ensuremath{\varphi}}
\newcommand{\ndset}{\ensuremath{\mathsf{Por}}\xspace}
\newcommand{\fullset}{\ensuremath{\mathsf{Por}}\xspace}
\newcommand{\nxt}{\ensuremath{\mathsf{next}}\xspace}
\newcommand{\Bad}{\B}
\newcommand{\cBad}{\ensuremath{\Bad^c}}
\newcommand{\Vis}{\V}
\newcommand{\Ret}{\Rc}

\newcommand{\dreg}{\ensuremath{d^{\reg}}}
\newcommand{\dist}{\ensuremath{D}}

\newcommand{\Alg}{\A}
\newcommand{\knap}{\ensuremath{\mathsf{knap}}\xspace}
\newcommand{\optknap}{\OPT^{\knap}}
\newcommand{\wt}{\ensuremath{\mathsf{wt}}}
\newcommand{\dbrack}[1]{\ensuremath{\llbracket{#1}\rrbracket}}
\newcommand{\orientkd}{\ensuremath{\mathsf{OrientKD}}\xspace}

\newcommand{\knapokd}{\ensuremath{\mathsf{Knap}\orientkd}\xspace}

\newcommand{\kdo}{\ensuremath{\mathsf{KnapDO}}\xspace}
\newcommand{\knapd}{\ensuremath{\mathsf{KD}}\xspace}
\newcommand{\knbudg}{\ensuremath{T}}
\newcommand{\knwt}{\ensuremath{b}}
\newcommand{\excep}[1][R_v]{\ensuremath{{#1}^{>W/2}}}
\newcommand{\trunc}[1][R_v]{\ensuremath{{#1}^{\leq W/2}}}
\newcommand{\visit}{\ensuremath{\mathsf{Travel}}\xspace}
\newcommand{\csolp}{\ensuremath{\text{(CO-P)}}\xspace}
\newcommand{\csodual}{\ensuremath{\text{(CO-D)}}\xspace}
\newcommand{\optcsolp}{\ensuremath{\OPT_{\mathsf{CO}\text{-}\mathsf{P}}}\xspace}
\newcommand{\optcso}{\ensuremath{\OPT^{\cso}}\xspace}

\newcommand{\dep}{H}
\newcommand{\tim}{T}

\SetAlgoProcName{Algorithm}{Algorithm}
\SetFuncSty{textsc}
\SetAlCapNameSty{textsc}

\title{Approximation Algorithms for Correlated Knapsack Orienteering}
\author{David Alem\'an Espinosa}{Dept. of Combinatorics and Optimization, Univ. Waterloo,
  Waterloo, ON N2L 3G1.}{dalemane@uwaterloo.ca}{}{Supported in part by NSERC grant 327620-09.} 

\author{Chaitanya Swamy}{Dept. of Combinatorics and Optimization, Univ. Waterloo,
  Waterloo, ON N2L 3G1.}{cswamy@uwaterloo.ca}{https://orcid.org/0000-0003-1108-7941}%
  {Supported in part by NSERC grant 327620-09.} 

\authorrunning{D. Al\'eman Espinosa and C. Swamy}

\Copyright{David Al\'eman Espinosa and Chaitanya Swamy} 



\nolinenumbers %uncomment to disable line numbering

\relatedversion{A preliminary version~\cite{EspinosaS24} appeared in the Proceedings of APPROX 2024.}

\hideLIPIcs  %uncomment to remove references to LIPIcs series (logo, DOI, ...), e.g. when preparing a pre-final version to be uploaded to arXiv or another public repository

\EventEditors{Amit Kumar and Noga Ron-Zewi}
\EventNoEds{2}
\EventLongTitle{Approximation, Randomization, and Combinatorial Optimization. Algorithms and Techniques (APPROX/RANDOM 2024)}
\EventShortTitle{\mbox{\scriptsize APPROX/RANDOM 2024}}
\EventAcronym{APPROX/RANDOM}
\EventYear{2024}
\EventDate{August 28--30, 2024}
\EventLocation{London School of Economics, London, UK}
\EventLogo{}
\SeriesVolume{317}
\ArticleNo{29}
\begin{document}

\maketitle
%\def\thepage{}
%\thispagestyle{empty}

%%%%%%%%%%%%%%%%%%%%%%%%%%%%%%%%%%%%%%%%%%%%%%%%%%%%%%%%%%%%%%%%%%%%%%%%%%%%%%%%%%%%%%%%%%
%%%%%%%%%%%%%%%%%%%%%%%%%%%%%%%%%%%%%%%%%%%%%%%%%%%%%%%%%%%%%%%%%%%%%%%%%%%%%%%%%%%%%%%%%%

%%%%%%%%%%%%%%%%%%%%%%%%%%%%%%%%%%%%%%%%%%%%%%%%%%%%%%%%%%%%%%%%%%%%%%%%%%%%%%%%%%%%%%%%%%

%\begin{center}{\bf Draft - Do Not Distribute}\end{center}

\begin{abstract}
We consider the {\em correlated knapsack orienteering} (\csko) problem: 
%which is the
%two-budget version of the {\em correlated orienteering \cso} problem first considered by
%Gupta et al. \cite{GuptaKNR12}. 
%or alternatively, the correlated knapsack problem with metric switching costs between jobs.
%The input consists of 
we are given a travel budget $B$, processing-time budget $W$, finite metric space
$(V,\di)$ with root $\rt\in V$, where each vertex is associated with a job with possibly
correlated random size and random reward that become known only when the job completes. 
%are possibly correlated random variables.
%Each vertex $v\in V$ consist of a job whose size and reward are, possibly correlated,
%random variables that are drawn from an arbitrary distribution that is given to us as input.
Random variables are independent across different vertices.
%The reward and processing time of a job is only known when the job is fully processed.
The goal is to compute a $\rt$-rooted path %(starting at a given root-vertex) 
of length at most $B$, in a possibly adaptive fashion, 
that maximizes the reward collected from jobs that are processed by time $W$.
%(knapsack) time horizon $W$. 
%The adaptivity gap is defined as the worst-case ratio between the expected rewards of an optimal adaptive policy and an optimal non-adaptive policy.
To our knowledge, \csko has not been considered before, though prior work has
considered the uncorrelated problem, {\em stochastic knapsack orienteering}, and 
{\em correlated orienteering}, which features only one budget constraint on the {\em sum}
of travel-time and processing-times. % are bundled

Gupta et al.~\cite{GuptaKNR12} observed that the {\em uncorrelated} version of this
problem has a constant-factor adaptivity gap.  
We show that, perhaps surprisingly and in stark contrast to the uncorrelated problem, 
%it turns out that 
the {\em adaptivity gap of \csko is not a constant, and is at least
%which turns out to not be the case for the correlated case.
%We prove that the adaptivity gap of \csko is at least
$\Omega\bigl(\max\{\sqrt{\log{B}},\sqrt{\log\log{W}}\}\bigr)$}. 
Complementing this result, we devise {\em non-adaptive} algorithms 
%approximation algorithms that hence also establish adaptivity-gap upper bounds.
%To complement this result we provide a quasi-polynomial time $O(\log\log W)$-approximation
%algorithm for \csko, using similar ideas as those applied by Bansal and Nagarajan
%\cite{BansalN14} to the correlated orienteering problem.  
that obtain: (a) $O(\log\log W)$-approximation in quasi-polytime; and 
(b) $O(\log W)$-approximation in polytime.
This also establishes that the adaptivity gap for \csko is at most $O(\log\log W)$.
We obtain similar guarantees for \csko with cancellations, wherein a job can be cancelled
before its completion time, foregoing its reward. %at the expense of not collecting its reward.  
We show that an $\al$-approximation for \csko implies an $O(\al)$-approximation for \csko
with cancellations.

We also consider the special case of \csko, wherein job sizes are weighted Bernoulli
distributions, and more generally where the distributions are supported on at most two
points (\tcsko). 
%Weighted Bernoulli distributions have often been considered for
%stochastic optimization problems
Although weighted Bernoulli distributions suffice to yield an $\Omega(\sqrt{\log\log B})$
adaptivity-gap lower bound for (uncorrelated) {\em stochastic orienteering}, we show that
they are easy instances for \csko. We develop non-adaptive algorithms that achieve
$O(1)$-approximation in polytime for weighted Bernoulli distributions, and 
%more generally, $O(1)$-approximation 
in $(n+\log B)^{O(\log W)}$-time for the more general case of \tcsko.
(This also implies that our adaptivity-gap lower-bound example, which uses distributions of
support-size $3$, is tight in terms of the support-size of the distributions.)

Finally, we leverage our techniques to provide a quasi-polynomial time $O(\log\log B)$
approximation algorithm for correlated orienteering 
%(where $B$ denotes the budget for the sum of traveling times and job-processing times), 
improving upon the approximation guarantee in \cite{BansalN14}.  
\begin{comment}
We show that the special version of \csko, in which for every $v\in V$, the size of $v$
can take at most 2 different values, 
has a constant-factor adaptivity gap,
and that up to small constants, 
the problem is equivalent to a problem which we call {\em double knapsack deadline
  orienteering \tkdo}.  
The input of \tkdo consists of a travel budget $B$, a knapsack budget $T$, and a metric
$(V,\di)$ in which every vertex $v\in V$ has a reward, a size, a weight and a deadline,
and the goal is to compute a maximum-reward rooted path of length at most $B$, total size
at most $T$, such that for every vertex $v$ along the path the weight of $v$ plus the
weight of its predecessors is at most its deadline.   
We provide a quasi-polynomial time $O(1)$-approximation algorithm, and a polynomial time
$O(\max_{v\in V}q_v)$-approximation algorithm for \tkdo, where $q_v$ denotes the deadline
of vertex $v$.  
We leave as an open problem whether \tkdo admits a constant-factor approximation algorithm
in polynomial time, even in the version of the problem in which all sizes are zero. 
%The latter guarantee translates to a $O(\log{W})$-approximation algorithm for \tcsko
\end{comment}
\end{abstract}

%\newpage
%\pagenumbering{arabic} \normalsize

\section{Introduction} \label{intro}
The {\em orienteering} problem, first introduced by \cite{GoldenLV87}, is a 
fundamental and widely-studied vehicle-routing problem (\vrp{}).
The input to the problem consists of a length/travel bound $B$, finite metric space
$(V,\di)$ representing travel times, root vertex $\rt\in V$, and non-negative rewards
associated with the vertices. %has a non-negative reward.  
The goal is to compute a path rooted at $\rt$ of length at most $B$ 
%on the underlying complete graph with vertex set $V$ in the given metric, 
that collects maximum reward.
Orienteering %belongs to a genre of \vrp{}s that can be dubbed {\em max-reward \vrp{}s},
%wherein one needs to both select
often arises as a subroutine in devising algorithms for other more complex \vrp{}s, both in
approximation algorithms% (see, e.g.,
~\cite{BC+94,FakcharoenpholHR07,ChakrabartyS16,NagarajanR12,FriggstadS14,FriggstadS17,BansalBCM04},
%---e.g., for minimum-latency
%problems~\cite{BlumCCPRS94,FakcharoenpholHR07,ChakrabartyS16,PostS14}, \vrp{}s with
%distance bounds~\cite{NagarajanR12} and regret-bounds~\cite{FriggstadS14,FriggstadS17},
%TSP with time windows~\cite{BansalBCM04}---
as also in computational methods, where it
%corresponds to the pricing problem %that needs to be solved 
arises as the pricing problem when using a branch-cut-and-price method on a
set-covering/configuration LP.  
%as also stochastic \vrp{}s, where part of the input is specified by random variables. 

%Inspired by the seminal paper of Dean et al.~\cite{DeanGV08} on a stochastic version of
%{\em knapsack}, 
%Motivated by the fact that uncertainty is a facet of many decision environments, 
Gupta et al.~\cite{GuptaKNR12} introduced the following {\em stochastic} version of
orienteering %which extends orienteering to a setting 
to model settings where one must spend some uncertain amount of time %that is subject to uncertainty 
%(given by a random variable) 
processing a visited node in order to collect its reward.  
%where some of the data is specified by random variables.
Formally, each vertex corresponds to a job with a random, possibly correlated, processing
time and reward, %that are possibly correlated,
%possibly correlated with each other, 
drawn from a given probability distribution.  
%We use jobs and vertices/nodes interchangeably. 
Random variables corresponding to different vertices are independent. 
%The length of each edge corresponds to the traveling time between its endpoints (in any direction). 
The reward and processing time of a job become known only when the job is fully processed.
The goal is to devise an algorithm, also called {\em policy}, that visits a sequence of
vertices (starting at $\rt$) in a possibly {\em adaptive} fashion that maximizes the
expected total reward collected, subject to the constraint that the total time expended in
traveling and processing jobs is at most $B$. 
%the sum of the travel times and job-processing times being at most $B$.
%A job can only be processed when its corresponding vertex is visited, and 
Jobs cannot be
preempted, and %(i.e., interrupted once started); 
only jobs completed by the time-horizon $B$ yield reward. 
%A job can only be processed while its corresponding vertex is currently being visited,
%and jobs must be processed in a non-preemptive way.  
%If a job is completed after the time horizon $B$ no reward is collected from it.
%The benchmark (as in all of the stochastic problems considered here) is the expected
%reward of the optimal {\em adaptive} policy.  
%A policy may be {\em adaptive} and base its decisions on the realizations of the random
%variables seen thus far, {\em non-adaptive}, wherein, decisions are made
%beforehand given only the distributions of the underlying random variables, and only the
%stopping-point depends on the realization of the the random variables.
This is the {\em correlated orienteering} (\cso) problem. We refer to the special case
where rewards and sizes are independent, simply as {\em stochastic orienteering}; 
due to independence, this is equivalent to the setting with deterministic rewards
%(and random sizes) 
since one can replace the random rewards with their expectations.

A related problem, and the focus of this paper, is {\em correlated knapsack orienteering}
(\csko), wherein there are two separate budgets: $B$ for the (deterministic) travel
time, and $W$ for the total time spent in processing jobs. Again, we refer to the
uncorrelated problem as {\em stochastic knapsack orienteering}. 
%deterministic version, where both rewards and sizes are deterministic is called
Correlated knapsack
orienteering can be motivated from a similar perspective as \cso. Indeed, 
%in a %certain %sense, 
%can be viewed as a kind of fine-grained examination of \corro, 
it is quite natural to decouple the ``apples and oranges'' entities of travel time and 
processing time %(instead of bundling them together) 
when these may represent disparate resources; e.g., travel time may represent
latency of access of jobs in a distributed network, and processing time may present CPU time. 
%in one constraint, and whenever 
%the ``apples-and-oranges'' constraint that clubs together travel time and
%processing time
%could have jobs since travel time and
%processing time are different

In general, %as alluded to above, 
a policy may be {\em adaptive}
%We say that a policy is {\em adaptive} if it 
and choose the next vertex to visit based on the (size, reward) realizations of the
vertices previously visited;
unless otherwise stated, the approximation ratio is always measured relative to 
%$\OPT$,  
the maximum expected reward $\OPT$ %is always the maximum reward 
that can be achieved by an adaptive policy.
%is always the the approximation ratio is always mesured with respect to 
%instantiations of the random variables corresponding to the vertices visited so far.  
On the other hand, a {\em non-adaptive} policy fixes beforehand the sequence of vertices
to visit, and only the stopping-point (when the time-horizon $B$ is exceeded) depends on 
the (size, reward) realizations. While adaptive policies may collect much greater reward,
non-adaptive policies are usually easier to implement, specify, and analyze, by virtue of
the fact that they admit a much-more compact description compared to the
decision tree associated with an adaptive policy, whose description may require space
%easily be
that is exponential in the input size. 
%(Indeed, even for {\em stochastic knapsack}, which
%the special case where all nodes are co-located, certain
Consequently, much work in stochastic optimization has focused on
developing good non-adaptive policies and obtaining bounds on the 
%worst-case loss that
%can occur by moving to non-adaptive policies 
%where the latter is quantified by the 
{\em adaptivity gap}, %of a problem, defined as the supremum, 
which is the supremum, over all problem instances, of
%the ratio 
%(optimal reward achieved by an adaptive policy)
$\OPT$/(maximum reward achieved by a non-adaptive policy); see
e.g.,~\cite{DeanGV08,GuhaM09,EneNS18,GuptaKMR11,GuptaKNR12,ChawlaGTTR20,JiangLLS20}.  

To our knowledge, {\em there is no prior work on \csko}, but there has been work on two 
special cases, stochastic knapsack orienteering (i.e., the uncorrelated version of \csko), 
{\em correlated knapsack} (where all nodes are co-located), and the related problem, \cso.
Our current knowledge for these problems can be %encapsulated 
summarized as follows.
(1) Stochastic knapsack orienteering~\cite{GuptaKNR12} and correlated
knapsack~\cite{GuptaKMR11} both admit non-adaptive $O(1)$-approximation algorithms,
implying $O(1)$ adaptivity gap for these problems.
(2) For stochastic orienteering (i.e., the uncorrelated version of \cso), the adaptivity
gap is $\Omega(\sqrt{\log\log B})$~\cite{BansalN14} and $O(\log\log B)$~\cite{GuptaKNR12};
also, there is a non-adaptive algorithm that achieves an $O(1)$-approximation with respect
to {\em the non-adaptive optimum}~\cite{GuptaKNR12}, 
and hence obtains an $O(\log\log B)$-approximation. 
%the approach leading to the latter
%result also yields an $O(1)$-approximation for {\em stochastic knapsack orienteering}.
(3) For \cso, the adaptivity gap is also $O(\log\log B)$~\cite{BansalN14}, but this is
established non-constructively; the current best algorithmic results obtain
%and we do not have an algorithm yielding a matching approximation ratio. 
$O(\al\log B)$-approximation in polytime~\cite{GuptaKNR12} and
%algorithm, while~\cite{BansalN14} show that the adaptivity gap is in fact $O(\log\log B)$
%(as for stochastic orienteering) and give a non-adaptive quasi-polytime 
$O\bigl(\al\cdot\frac{\log^2\log B}{\log\log\log B}\bigr)$ in
quasi-polytime~\cite{BansalN14}, where $\al$ is the approximation ratio for 
{\em deadline-TSP}.  

%The upshot is that, 
%In particular, we do not know: (a) if \csko is more
%difficult than (uncorrelated) stochastic knapsack orienteering; and 
%(b) how \csko compares in difficulty to \cso.
As is evident from this summary, \cso admits quite different guarantees compared to
stochastic knapsack orienteering and correlated knapsack, and this state of affairs 
prompts the natural question: where does \csko stand in terms of
difficulty relative to these problems? Is it more difficult than the uncorrelated
problem? How does it compare relative to \cso?

%We say that a policy is {\em adaptive}
%The goal is to compute a path (starting at a given root-vertex) of length at most $B$, in a possibly adaptive fashion,
%that maximizes the reward collected from the vertices whose jobs get processed by the (knapsack) time horizon $W$.
% [TO BE COMPLETED]

\subsubsection*{Our contributions}
We initiate a study of correlated knapsack orienteering, and obtain results that, in
particular, shed light on these questions. 
%and obtain the {\em first} results for \csko. 
Our chief contributions are as follows. %threefold.

\begin{enumerate} %[label=$\bullet$, topsep=0.2ex, itemsep=0.1ex, leftmargin=*]
%\begin{itemize}[topsep=0.2ex, itemsep=0.1ex, leftmargin=*]
%\begin{description}[style=sameline, leftmargin=0ex, topsep=1ex, itemsep=2ex, parsep=0ex]
\item %[$\bullet$] 
{\bf Adaptivity gap and approximation algorithms.}
Somewhat surprisingly, and in stark contrast with (uncorrelated) stochastic knapsack
orienteering and correlated knapsack, we prove that 
{\em the adaptivity gap for \csko is not a constant}, showing 
that the correlated problem is strictly harder than the uncorrelated problem.

\begin{theorem}[see Section~\ref{adapt-lbound}] \label{adapgap}
The adaptivity gap for \csko is
$\Omega\bigl(\max\{\sqrt{\log B},\sqrt{\log\log{W}}\}\bigr)$, where $B$ is the travel
budget and $W$ is the processing-time budget. 
\end{theorem}

%$\Omega\bigl(\max\{\sqrt{\log B},\sqrt{\log\log W}\}\bigr)$}
%(Section~\ref{adapt-lbound}). 

%\item[$\bullet$] 
Complementing this lower bound, we develop various {\em non-adaptive}
approximation algorithms for \csko. 
%which therefore also establish upper bounds on the adaptivity gap. 
Our main algorithmic result is a 
{\em quasi-polytime $O(\log\log W)$-approximation algorithm} for \csko,
which thus shows that the adaptivity gap is $O(\log\log W)$. %upper bound on the adaptivity gap. 
%We also devise a polytime $O(\log W)$-approximation algorithm.

\begin{theorem} \label{csko-guarantees}
There are non-adaptive algorithms for \csko with the following guarantees:
\begin{enumerate}[label=(\alph*), ref={\thetheorem\,(\alph*)}, topsep=0.2ex, noitemsep, leftmargin=*]
\item $O(\log\log W)$-approximation in time $(n+\log B)^{O(\log W\log\log W)}$
  (Section~\ref{csko-quasipoly}); \label{theorem-approxcsko} \label{approx-cskothm}
\item $O(\log W)$-approximation in polynomial time (Section~\ref{csko-poly}).
\label{poly-approxcsko}
\end{enumerate}
\end{theorem}

%Finally, as a by-product of our techniques, 
By leveraging the approach leading to Theorem~\ref{approx-cskothm}, 
we also obtain 
%improved guarantees for
%correlated orienteering. We can leverage the approach leading to Theorem~\ref{approx-cskothm} 
%to obtain 
the following guarantee for correlated orienteering, which improves upon the approximation
guarantee in~\cite{BansalN14} (that also runs in quasi-polytime) by an 
$O\bigl(\frac{\log\log B}{\log\log\log B}\bigr)$-factor.  

\begin{theorem}\label{approxcso} \label{csoapprox-thm} \label{cso-approxthm}
Given an $\al$-approximation algorithm for deadline TSP with running time $\tim$, 
we can obtain a non-adaptive $O(\al\log\log{B})$-approximation algorithm for \cso with
running time $(n+\log{B})^{O(\log{B}\log\log{B})}\cdot\tim$. 
Using the algorithm for deadline TSP in~\cite{FriggstadS21}, we obtain an
$O(\log\log B)$-approximation in quasi-polytime.
\end{theorem}

%leveraged to obtain an
%an $O(\al\log\log B)$-approximation algorithm for correlated orienteering
%(\cso) with running time $(n+\log B)^{O(\log B\log\log B)}T$, where $T$ is the running time of a
%given $\al$-approximation algorithm for deadline TSP.  
%also yields an improved approximation guarantee for \cso. 
%This improves upon the approximation guarantee in~\cite{BansalN14},
%which also runs in quasi-polytime, by an 
%$O\bigl(\frac{\log\log B}{\log\log\log B}\bigr)$-factor.  

\item %[$\bullet$] 
{\bf \boldmath \csko with 2-point distributions.}
Our adaptivity-gap lower bound uses distributions of support-size $3$, whereas the
$\Omega(\sqrt{\log\log B})$ adaptivity-gap lower-bound example for stochastic orienteering
in~\cite{BansalN14} considers {\em weighted Bernoulli size distributions}. 
%We take a fine-grained complexity

In Section~\ref{csko-refine}, we investigate \csko from a fine-grained-complexity
perspective to understand this discrepancy.
In contrast with stochastic orienteering, 
we show that when all distributions are supported on at most $2$ points---we call this
\tcsko---% 
{\em the adaptivity gap becomes $O(1)$} (Theorem~\ref{tcsko-nogap}), and we can 
obtain a {\em non-adaptive $O(1)$-approximation in time $(n+\log B)^{O(\log W)}$}
(Corollary~\ref{tcsko-equiv}). Moreover, for 
weighted Bernoulli size distributions, we can obtain a 
{\em polytime non-adaptive $O(1)$-approximation} (Theorem~\ref{bercsko-thm}). 
Our key insight here lies in identifying a novel {\em deterministic} problem, that we call 
{\em orienteering with knapsack deadlines} (\orientkd), which we show is equivalent 
%terms of approximability, up to constant factors 
to \tcsko, up to constant factors. %approximation loss. 
In \orientkd, in addition to orienteering, each vertex $v$ has a weight and
knapsack deadline, and an orienteering-solution $P$ is feasible, if for every $v\in P$,
the total weight of all nodes on $P$ up to (and including) $v$ is at most its
knapsack deadline.
For instance, in a setting where jobs distributed over a network have to be processed on a
single machine, travel times could represent the latency involved in 
%arriving at the location of 
accessing a job, and the knapsack deadlines would capture completion-time deadlines on the
machine.  
\begin{comment}
\orientkd is an interesting problem that can be motivated from the perspective
of both vehicle routing and scheduling.
%capturing aspects of both orienteering and scheduling problems. 
For instance, as a scheduling problem, one consider a setting where we have jobs located
in different parts of a network that are processed on a machine; the travel times 
represent the latency involved in arriving at a job, and the knapsack deadlines correspond
to completion-time deadlines on the machine.
\end{comment}
We obtain %algorithms with 
the above approximation guarantee for \orientkd (Section~\ref{okd-alg}), 
and hence obtain the same guarantee (up to constant factors) for \tcsko. 

Our results for \tcsko thus show that 
%the fact that we utilize distributions of support-size $3$
%is not an artifact of our construction: {\em any} such lower-bound example {\em must}
%involve some distributions of support-size at least $3$; in this sense, 
our adaptivity-gap
example is tight in terms of the support-size of underlying distributions: 
%in that  
{\em any} such lower-bound example {\em must}
involve some distributions of support-size at least $3$. %in this sense,   
%any such example must utilize
%distributions of support-size at least $3$ (and this is not an artifact of our
%construction).

\item %[$\bullet$] 
{\bf \boldmath \csko with cancellations.} 
In this version of the problem (see Section~\ref{csko-cancel}), 
%we consider {\em \csko with cancellations}, wherein 
we may {\em cancel} (i.e., discard) the current vertex $v$ at any time-step
prior to its completion, foregoing its reward, and we are not allowed to return to $v$.
We obtain the
same approximation guarantees for this problem as for \csko: i.e., 
quasi-polytime $O(\log\log W)$-approximation, and polytime $O(\log W)$-approximation. 
En route, we %show that we can 
obtain an $O(1)$-approximation for the special case where
we obtain non-zero rewards only when jobs instantiate to size at most $W/2$
(Theorem~\ref{cskoc-small-thm}).
%\end{description}
%\end{itemize}
\end{enumerate}

\noindent
Our results paint a nuanced picture of the complexity of \csko vis-a-vis 
\cso and stochastic knapsack orienteering. While \csko is harder than stochastic
knapsack orienteering, our algorithmic results suggest that \csko is easier than
correlated orienteering. We obtain similar approximation factors for both problems
in quasi-polytime, but in polytime, we obtain approximation factor of $O(\log W)$ for
\csko, %for \cso, the best polytime approximation factor is 
while the current-best polytime factor for \cso is $O(\log n\log B)$; moreover, with
weighted Bernoulli distributions, 
\csko is provably easier than \cso.
%we obtain a non-adaptive $O(1)$-approximation for \csko,
%but for \cso, the adaptivity gap of $\Omega(\sqrt{\log\log B})$ is realized on such instances.

%\vspace*{-1ex}
\subsubsection*{Technical overview}
We highlight the key %give an overview of the techniques 
technical ideas underlying our results. Let $\OPT$ be the optimal
reward for \csko. Let $\Sv$ denote the random size of vertex $v$. For an integer $j\geq 0$, let
$X^j_v:=\min\{\Sv,2^j\}$ and $\mu^j_v=\E{X^j_v}$. 
%denote the truncated size and expected truncated size respectively. 
The significance of these quantities is
that if $\mu^j(P_{\rt,v}-v)\leq c\cdot 2^j$, where $P$ is a rooted path, $v\in P$,
and $P_{\rt,v}$ is the $\rt\leadsto v$ portion of $P$, then %we can ensure that 
a random subpath $P''$ of $P$ where we retain each $u\in P$ independently with probability 
$\frac{1}{2c}$ %---call this node sampling---
satisfies %has the property that 
$\Pr[\text{$v\in P''$ and starts being processed by time $2^j$}]\geq\frac{1}{4c}$; this indicates
%that the smallest $j$ for which the above 
that $\pi_v(2^j)$, which is the expected reward of $v$ if its processing starts by time
$2^j$, can serve as a good proxy for the expected reward obtained from $v$.

\vspace*{-1ex}
\subparagraph*{Algorithms for \boldmath \csko and \cso.}
%Hence, if we can find a path $P$
Our quasi-polytime $O(\log\log W)$-approximation for \csko builds upon a structural result
for \cso shown by~\cite{BansalN14}. Adapting their result (see
Lemma~\ref{cso-simpstructhm}) to the setting of \csko immediately
yields %They show 
that one can extract a suitable rooted path
$\spath$ from the decision tree representing an optimal adaptive policy, 
%such that for 
and suitable nodes $\nd_{-1}=\rt,\nd_0,\nd_1,\ldots,\nd_k$ on $\spath$, where
$k\leq\log W$, such that (roughly speaking): 
(a) the prefix property $\mu^j(\spath_{\rt,\nd_j}-\nd_j)\leq O(K)\cdot 2^j$ holds for every
$j=0,\ldots,k$, and (b) $\sum_{j=0}^k\sum_{v\in\spath_{\nd_{j-1},\nd_j}}\pi_v(2^j)=\Omega(\OPT)$, where
$K=O(\log\log W)$. %(see Theorem~\ref{cso-simpstructhm}). 
So if we could find this path $\spath$,
then using the sampling idea  
described above, one can easily obtain an $O(K)$-approximation. For \cso, Bansal and
Nagarajan~\cite{BansalN14} ``guess'' the {\em portal nodes} $\nd_0,\ldots,\nd_k$ and write
a configuration LP to find suitable paths between every pair of consecutive portal
nodes. %An $\al$-approximation to this configuration LP can 
They use randomized rounding to round a fractional solution, which %This rounding
incurs a $\frac{\log k}{\log\log k}$-factor violation of the prefix property due to
Chernoff bounds, since 
%the prefix-property for index $j$
for each $j$, $\mu^j(\spath_{\rt,\nd_j}-\nd_j)$ can be written as a sum of %$(j+1)$ 
$O(K)\cdot 2^j$-bounded independent random variables.
When one combines this with the node-sampling step, %overall one 
%Overall, 
one therefore incurs an $O\bigl(K\cdot\frac{\log k}{\log\log k}\bigr)$-factor loss 
%in approximation
relative to the value of the LP solution.  
%(i.e., we have 
%$\mu^j(\spath_{\rt,\nd_j}-\nd_j)\leq K\cdot\frac{\log k}{\log\log k}\cdot 2^j$

For \csko (and \cso), we proceed similarly, but we guess many more portal vertices. 
We split each $\spath_{\nd_{j-1},\nd_j}$ into $O(K)$ segments having
$\mu^j$-weight at most $2^j$, and guess the end-points of all such segments (see
Theorems~\ref{structhm}, \ref{strucdthm}).
We then again set up a configuration LP and use randomized rounding; however, we
can now ensure that {the prefix property holds with $O(1)$ violation}, since we can
decompose $\mu^j(\spath_{\rt,\nd_j}-\nd_j)$ into a {\em sum of $2^j$-bounded random variables}
corresponding to the $\mu^j$-weight of each random segment. Thus, an application of
Chernoff bounds and the union bound only incurs an $O(1)$-factor violation of the prefix
property, since $K=\Omega(\log k)$; therefore, we lose only an $O(K)$-factor compared to
the value of the LP solution. This idea also underlies our improved approximation for
\cso. The only essential difference from \csko %and \cso 
comes from how well we can solve
the corresponding configuration LP; for \csko, we can obtain an $O(1)$-approximation to
the LP-optimum using an $O(1)$-approximation algorithm for knapsack orienteering (see below), 
but for \cso, we obtain an $O(\al)$-approximate LP solution given an
$\al$-approximation for deadline TSP. 

%\medskip
The $O(\log W)$-approximation for \csko proceeds by
relating the problem to {\em knapsack orienteering} (\knapo), which is orienteering with
an additional budget constraint on the total node-weight of the path. 
For each index $j=0,1,\ldots,\log W$, we use the portion of the optimal
adaptive-policy tree corresponding to nodes processed at some point in $[2^j,2^{j+1})$,
%We use the optimal adaptive policy 
to extract a good {\em fractional solution to an 
LP-relaxation \eqref{kolp} for \knapo}, where we exploit the LP-relaxation for
orienteering in~\cite{FriggstadS17}. This translation is quite easy because one can
naturally interpret the LP variables as corresponding to certain probabilities obtained
from an adaptive policy. %We lose a $\log W$-factor because we consider each index
%$j=0,1,\ldots,\log W$ separately, and the portion of the adaptive-policy tree corresponding
%to nodes processed at some point in $[2^j,2^{j+1})$.

%\medskip
We remark that one can combine the LP-relaxations for orienteering~\cite{FriggstadS17} and
the {\em correlated knapsack} problem~\cite{GuptaKMR11}, which is the special
case where all nodes are co-located, to obtain an LP for \csko. However, the
chief impediment in rounding an LP solution is that the rounding algorithms for
orienteering and correlated knapsack may give rise to incompatible orderings.
Rounding the orienteering-portion of the LP solution yields a node sequence, and we
need to stick with a subsequence of this to satisfy the travel-budget
constraint. However, forcing one to consider items in a prescribed order for correlated
knapsack can drastically reduce the reward obtained, because jobs that yield reward only
when they instantiate to
large sizes (i.e., $>W/2$) may need to be processed in a different incompatible order; see 
Appendix~\ref{append-corrknap}. 
This tension is real, as evidenced by our adaptivity-gap lower bound, and seems challenging
to deal with. 
%Instead, we utilize the structural result to infer a suitable path as
%discussed above.

\vspace*{-1ex}
\subparagraph*{\boldmath \tcsko.}
For \tcsko, the chief insight is that the problematic case where we obtain reward only
from large-size instantiations becomes quite structured in two ways. 
(1) There is no adaptivity gap (Theorem~\ref{tcsko-nogap}), since only the path in the
adaptive-policy tree corresponding to small-size (i.e., $\leq W/2$) 
instantiations can yield non-zero reward. (2) Given (1), one can infer that %we obtain
the reward obtained from a vertex $v$ is a function of the total small size of all
vertices visited up to $v$, and the total probability mass of vertices visited up to $v$
(see Claim~\ref{expreward}).
%handle the reward obtained from large-size instantiations by defining an instance of 
This allows one to define an instance of {\em orienteering with knapsack deadlines}
(\orientkd) to capture the stochastic problem.
%where each vertex $v$ has a weight 
%$\wt_v$ and deadline $\knapd_v$, and an orienteering-solution $P$ is feasible if
%$\wt(P_{\rt,v})\leq\knapd_v$ for all $v\in P$

%For \orientkd, we devise an $O(1)$-approximation algorithm in almost quasi-polytime, and a
%polytime $O(\log W)$-approximation. 

\vspace*{-1ex}
\subparagraph*{\boldmath \csko with cancellations.}
The algorithm for \csko with cancellations (\cskocancel) considers two cases. For the
large-size instantiations, it is not hard to argue that cancellations do not help (as with
correlated knapsack~\cite{GuptaKMR11}). For the small-size instantiations, we formulate an
LP by combining the LPs for orienteering~\cite{FriggstadS17} and correlated knapsack with
cancellations~\cite{GuptaKMR11}. %Roughly speaking, 
We show that from an LP solution, one can define a suitable \knapo-instance and extract a
good LP solution to this \knapo-instance. The \knapo-instance is defined in such a way
that feasible solutions to this instance can be mapped to fractional solutions to the
correlated-knapsack LP. So we can first round the solution to obtain an integral
\knapo-solution $Q$, and then utilize the LP-rounding algorithm in~\cite{GuptaKMR11} for
correlated knapsack with cancellations to process vertices, with cancellations, 
{\em in the order they appear on $Q$}. We note here that it is crucial that the algorithm
in~\cite{GuptaKMR11} for small-size instantiations has the flexibility that one can
specify a prescribed order for considering vertices (unlike the case of \csko with
large-size instantiations).

%\vspace*{-1ex}
\subsubsection*{Related work}
As mentioned earlier, {\em orienteering} is a fundamental problem in combinatorial
optimization that finds various applications. 
%(e.g.,
%\cite{BC+94,BansalBCM04,FakcharoenpholHR07,BockGKS11,FriggstadS14,ChakrabartyS16,PostS15}). 
Blum et al.~\cite{BlumCKLMM07} devised the first constant-factor approximation algorithm
for orienteering, and the current best approximation factor is $(2+\e)$ for any
$\e>0$~\cite{ChekuriKP12}. 
%and
%also showed that this problem is \apx-hard. Their approximation-factor was later improved
%by Bansal et al.~\cite{BansalBCM04}, and ultimately by Chekuri et al.~\cite{ChekuriKP12},
%resulting in a $(2+\epsilon)$-approximation for any $\epsilon>0$.  
Friggstad and Swamy~\cite{FriggstadS17} gave the first LP-based $O(1)$-approximation
algorithm. %for the first time. 
Their LP plays an important role for obtaining some of our results.
%We adapt their LP to the \csko setup in order to obtain some of our (LP-relative)
%approximation guarantees. 
% 
%A closely related problem to \kdo is 
{\em Deadline TSP}, also known as {\em deadline orienteering}, is a generalization of
orienteering, where nodes now have deadlines, and a path $P$ is feasible if, for every
$v\in P$, its travel time along $P$ is at most its deadline; the goal is again to compute
a maximum-reward feasible path. 
Both orienteering and
deadline TSP can be considered in the rooted, or {\em point-to-point} (\ptp) setting,
where both the start and end nodes of the path are specified.
%
%In this problem each vertex in the metric has a deadline and a reward, and the goal is to
%compute a (rooted) path of maximum reward that only contains vertices whose distance from
%the root along the path is at most its deadline.  
%For this problem 
Deadline TSP admits a polytime $O(\log n)$-approximation~\cite{BansalBCM04} and an
$O(1)$-approximation in time $n^{O(\log(\text{maximum deadline}))}$~\cite{FriggstadS21}.
%Bansal et al.~\cite{BansalBCM04} provided a $O(\log{n})$-approximation which is still the
%best approximation guarantee for this problem. Recently, 
Friggstad and Swamy~\cite{FriggstadS21} 
%provided a constant-factor approximation algorithm
%with quasipolynomial runtime for both deadline TSP, and 
also consider the more general {\em monotone-reward TSP}, where there are no deadlines but
the %instead of a deadline,
%each vertex $v\in V$ has non-increasing reward function $\rewd_v(.)$, where $\rewd_v(t)$
the reward of a node $v$ having travel time $t$ is given by $\rewd_v(t)$, where $\rewd(.)$
is a non-increasing function. They showed that this problem is essentially equivalent to
deadline TSP.

The literature on stochastic optimization problems is rich, and we discuss below only the
work that is most relevant to our work.
%We only mention related stochastic optimization problems to the ones we consider here.
\begin{itemize}%[nosep, leftmargin=*]
\item {\bf Stochastic knapsack problems.} 
Stochastic orienteering and \csko generalize respectively {\em stochastic knapsack}, which
was studied in the seminal work of~\cite{DeanGV08}, and 
{\em correlated knapsack}~\cite{GuptaKMR11,Ma14}, which correspond to the special case
where all nodes are co-located (i.e., the travel budget is irrelevant). 
%in their seminal paper Dean et al.~\cite{DeanGV08} studied the {\em stochastic knapsack} problem which is equivalent to (uncorrelated) stochastic \knapo in a graph where all vertices are co-located.
%They showed that this problem has a constant (larger than one) adaptivity gap,
%and provided various key insights that were later used in more general problems.
The state-of-the-art for stochastic knapsack is a $(2+\e)$-approximation~\cite{Bhalgat11}.
%and a bicriteria $(1+\e)$-approximate PTAS, where the knapsack constraints are relaxed by
%a $(1+\e)$-factor.
%Their approximation ratio was later improved to $2+\epsilon$ by
%Bhalgat~\cite{Bhalgat11}. Bhalgat et al.~\cite{BhalgatGK11} provided a relaxed PTAS, that
%for any $\epsilon>0$ computes a $(1+\epsilon)$-approximate adaptive policy when the size
%of the knapsack is relaxed to $(1+\epsilon)B$. 
Gupta et al.~\cite{GuptaKMR11} obtained the first constant-factor approximation for
correlated knapsack, and the constant was improved to $(2+\e)$ by Ma~\cite{Ma14}.
%considered the correlated version of this problem, in which
%each item has a random, possibly correlated, size and reward. They provided an LP-based
%non-adaptive constant-factor approximation for both this problem and the version of
%correlated knapsack in which jobs are allowed to be cancelled before their completion
%time; if a job is cancelled no reward is collected from it and the job is abandoned for
%ever. The state of the art for both of these correlated variants is a
%$(2+\epsilon)$-approximation due to Ma~\cite{Ma14}. 
%Desphande et al.~\cite{DesphandeHK16}
%considered the minimization version of stochastic knapsack, 
%known as the {\em stochastic knapsack cover} problem, 
%and provided an adaptive $2$-approximation algorithm. In this problem items have a fixed
%size and a random reward. A solution has to fit items into the knapsack until collecting
%total reward that is at least a certain given threshold and the goal is to minimize the
%size of the knapsack.  

\item{\bf\boldmath Stochastic \vrp{}s.} 
We have already mentioned the works of Gupta et al.~\cite{GuptaKNR12} and~\cite{BansalN14}
that consider (uncorrelated) stochastic orienteering and correlated orienteering.
%Gupta et al.~\cite{GuptaKNR12} were the first to consider the stochastic orienteering
%problem, which generalizes orienteering, and the stochastic knapsack problem considered in
%\cite{DeanGV08,GuptaKMR11}.  
%They provided a non-adaptive $O(\log\log B)$-approximation algorithm for the case in which
%the reward of a vertex is uncorrelated with its size, and a non-adaptive
%$O(\alpha\cdot\log B)$-approximation algorithm for the more general {\em correlated
%  orienteering (\cso)}, where $\alpha$ denotes the best approximation ratio for deadline
%TSP.  
%Bansal and Nagarajan~\cite{BansalN14} proved that the adaptivity gap of stochastic
%orienteering is at least $\Omega(\sqrt{\log\log B})$, even in line metrics for the
%uncorrelated case, and proved that it is at most $O(\log\log B)$ for the more general
%correlated problem. They provided a $O\left(\alpha\cdot\frac{\log^2\log B}{\log\log\log
%  B}\right)$-approximation non-adaptive algorithm with quasi-polynomial running time
%algorithm, where $\alpha$ denotes the approximation ratio for a deadline-TSP algorithm. 
A minimization version of stochastic orienteering, called {\em stochastic $k$-TSP} was considered
by~\cite{EneNS18,JiangLLS20}, where instead of a travel budget, we want to collect a
reward of at least $k$, and seek to minimize the expected travel time. 
%which generalizes both the stochastic knapsack cover problem considered in
%\cite{DesphandeHK16} and {\em k-TSP}, is the {\em stochastic $k$-TSP} problem first
%considered by Ene et al.~\cite{EneNS18}.  
%In here every vertex has a random reward and the goal is to compute a (rooted) tour of
%minimum expected length that collects reward at least $k$. 
Ene et al.~\cite{EneNS18} gave an adaptive $O(\log{k})$-approximation algorithm for this
problem, and Jiang et al.~\cite{JiangLLS20} obtained a non-adaptive $O(1)$-approximation.
%improved their result by providing a non-adaptive $O(1)$-approximation algorithm.
The special case where all nodes are co-located is called {\em stochastic knapsack cover}
for which~\cite{DeshpandeHK16} obtained a $(2+\e)$-approximation.

\item {\bf Multi-armed bandits with metric switching costs.}
A related problem to \csko is the {\em multi-armed bandit} problem with metric switching
costs, considered by Guha and Munagala~\cite{GuhaM09}, which can be viewed as a setting
where each vertex corresponds to a Markov chain (i.e., arm) with known transition
probabilities and rewards. Guha and Munagala consider this setting under a crucial 
{\em  martingale assumption}, which does not hold for \cso or \csko, with separate budgets
for the travel-cost and the number of arm-pulls, as in \csko. In their setting, one can
also abandon a vertex and possibly return to this vertex at a later time.
They devise an $O(1)$-approximation algorithm for this problem
%Their algorithm is somewhat 
that is a hybrid between adaptive and non-adaptive policies: it
non-adaptively specifies the sequence of arms to visit, 
%in the sense that it visits the arms in the order
%implicit by the path obtained from the orienteering instance, and the only adaptive
%decisions that are made is 
but adaptively decides when an arm should be abandoned.  
%
%In this problem it is assumed that each vertex in the metric corresponds to a Markov chain
%(arm), with some known transition probabilities.  
%Each arm starts at some predetermined initial state. Each state of an arm has a reward
%that is collected if the arm reaches that state. An arm evolves to a new state if one
%plays (pulls) that arm, while all other arms retain their present state. One can pull arms
%at most $W$ times in total. 
%If one chooses to pull a different arm than the most recently pulled arm, one must pay the
%travelling cost between them (one may pull an arm again after choosing to abandon it
%previously). One can travel a distance of at most $B$. The objective is to maximize
%expected reward. 
%Guha and Munagala \cite{GuhaM09} showed that under the crucial {\em martingale assumption}
%there is a $O(1)$-approximation algorithm for this problem.  
%An arm satisfies the {\em martingale assumption} if the expected reward that one obtains
%from it after pulling it at any state $u$ is equal to the reward obtained when reaching
%$u$. We emphasize that correlated knapsack (and hence \csko) does not satisfy the strong
%martingale assumption. 
They use an elegant Lagrangian-relaxation idea to reduce the problem to orienteering; this
Lagrangian-relaxation idea was also later used in~\cite{GuptaKNR12}.
%instance at the expense of a constant-factor
%loss.  
%Their algorithm is somewhat non-adaptive in the sense that it visits the arms in the order
%implicit by the path obtained from the orienteering instance, and the only adaptive
%decisions that are made is when an arm should be abandoned.  
%As mentioned earlier, reducing \csko to an instance of \knapo (or orienteering) is highly
%problematic on the other hand (see appendix \ref{append-badexample}). 
\end{itemize}

\section{Preliminaries and notation} \label{prelim}
For an integer $n\geq 0$, we use $[n]$ to denote $\{1,\ldots,n\}$, where $[0]:=\es$, and
$\dbrack{n}$ to denote $\{0\}\cup[n]$. %we define $[0]:=\es$.
For any universe $U$, set $S\sse U$ and element $e\in U$, we sometimes use $S-e$ and $S+e$ to denote
$S\sm\{e\}$ and $S\cup\{e\}$ respectively.

The problems we consider involve a metric space $(V,\di)$ and root $\rt\in V$. 
%with vertex set $V$, and
The metric $\di:V\times V \mapsto \Z_{\geq 0}$ is symmetric and 
%are symmetric integer distances 
captures travel times between vertices; by scaling we may assume that these are integers. 
%We assume that there is a distinctive root vertex $\rt\in V$. 
Let $n=|V|$ and $\diam$ be the diameter of the metric space.
For a set $S$ of edges of the underlying complete graph $(V,E)$, 
we use $\di(S)$ to denote $\sum_{e\in S}\di(e)$. Similarly, for any 
%function $f: V\mapsto\R $ and any set of vertices $U$ we use 
$f\in\R^V$ and $U\sse V$, $f(U)$ denotes $\sum_{v\in U}f_v$. 
%specifying the start node of our path. 
%We often use $\nV$ to denote $\{r\}\cup V$.
%Let $n=|\nV|=|V|+1$.
%Let $\dist_v$ denote $c_{rv}$ for all $v\in V\cup\{r\}$.
%Also, we may assume that $\dist_v\geq 1$ for all $v\in V$ since nodes at distance $0$ from  
%$r$ can always be and visited by their deadlines, and hence can be merged with $r$. 
We say that a path $P$ in $G$ is rooted if it begins at $\rt$. 
We always think of the nodes on a rooted path $P$ as being ordered in increasing order of
their distance from $\rt$ along the path. 
%If $P=v_0v_1\hdots v_\ell$, where $v_0:=\rt$, then we always view the $v_i$'s as being
%ordered in increasing order of $i$.  
For any $u,w\in P$, we say $u\prec_P w$ to denote that $u$ comes before $w$ on $P$, 
%$u=v_i$ and $w=v_j$ for some $0\leq i<j\leq\ell$, and use 
and $u\preceq_P w$ means that $u=w$ or $u\prec_P w$;
%to denote that $u=w$ and $w=v_j$ for some $0\leq i\leq j\leq\ell$.
we omit the subscript $P$ %from $\prec_P$ (or $\preceq_P$) 
when $P$ is clear from the context. %what path $P$ is being considered.  
%
%We extend the relation $\prec_P$ to sets of
%vertices: given two sets $S,T\sse V\cup\{\rt\}$ of vertices on $P$, we say that
%%$S\prec_P T$ if $u\prec_P v$ for every $u\in S, v\in T$, and 
%$S\preceq_P T$ if $u\preceq_P v$ for every $u\in S, v\in T$. Note that 
%%$S\prec_P T$ implies that $S\cap T=\es$, and 
%$S\preceq_P T$ implies that $|S\cap T|\leq 1$ and if $S\cap T=\{w\}$, then $w$ is the last
%node in $S$ and the first node in $T$.
%
We will interchangeably think of a path as an edge-set, or a sequence of nodes; the
meaning will be clear from the context.
For any path $P$ and nodes $a,b\in P$, we use $P_{a,b}$ to denote the $a$-$b$ portion of
$P$. For a path $P$ starting at node $r$, and a node $v\in P$, we define the travel time
of $v$ as $d(P_{r,v})$. 
%and we view the nodes on $P_{a,b}$ as ordered from $a$ to $b$ (i.e., in increasing
%order of their distance from $a$ along $P_{a,b}$).
%We will often need to consider the nodes of a path $P$ excluding one of its end-points:
%for a $u$-$v$ path $P$, which we view as going from $u$ to $v$, define 
%We use $\intp{P}$ to denote the internal nodes of $P$, i.e., the nodes of $P$ excluding its
%end-points. 
\begin{comment}
If $P$ is a $u$-$v$ path, and $a\in P$, we define the 
{\em regret of $P$} to be $\dreg(P):=d(P)-d(u,v)$, and the
{\em two-point regret of $P$ with respect to $a$} as
$\dreg(P,a):=d(P)-d(u,a)-d(a,v)=\dreg(P_{u,a})+\dreg(P_{a,v})$. 
\end{comment}
%We always think of the nodes on a rooted path $P$ as being ordered in increasing order of
%their distance along $P$ from $r$. 
%For any path $P$ and nodes $u,v\in P$, let $P_{uv}$ denote the $u$-$v$ subpath of $P$.  
%For a rooted path $P$ and node $v\in P$, %we use $c_P(v)$ to denote 
%the travel time of $v$ is the distance from $r$ to $v$ along $P$, which we denote by
%$c_P(v):=c(P_{rv})$.  
%the $r$-$v$ subpath of $P$.
%To avoid excessive notation, we will view a path $P$ sometimes as its edge-set, and
%sometimes as its node-set; the meaning will be clear from the context. 

\vspace*{-1ex}
\subparagraph*{Deterministic max-reward vehicle routing.}
The following three vehicle routing problems (\vrp{}s) play a prominent role in the study
of stochastic orienteering. All three problems fall in the genre of max-reward \vrp{}s,
wherein we have nonnegative node rewards $\{\pi_v\}_{v\in V}$, and we need to select some
vertices and find a suitable path visiting these vertices, so as to maximize the reward
obtained. The differences in the problems lie in which paths are allowed, and the
definition of the reward collected by a path.
All these problems below can be considered in the {\em rooted} setting, where we have a root
$\rt$ and the feasible paths form a subset of rooted paths, or in the {\em point-to-point}
(\ptp) setting, where both a start-node $a$ and end-node $b$ are specified, and the
feasible paths are a subset of $a$-$b$ paths.
\begin{itemize}[nosep, leftmargin=*]
\item {\bf Orienteering.} We have a length budget $B$, and feasible paths (in both the
rooted and \ptp versions) are those with length at most $B$; we collect the reward of all 
nodes on a feasible path.
%In {\em rooted orienteering}, the 
%feasible paths are rooted paths of length at most $B$; in {\em \ptp-orienteering}, both
%endpoints of the path are specified, say $a$, $b$, and the feasible paths are $a$-$b$
%paths of length at most $B$.
%The state-of-the-art for \ptp-orienteering is a $(2+\e)$-approximation
%algorithm~\cite{ChekuriKP12}. Friggstad and Swamy~\cite{FriggstadS17} developed
%LP-rounding algorithms for orienteering.

\item {\bf Deadline TSP}, also called {\bf deadline orienteering.} Here nodes have
deadlines $\{\dead_v\}_{v\in V}$. %Again, we can consider the rooted, or \ptp version,
A path $P$ with the appropriate end-points is feasible if the travel time of each
node in $P$ is at most its deadline.
%---i.e., starts at $\rt$ in the
%rooted case, and starts and ends at the given start- and end- nodes $a$, $b$ in the \ptp
%case---and the visiting time of each node  $v\in P$, defined as $d(P_{r,v})$ where $r$ is
%the start node of $P$, is at most its deadline $\dead_v$. 
%$P_{r,v}$ portion
So in the rooted case, a rooted path $P$ is feasible if $d(P_{\rt,v})\leq\dead_v$ for all
$v\in P$; %of each node $v\in P$ is at most its deadline $\dead_v$. Similarly, 
in the \ptp-case, an $a$-$b$ path $P$ is feasible if $d(P_{a,v})\leq\dead_v$ for all
$v\in P$.
%\end{itemize}
We collect the reward of all nodes on a feasible path.
(Equivalently, one can say that the feasible paths are {\em all} paths with the prescribed
end-points, and we collect the reward from all nodes on the path that are
visited {\em by their deadlines}.)

Observe that orienteering is the special case where the deadline of each node
is the length bound $B$.
%It was observed by~\cite{FriggstadS17} that 
Also, the rooted and \ptp versions of deadline TSP are equivalent~\cite{FriggstadS17}. 
%clearly, the \ptp-deadline TSP generalizes \ptp-orienteering.

\item {\bf Monotone-reward TSP.} This is a generalization of deadline TSP, where each
node $v$ has a non-increasing reward-function $\pi_v:\Z_+\mapsto\R_+$, where $\pi_v(t)$
gives the reward obtained from $v$ if $v$ is visited at time $t$.
Every path $P$ with the appropriate end-points is feasible, 
%Every path $P$ that starts at the appropriate node $r$ (i.e., $\rt$, in the rooted case
%and the given start node $a$ in the 
%In the rooted problem, every rooted path $P$ is feasible, 
and the reward of $P$ is given by
$\sum_{v\in P}\pi_v(\text{travel time of $v$})=\sum_{v\in P}\pi_v\bigl(d(P_{r,v})\bigr)$,
where $r$ is the start node of $P$. 
%In the \ptp problem with start and end nodes
%$a$, $b$ respectively, every $a$-$b$ path $P$ is feasible and yields reward 
%$\sum_{v\in P}\pi_v\bigl(d(P_{a,v})\bigr)$. 

Friggstad and Swamy~\cite{FriggstadS17} showed that monotone-reward TSP can be reduced to 
deadline TSP losing a $(1+\e)$-factor, for an $\e>0$.
Monotone-reward TSP will play a key role in the algorithm for correlated orienteering. 
\end{itemize}

\vspace*{-1ex}
\subparagraph*{Stochastic orienteering problems.}
In the {\em correlated knapsack orienteering} (\csko) problem, each vertex $v\in V$ is
associated with an stochastic job with a random processing time or size 
$\Sv\in\Z_{\geq 0}$ and a possibly {\em correlated} random reward 
$\Rv\in\mathbb{R}_{\geq 0}$. 
We use the terms processing time and size interchangeably.
%The reward and processing time for each vertex may be correlated. 
%We assume that 
These random variables are independent across different
vertices, and their distributions are specified in the input. We are given a length or
travel-time budget
$B$, and a processing-time budget $W$. A solution, or policy, for \csko
visits a sequence of (distinct) vertices starting from the root $\rt$, in a possibly
adaptive fashion, without exceeding the travel-time and processing-time budgets.
More precisely, when a vertex $v$ is visited, it's corresponding job is processed
non-preemptively, and we get to know the processing time and reward of the job only upon
its completion; the completion time of job $v$ is the total processing time of all jobs up
to and including $v$.
So if the adaptive policy visits vertices $v_0:=\rt,\ldots,v_\ell=u$ in that order,
then it must be that the total travel-time $\sum_{i=1}^\ell\di(v_{i-1},v_i)$ to get to $u$
is at most $B$, and the total processing time of (the jobs associated with)
$v_1,\ldots,v_{\ell-1}$ is at most $W$. We collect the rewards of 
%(the jobs associated with) vertices
$v_1,\ldots,v_{\ell-1}$, and we collect $u$'s reward if its completion time 
%which is the sum of the processing times of $v_1,\ldots,v_\ell$, 
is at most $W$.
%The job's processing time and reward become known to us only upon its
%completion we collect its reward only if total processing time spent so far (on $v$
%and the jobs processed before $v$) is at most $W$. 
%the job completes by time horizon $W$
%completed job is collected only if the sum of completion times of vertices processed so
%far is at most $W$.  
%%a total travelled distance of $B$ i.e., 
%So if the adaptive policy visits vertices $v_0:=\rt,\ldots,v_\ell$ in that order,
%collecting their rewards, then we have $\sum_{i=1}^{\ell-1} \di(v_{i-1},v_i)\leq B$, and
%$\sum_{i=1}^\ell\nsize[v_i]\leq W$
The goal is to maximize the expected total reward collected. %from the visited vertices. 
%For any vertex $v\in V$ let
%$\piv(t):=\sum_{t^\prime=t}^W\Pr[\Sv=t^\prime]\cdot\Ex[\Rv\,|\,\Sv=t^\prime]$ denote the
%expected reward obtained from vertex $v$ if its starts being processed at timestep
%$t$. Note that $\piv(t)=0$ for any $t>W$. We assume without loss of generality that the
%distance $\di_{\rho v}$ between each vertex $v$ and the root $\rho$ is at most $B$. %We
%abuse notation and denote the optimal adaptive policy and its expected reward by OPT. 
For notational convenience, we also assign a deterministic value of $0$ to the reward and
processing time of $\rt$. 
%to be deterministically $0$.

In the {\em correlated orienteering} (\cso) problem, the setup is almost the
same as in \csko, except that there is only one budget $B$, which is the budget
for the {\em sum} of the travel times and processing times. (That is, we have one
notion of time, which advances due to both travel and the processing of jobs.)
So if an adaptive policy for \cso visits vertices $v_0:=\rt,\ldots,v_\ell=u$ in that
order, 
%$\sum_{i=1}^\ell\di(v_{i-1},v_i+\sum_{i=1}^{\ell}\nsize[v_i]$ is at most $B$.
then we must have $\sum_{i=1}^\ell\di(v_{i-1},v_i)+\sum_{i=1}^{\ell-1}\nsize[v_i]\leq B$;
that is, the completion time of each $v_i$ for $i=1,\ldots,\ell-1$, 
as also the time when we reach $v_\ell$, {\em taking into account both travel time and
processing time}, should be at most $B$. 
We collect rewards from $v_1\ldots,v_{\ell-1}$, and we collect $u$'s reward if 
%it's completion time %{\em taking into account the travel time}, i.e.,
$\sum_{i=1}^\ell\di(v_{i-1},v_i)+\sum_{i=1}^{\ell}\nsize[v_i]\leq B$.
%$\sum_{i=1}^\ell\bigl(\di(v_{i-1},v_i)+\nsize[v_i]\bigr)+\nsize[u]\leq B$
%corresponds to an stochastic job with a random processing time $\Sv\in\Z_{\geq 0}$ and a
%random reward $\Rv\in\mathbb{Z}_{\geq 0}$, that can be correlated. Random variables
%corresponding to different vertices are independent. The difference with \csko is that
%there is only one budget $B$, 
%which can be viewed as a time horizon if we consider $d$ as traveling times between any pair of vertices. 
%A solution to \cso must visit vertices starting from $\rt$, in a possibly adaptive
%fashion. Each job must be executed non-preemtively upon visiting it and we only know its
%size and reward when the job is completed. We only collect reward from the vertices that
%finish their processing times by the time horizon $B$.  

Any adaptive policy for \csko or \cso can be represented by a decision tree $\T$ rooted at
$\rho$, 
whose nodes are labeled by vertices of $V$, and the branches of a node labeled $v\in V$
correspond to the different size and reward instantiations of $v$, with each branch
specifying the next node to visit under the corresponding instantiation.

A {\em nonadaptive policy} (for \csko or \cso) fixes a priori the sequence of vertices to
potentially visit, {\em without looking at the size and reward instantiations}. The 
{\em adaptivity gap} for an instance is the ratio (optimal expected reward
collected by an adaptive policy)/(optimal reward collected by a nonadaptive policy), and
the adaptivity gap for a problem is the supremum over all instances of the adaptivity gap
for the instance.

%\vspace*{-1ex}
%[MENTION REGRET DISTANCES, DOUBLE KNAPSACK ORIENTEERING, P2P DOUBLE KNAP ORIENTEERING,
%PROBLEMS WITH CANCELLATIONS] 

\vspace*{-1ex}
\subparagraph*{Deterministic knapsack-constrained vehicle routing.}
%In solving these 
Algorithms for stochastic orienteering problems 
%will frequently encounter the following
frequently utilize knapsack-constrained variants of deterministic \vrp{}s,
%vehicle-routing problems (\vrp{}s), 
wherein we seek a feasible solution to the \vrp
satisfying an additional knapsack constraint on the total vertex-weight of the path.
%constraint limiting the a knapsack constraint, where We are
More precisely, suppose we have an underlying ``base'' max-reward \vrp, specified
by a collection $\I$ of feasible paths along with nonnegative vertex-rewards
$\{\pi_v\}_{v\in V}$, where the goal is to find a maximum-reward path in $\I$. In the  
{\em knapsack-constrained version of this \vrp}, we
also have a knapsack constraint specified by nonnegative knapsack weights
$\{\wt_v\}_{v\in V}$ and 
knapsack budget $W$, which restricts the set of feasible solutions to
%the set of feasible solutions to the knapsack-constrained problem is 
$\I^\knap:=\{\tau\in\I: \sum_{v\in\tau}\wt_v\leq W\}$;
the goal is to find a maximum-reward path in $\I^\knap$, i.e., a maximum-reward
path in $\I$ satisfying the knapsack constraint.
%For example, 
When the base \vrp is: (i) orienteering, the knapsack-constrained
version is {\em knapsack orienteering} (\knapo); (ii) deadline-TSP, 
%(also called deadline orienteering), 
the knapsack-constrained version is 
{\em knapsack deadline orienteering} (\kdo).  
%These knapsack-constrained problems 
%\knapo and \kdo 
These two problems were considered by~\cite{GuptaKNR12,BansalN14} in the
context of %(uncorrelated and correlated) 
stochastic orienteering.
%who showed that one 
%or deadline-TSP (i.e., orienteering with deadlines). 
We say that the base-\vrp is a rooted-\vrp, if all paths in $\I$ start at the same vertex,
and it is a \ptp-\vrp, if all paths in $\I$ have the same start and end nodes.

We give a general reduction (Theorem~\ref{knapvrp}) showing if the base-\vrp is a
rooted-\vrp or \ptp-\vrp, and satisfies a certain subpath-closure property, 
then an $\al$-approximation for the \vrp can be used as a black-box to obtain an
$(\al+2)$-approximation for the knapsack-constrained \vrp.
Let $\tau$ be a path with ends $a,b\in V$, which we will view as a sequence of nodes.
By a {\em \ptp-subpath} of $\tau$, we mean any
$a$-$b$ path whose node-sequence is a subsequence of $\tau$; 
%the node-sequence of $\tau$; 
by a {\em rooted-subpath} of $\tau$, we mean any path starting at $a$ whose node-sequence
is a subsequence of $\tau$.
(Note that any subsequence of $\tau$ yields a path, since we are working with a complete
graph.) 
The {\em subpath-closure} property requires that for every path $\tau\in\I$: 
(a) for rooted-\vrp, every rooted-subpath $\tau'$ of $\tau$ is also in $\I$,  
(b) for \ptp-\vrp, every \ptp-subpath $\tau'$ of $\tau$ is also in $\I$. 
%If the rewards in the base-\vrp are not static but depend on the path (as in the case of
%monotone-reward TSP), then we also require that the reward obtained from $\tau'$ is at
%least the reward obtained from $\tau$. 
%with (say)
%ends $a,b\in V$, every subsequence of the node-sequence of $\tau$ that starts and ends at
%$a$ and $b$ respectively, 
%also yields a path in $\I$. Note any such subsequence yields an $a$-$b$ path (since we have a
%complete graph); the requirement is that this path should lie in $\I$.
%We say that $\tau'$ is a grounded-subpath of an $a$-$b$ path $\tau$, if $\tau'$ is also  
%We say that the base problem has the {\em subpath-closure} property if the following holds:
%if $\tau\in\I$ is given by the node-sequence $v_1,\ldots,v_\ell$, then any subsequence of
%$v_1,\ldots,v_\ell$ containing $v_1$
% $\tau'$ of $\tau$ between the end-points of $\tau$ formed
%by taking 
%is also in $\I$. 
%We will also assume that the membership problem for $\I$---decide if a given path lies in
%$\I$---can be efficiently solved.
Most max-reward \vrp{}s---e.g., orienteering, deadline TSP%
%monotone-reward TSP~\cite{FriggstadS21}% 
%(reward of a node is a non-increasing function of its visiting time),
%which are \vrp{}s that can be considered in both the rooted and \ptp settings
---satisfy %these properties.  
the subpath-closure property. 
%We show that an $\al$-approximation for a base \vrp having the subpath-closure property
%can be used in a black-box fashion to obtain an $(\al+2)$-approximation for the
%knapsack-constrained version of the base \vrp.
%Orienteering and deadline TSP can be considered in both the rooted and \ptp settings. 
(Also, note that if a \vrp satisfies the subpath-closure property, then so does the
knapsack-constrained \vrp.)

The above reduction is based on a Lagrangian-relaxation idea that was also used
by~\cite{GuptaKNR12}, specifically to obtain approximation algorithms for \knapo and \kdo. 
%to obtain approximation algorithms for these problems, 
However, their approach results in a constant-factor blowup in the approximation ratio
(factor $2$ for \knapo, and factor $4$ for \kdo%
\footnote{\cite{GuptaKNR12} do not explicitly state a result for \kdo,
and instead embed this result within their algorithm for correlated
orienteering. We can infer this factor by tracing through their algorithm and analysis.}  
%one can infer an approximation factor of $4\al$ given an $\al$-approximation for deadline
%TSP.}% 
when going from the \vrp to the
knapsack-constrained \vrp; our general reduction yields a better factor, in a somewhat 
simpler fashion. %and more-direct fashion.

\begin{theorem} \label{knapvrp}
Consider a max-reward rooted-\vrp or \ptp-\vrp, specified by a set $\I$ of
feasible solutions satisfying the subpath-closure property.
%and the membership problem can be efficiently solved.
%Suppose that the membership problem for $\I$---decide if a given 
For any $\e>0$, an $\al$-approximation algorithm $\Alg$ (where $\al\geq 1$)
for the \vrp can be used to obtain an $(\al+2)(1+\e)$-approximation for the
knapsack-constrained \vrp %The resulting algorithm makes 
by making $O\bigl(\frac{\log n}{\e}\bigr)$ calls to $\Alg$. 
\end{theorem}

\begin{proof} %{Theorem~\ref{knapvrp}}
Let $\{\pi_v\geq 0\}_{v\in v}$ be the node rewards.
Let the knapsack constraint be specified by nonnegative weights $\{\wt_v\}_{v\in V}$ and
knapsack budget $W$. Let $\I^\knap\sse\I$ be the set of feasible solutions to the
knapsack-constrained \vrp.
%We discuss the \ptp case here. The arguments for rooted-\vrp are quite similar, and we
%defer this to Appendix~\ref{append-prelim}.
By a subpath of a path $\tau\in\I$, we mean a rooted-subpath for a rooted-\vrp,
and a \ptp-subpath for a \ptp-\vrp.
We call the common start node of paths in $\I$, a {\em terminal} node; for a
\ptp-\vrp, the common end node of paths in $\I$ is also designated a terminal node.
%For a rooted-\vrp, we call the start node of paths in $\I$ a {\em terminal} node

%Let $V'$ be the set of nodes that appear on some path in $\I^\knap$.
%Observe that 
%If \vrp is rooted with start node $a$, let
%$V'=\{v\in V: \wt_a+\wt_v\leq W\}$
%\text{path $a,v$ is in $\I^\knap$}\}$, 
%and if we have a \ptp-\vrp with start and end nodes 
%Let $a$, $b$ be the start and end-nodes respectively of the \ptp-\vrp instance.
We may assume that: 
%
%\begin{enumerate}[label=(\roman*), nosep, leftmargin=*]
%\item 
\begin{comment}
(i) $\pi_w=0$ for every terminal node $w$, 
%$\pi_a=\pi_b=0$, 
since, for any $c\geq 1$, a
$c$-approximation for the knapsack-constrained problem when we set the rewards of terminal
node(s) %$a$ and $b$ 
to $0$, implies a $c$-approximation for the problem with the original rewards; %and
\end{comment}
%\item 
(i) $\wt_w=0$ for every terminal node $w$,
%$\wt_a=\wt_b=0$, 
since we can always set the knapsack weights of %$a$, $b$ 
terminal node(s) to $0$ and work with the residual knapsack budget of 
$W-$(total weight of terminal nodes)  
%the knapsack budget to $W-\wt_a-\wt_b$ 
without affecting the set $\I^\knap$ of feasible solutions to the knapsack-constrained
problem; and 
%\item givaen (ii), 
(ii) $\wt_v\leq W$ for every node $v$, as otherwise we can simply discard node $v$. 
%Let $V'=\{v\in V: \wt_a+\wt_v+\wt_b\leq W\}$. Clearly, only nodes in $V'$ can belong to
%paths in $\I^\knap$, and $V'$ can be efficiently identified. We can always preprocess and  
%discard nodes not in $V'$, so to keep notation simple, in the sequel, we assume that
%$V'=V$. 
%\end{enumerate}

%Let $\optknap$ denote the optimal value 
Let $\tau^*$ be an optimal solution for the knapsack-constrained \vrp, and
$\optknap=\pi(\tau^*)$ be the optimal value.
Due to (ii), we have $\max_{v\in V}\pi_v\leq\optknap\leq\sum_{v\in V}\pi_v$, so by
considering all powers of $(1+\e)$, we can assume that we have an estimate $\LB$ such that 
$\LB\leq\optknap\leq(1+\e)\LB$. 
More precisely, there are at most $O\bigl(\frac{\log n}{\e})$ powers of $(1+\e)$ to
consider; we run the procedure below for each of these values, and return the
best solution obtained.

Set $\ld=\frac{2}{\al+2}$. Essentially, we run $\Alg$ on the base-\vrp instance with
node-rewards $\{\pi_v-\frac{\ld\cdot\LB}{W}\cdot\wt_v\}_{v\in V}$ and return a suitable
subpath of the path returned by $\Alg$.
%to obtain a path $\tau\in\I$, and then return a suitable subpath of $\tau$.
Set $\pi'_v=\max\bigl\{0,\pi_v-\frac{\ld\cdot\LB}{W}\cdot\wt_v\bigr\}$ for all
$v\in V$, %(note that $\pi'_w=0$ for a terminal node $w$),
%and $\pi'_a=\pi'_b=0$, 
and run $\Alg$ on the base-\vrp instance with $\{\pi'_v\}$ node rewards to obtain a path
$\tau$. 
Let $\tau'$ be the subpath of $\tau$ whose non-terminal nodes %are given by the subsequence
%comprising 
are all the non-terminal nodes $v\in\tau$ for which $\pi'_v>0$. By the subpath-closure
property, $\tau'\in\I$. 
\begin{enumerate}[label=\arabic*., topsep=0.2ex, noitemsep, leftmargin=*]
\item If there is some $v\in\tau$ for which $\pi_v\geq\frac{\LB}{\al+2}$, then we return
%the path $a,v$, if \vrp is a rooted instance, or 
the path whose only potential non-trivial node is $v$. More precisely, if $v$ is a
terminal, then we return the trivial path consisting of only the terminals;  
otherwise, we return the path having $v$ as the only non-terminal node (so $a,v$ for a
rooted-\vrp, and $a,v,b$ for a \ptp-\vrp).
%the path $a,v,b$. %if \vrp is a \ptp-instance. 
By the subpath-closure property, the path returned lies in $\I$, and
by (ii), it satisfies the knapsack constraint.
%our preprocessing step, we know that the path returned satisfies the knapsack
%constraint; 
So the path returned is feasible for the knapsack-constrained \vrp, and
clearly it achieves reward at least $\frac{\optknap}{(\al+2)(1+\e)}$.

\item Otherwise, we find a subpath $\tau''$ of $\tau'$ whose 
nodes form a maximal subsequence of $\tau'$ satisfying
$\pi(\tau'')\leq\ld\cdot\LB$. We return $\tau''$. By the subpath-closure property,
$\tau''\in\I$. Also, since $\wt_v\leq\pi_v\cdot\frac{W}{\ld\cdot\LB}$ for every
$v\in\tau''$ (note that $\wt_v=0$ if $v$ is a terminal), we have 
$\wt(\tau'')\leq\pi(\tau'')\cdot\frac{W}{\ld\cdot\LB}\leq W$.
%and we argue below that $\wt(\tau'')\leq W$ 
\end{enumerate}

We lower bound the reward of $\tau''$ in step 2 above.
First, observe that the optimal $\pi'$-reward for the base-\vrp instance is at least
$\frac{\al}{\al+2}\cdot\optknap$, since %$\tau^*$ yields $\pi'$-reward at least
$\pi'(\tau^*)\geq\pi(\tau^*)-\frac{2\cdot\LB}{(\al+2)W}\cdot\wt(\tau^*)\geq\optknap-\frac{2\cdot\optknap}{\al+2}$,
as $\wt(\tau^*)\leq W$ and $\LB\leq\optknap$.
Therefore, since $\Alg$ is an $\al$-approximation algorithm for the base-\vrp, we obtain
that $\pi'(\tau')=\pi'(\tau)\geq\frac{\optknap}{\al+2}$.
If $\tau''\neq\tau'$, then since $\tau''$ is maximal, we have 
$\pi(\tau'')\geq\ld\cdot\LB-\frac{\LB}{\al+2}$, because $\pi_v<\frac{\LB}{\al+2}$ for
every $v\in\tau$. It follows that
%
%\begin{equation*}
$\pi(\tau'')\geq\min\bigl\{\pi(\tau'),\tfrac{\LB}{\al+2}\bigr\}
\geq\min\bigl\{\pi'(\tau'),\tfrac{\LB}{\al+2}\bigr\}\geq\frac{\optknap}{(\al+2)(1+\e)}$.
%\end{equation*}
%
\end{proof}

\begin{corollary} \label{knapvrp-approx}
%Let $n=|V|$ and $\diam$ be the diameter of the metric space $(V,d)$. 
There are algorithms with the following guarantees.
\begin{enumerate}[label=(\alph*), nosep, leftmargin=*]
%(a) 
\item $(4+\e)$-approximation, for any $\e>0$, for rooted- and \ptp- knapsack orienteering;
%(b) 
\item $O(\log n)$-approximation for the rooted and \ptp versions of knapsack deadline orienteering, and 
knapsack monotone-reward TSP;
%(c) 
\item $O(1)$-approximation in $O\bigl(n^{\log n\diam}\bigr)$ time, for the rooted and \ptp
versions of knapsack deadline orienteering, and knapsack monotone-reward TSP.
\end{enumerate}
\end{corollary}

\begin{proof} %{Corollary~\ref{knapvrp-approx}}
Part (a) follows because there is a $(2+\e)$-approximation~\cite{ChekuriKP12} for
\ptp-orienteering. 

For knapsack deadline orienteering, the guarantees in parts (b) and (c) follow because 
%Part (b) follows because 
the rooted and \ptp versions of deadline orienteering are equivalent~\cite{FriggstadS21}
and they admit: an $O(\log n)$-approximation algorithm, and 
an $O(1)$-approximation in $O\bigl(n^{\log n\diam}\bigr)$ time~\cite{FriggstadS21}.

For knapsack monotone-reward TSP, we cannot apply Theorem~\ref{knapvrp} directly because
the reward of a node depends on its travel time. However, the key idea underlying
Theorem~\ref{knapvrp} is that if we take a subpath $\tau'$ of a feasible path $\tau$, then
the nodes on $\tau'$ yield at least as much reward for $\tau'$ as for the original
path. This property clearly holds for monotone-reward TSP, and can be used to show, with
cosmetic changes to the proof of Theorem~\ref{knapvrp}, that knapsack monotone-reward TSP
can be reduced to monotone-reward TSP incurring the same approximation-factor loss as in
Theorem~\ref{knapvrp}.  
Friggstad and Swamy~\cite{FriggstadS21} showed that monotone-reward TSP can be reduced to
deadline orienteering losing a $(1+\e)$-factor. Coupling these two reductions, and
utilizing the aforementioned guarantees for deadline orienteering, yields the guarantees
in parts (b) and (c) for knapsack monotone-reward TSP.

We now briefly describe how to adapt Theorem~\ref{knapvrp} to monotone-reward TSP. Recall
that the reward of a node $v$ is now specified by a non-increasing function of its travel
time, $\pi_v:\Z_+\mapsto\R_+$. We may again assume that
%(i) $\pi_w$ is identically $0$, and $\wt_w=0$, for every terminal node $w$; and
$\wt_v\leq W$ for every node $v$, and is $0$ if $v$ is a terminal.
Let $a$ denote the common start node of all feasible paths.
Let $\tau^*$ be an optimal solution for knapsack monotone-reward TSP, and
$\optknap$ %=\pi(\tau^*):=\sum_{v\in\tau^*}\pi_v\bigl(d(\tau^*_{a,v})\bigr)$ 
be the optimal value. We now have 
$\max_{v\in V}\pi_v\bigl(d(a,v)\bigr)\leq\optknap\leq\sum_{v\in V}\pi_v\bigl(d(a,v)\bigr)$, so 
again we may assume that we have an estimate $\LB$ such that 
$\LB\leq\optknap\leq(1+\e)\LB$. 
Set $\ld=\frac{2}{\al+2}$. 
For every $v\in V$, consider the reward function given by
$\pi'_v(t)=\max\bigl\{0,\pi_v(t)-\frac{\ld\cdot\LB}{W}\cdot\wt_v\bigr\}$ for all $t\in\Z_+$.
We run $\Alg$ on the monotone-reward TSP instance with $\{\pi'_v\}$ node-reward functions
to obtain a path $\tau$. For $v\in\tau$, define
$\rewd_v:=\pi_v\bigl(d(\tau_{a,v})\bigr)$, which is the $\pi_v$-reward that $v$ obtains
under $\tau$. 
Let $\tau'$ be the subpath of $\tau$ whose non-terminal nodes 
are all the non-terminal nodes $v\in\tau$ for which
$\rewd_v>\frac{\ld\cdot\LB}{W}\cdot\wt_v$. 
%that obtain strictly positive reward on $\tau$. 
%
\begin{enumerate}[label=\arabic*., topsep=0.2ex, noitemsep, leftmargin=*]
\item If there is some $v\in\tau$ for which $\rewd_v\geq\frac{\LB}{\al+2}$, then we return
%the path $a,v$, if \vrp is a rooted instance, or 
the path whose only potential non-trivial node is $v$. 
This is a feasible path satisfying the knapsack constraint, and clearly, it
achieves reward at least 
$\pi_v\bigl(d(a,v)\bigr)\geq\rewd_v\geq\frac{\optknap}{(\al+2)(1+\e)}$.

\item Otherwise, we find a subpath $\tau''$ of $\tau'$ whose 
nodes form a maximal subsequence of $\tau'$ satisfying
$\rewd(\tau'')\leq\ld\cdot\LB$. We return $\tau''$. 
Since $\wt_v\leq\rewd_v\cdot\frac{W}{\ld\cdot\LB}$ for every
$v\in\tau''$ (note that $\wt_v=0$ if $v$ is a terminal), we have 
$\wt(\tau'')\leq W$.

The reward obtained by each $v\in\tau''$ is at least $\rewd_v$, so the total reward
obtained in this case is at least
$\rewd(\tau'')\geq\min\bigl\{\rewd(\tau'),\ld\cdot\LB-\frac{\LB}{\al+2}\bigr\}$. 
Observe that $\rewd(\tau')$ is at least the total reward obtained from
$\tau$ for the monotone-reward TSP instance with $\{\pi'_v\}$ node-reward functions. 
The optimal value for this monotone-reward TSP instance is at least
$\optknap-\ld\cdot\LB$, since the reward obtained from $\tau^*$ is at least this
value. Therefore, since $\Alg$ is an $\al$-approximation algorithm for monotone-reward
TSP, we have $\rewd(\tau')\geq\frac{\optknap-\ld\cdot\LB}{\al}\geq\frac{\optknap}{\al+2}$.
\qedhere
\end{enumerate}
\end{proof}

\vspace*{-2ex}
\subparagraph*{LP-relative guarantee for \boldmath \knapo.}
For rooted \knapo, %it will be useful to formulate the LP-relaxation for the problem and
we can utilize Theorem~\ref{knapvrp} to obtain an LP-relative approximation
guarantee. This will be useful in devising algorithms for \csko. 
Consider the following LP-relaxation for rooted \knapo along the lines of an 
LP-relaxation for rooted orienteering in~\cite{FriggstadS17}. 
Let $\rt$ be the root node for the \knapo instance. We bidirect the edges of the complete
graph on $V$ to obtain the arc-set $A$.

\begin{alignat}{3}
\max && \quad \sum_{u,v \in V}z^v_u&\cdot \pi_u \tag{KO-LP}
\label{kolp} \\
\text{s.t.}  && \quad
%z^v_u = 
x^v\bigl(\delta^{\into}(u)\bigr) & \geq x^v\bigl(\delta^{\out}(u)\bigr) \qquad 
&& \forall u\in V-\rt,\,v\in V \tag{O1} \label{pref-visit} \\
%&& x^v(\delta^{in}(u))&=z^v_u \qquad&&\forall u,v\in V\tag{OR2}\label{oc2}\\
&& x^v\bigl(\delta^{\into}(S)\bigr) & \geq z^v_u \qquad && 
\forall v\in V,\, S\subseteq V-\rt,\, u\in S \tag{O2} \label{subtour} \\
&& z^v_u & = 0 && \forall u,v\in V: d_{\rt,u}>d_{\rt,v} \tag{O3} \label{dbnd} \\
&& \sum_{a\in A}d_a\cdot x_a^v & \leq Bz_{v}^v, \quad 
x^{v}\bigl(\delta^{\out}(\rt)\bigr) = z^v_v\qquad && \forall v\in V \tag{O4}
\label{dbudget} \\
%&& x^{v}\bigl(\delta^{\out}(\rho)\bigr) & = z^v_v &&\forall v\in V\tag{OR6}\label{oc6}\\
&& \sum_{v\in V}z_v^{v} & = 1 \tag{O5} \label{unit} \\
&& \quad\sum_{u,v\in V}z_u^v\cdot\wt_u &\leq W && \tag{KN} \label{knbudget} \\
&& x,z &\geq 0. && \notag
\end{alignat}
\clonelabel{lp:ckosmallsizes}{kolp}\clonelabel{knapo-lp}{kolp}%
%\clonelabel{oc1}{pref-visit}
%\clonelabel{oc3}{subtour}
%\clonelabel{oc4}{dbnd}
%\clonelabel{oc5}{dbudget} 
\clonelabel{outflow}{dbudget}%
%\clonelabel{oc7}{unit}
\clonelabel{kn}{knbudget}%
\newcommand{\optkolp}{\ensuremath{\OPT_{\text{\textsf{\ref{kolp}}}}}}%
The $x^v_a$ and $z^v_u$ variables encode the arcs included, and vertices visited,
respectively by the \knapo-path, provided that $v$ is the node visited that is
furthest from $\rt$, i.e., $v$ maximizes $d(\rt,u)$ among all nodes $u$ on the path:
%visited by the path: 
constraints \eqref{dbnd} enforce this semantics;
in an integer solution, these variables will be $0$ if $v$ is not the furthest visited
node from $\rt$. Constraints \eqref{pref-visit} and 
\eqref{subtour} encode that the $\rt\leadsto u$-connectivity is $z^v_u$, and together with
\eqref{outflow} encode that $\{x^v_a\}$ is a $\rt$-preflow of value $z^v_v$ satisfying
the length budget.
Constraint \eqref{unit} enforces that overall $x$ is a $\rt$-preflow of value $1$.
Constraints \eqref{pref-visit}--\eqref{unit} are from the LP for rooted orienteering
in~\cite{FriggstadS17};  
\eqref{knbudget} is the new constraint encoding the knapsack budget.

%The reduction in Theorem~\ref{knapvrp} can be adapted to show that an (R-O)-relative
%approximation algorithm $\Alg$ 
%yield guarantees for \knapo with respect to $\optkolp$. 
%
%show that if $\Alg$ achieves an (R-O)- an $\al$-approximation with respect to the optimal
%value of (R-O)) 
%also be applied given an LP-relative appoximation algorithm for
%orienteering, to obtain an LP-relative approximation guarantee
%yields 
%approximation algorithms
%performed with respect to LP-relaxations for orienteering and \knapo

\begin{theorem} \label{knapo-round} \label{knapo-roundthm}
%Let $\Alg$ be an (R-O)-relative $\al$-approximation algorithm for rooted orienteering,
%i.e., it always returns reward at least $\OPT_{\text{(R-O)}}/\al$. 
%We can use $\Alg$ to obtain a \knapo-solution obtaining reward at least
%$\optkolp/(\al+2)$. 
We can obtain a \knapo-solution that obtains reward at least $\optkolp/5$.
\end{theorem}

\begin{proof}
%{Theorem~\ref{knapo-round}: LP-relative approximation algorithm for \boldmath \knapo} 
Let (R-O) be the LP with constraints \eqref{pref-visit}--\eqref{unit}, %are from the 
which is the LP for rooted orienteering formulated in~\cite{FriggstadS17}. 
Let $\Alg$ be the (R-O)-relative $3$-approximation algorithm for rooted orienteering
devised by~\cite{FriggstadS17}; 
%whose guarantee is relative to the optimum of the LP (R-O) given by
%constraints \eqref{pref-visit}--\eqref{unit}. 
that is, $\Alg$ returns a
rooted-orienteering solution that obtains reward at least $\OPT_{\text{\textsf{R-O}}}/3$.
We adapt the reduction in Theorem~\ref{knapvrp} to utilize $\Alg$.

We set $\LB=\optkolp$ in the reduction, and of course take $\al=3$. For the analysis, we
only need to argue that $\pi(\tau')\geq\optkolp/5$, where recall that $\tau$ is the path
returned by $\Alg$, and $\tau'$ is the rooted subpath of $\tau$ whose non-root nodes
are the non-root nodes $u\in\tau$ with $\pi_u>\frac{2\LB}{5W}\cdot\wt_u$. This follows because
if $(\bx,\bz)$ is an optimal solution to \eqref{kolp}, then due to the LP-relative
guarantee of $\Alg$, we have
\begin{equation*}
\pi'(\tau')=\pi'(\tau)\geq\frac{\sum_{u,v}\bz^v_u\pi'_u}{3}
\geq\frac{\optkolp-\frac{2\LB}{5W}\cdot\sum_{u,v}\bz^v_u\cdot\wt_u}{3}
\geq\frac{\optkolp}{5}
\end{equation*}
where the last inequality follows from \eqref{knbudget}.
\end{proof}

\section{An adaptivity-gap lower bound for \boldmath \csko} \label{adapt-lbound}
%In this section, 
We now show that the adaptivity gap for \csko is
$\Omega\bigl(\max\{\sqrt{\log B},\sqrt{\log\log W}\}\bigr)$, thereby proving Theorem~\ref{adapgap}. 
%Inspired by the bad example for adversarial orderings in the correlated knapsack problem,
We consider the following instance of {correlated knapsack orienteering} that has a 
similar spirit as the adaptivity-gap example in~\cite{BansalN14} for (uncorrelated)
stochastic orienteering.  
%construction used to prove a lower bound for the adaptivity gap for
%the stochastic orienteering problem by Bansal and Nagarajan \cite{BansalN14}.  
%We consider a complete graph $G=(V,E)$ with a metric distance $c:E\mapsto \mathbb{Z}_+$
%induced by 
The metric is a tree-metric induced by a complete binary tree $T$ on a vertex set $V$,
with root $r\in V$ and $\dep\geq 4$ levels, where the distances decrease geometrically as we
move away from $r$. To conform to our notation, we include a separate dummy node $\rt$
that serves as the root for \csko, with distance $0$ to $r$; but when we say root below,
we always mean the root $r$ of the tree $T$.
%where $\dep\geq 9$ is an integer. The distances 
The knapsack budget is $W:=2^{2^{\dep+1}}$ and the length/travel budget is $B:=2^{\dep-1}-1$. 
For a node $v\in V$, we use: $\level(v)$ to denote the level of $v$, $\pathh(v)$ to denote
the unique $r\leadsto v$-path in $T$, and $\parent(v)$ to denote the parent of $v$ if 
$v\neq r$. The root $r$ is at level $\dep$ and each leaf node is at level $1$; for a non-leaf
node $v$ at level $\ell$, the distance between $v$ and its children is $2^{\ell-2}$.
%For any $v\in V\setminus\{r\}$, we use
%$\parent(v)$ to denote the parent of $v$ i.e., the unique neighbor of $v$ in $\pathh(v)$.  
%Each non-leaf node has a left-child and a right-child. 
%We use $\level(v)$ to denote the level of $v$ in $T$, and we use $\pathh(v)$ to denote the
%unique path from the $r$ to $v$ in $T$. For any $v\in V\setminus\{r\}$, we use
%$\parent(v)$ to denote the parent of $v$ i.e., the unique neighbor of $v$ in
%$\pathh(v)$. We assume that $\level(r)=\dep$ and in general that
%$\level(v)=\level(\parent(v))-1$ for any $v\in V\setminus\{r\}$. Note that $\level(v)=1$
%for any leaf $v$ of $T$. 
For a rooted path $P$ in $T$ we say that a node $v\in P$ is a
right-branching (resp. left-branching) node if the node succeeding $v$ on $P$ is its
right-child (resp. left-child). We denote by $\rightt(P)$ and $\leftt(P)$, the right-branching nodes
and left-branching nodes of $P$, respectively. For notational convenience, we assume that
the end-node of $P$ other than $r$ %endpoint of the  path, distinct than $r$, belongs in 
is a left-branching node; 
%$\leftt(P)$. 
if $P=r$, then we say that $r\in \leftt(P)$. 
%The distances and probability distributions for the nodes are defined as
%follows: 
The (correlated) (size, reward) distribution %are as follows:
of node $v$ is supported on three points: 
\begin{equation*}
\begin{split}
\bigl(\sthreev,\rthreev) &= (0,0), %\quad \text{with probability }1-\tfrac{1}{\dep}-\tfrac{1}{\sqrt{\dep} \\
\qquad \bigl(\stwov,\rtwov\bigr) =  
\Bigl(2^{2^{\level(v)}}\cdot\prod_{w\in \rightt(\pathh(v))} 2^{2^{\level(w)}},\;0 \Bigr) \\  
\bigl(\sonev,\ronev\bigr) & =
\biggl(W-\sum_{w\in\rightt(\pathh(v))}\stwow,\;
\Bigl(1-\tfrac{1}{\sqrt{\dep}}\Bigr)^{|\rightt(\pathh(v))|}\;\biggr)
\end{split}
\end{equation*}
and we have
$\Pr[\Sv=\sthreev]=1-\frac{1}{\sqrt{\dep}}-\frac{1}{\dep}$, 
$\Pr[\Sv=\stwov]=\frac{1}{\sqrt{\dep}},\;\Pr[\Sv=\sonev]=\frac{1}{\dep}$; see
Fig.~\ref{figure:adapgap}.    
%\end{enumerate}
Observe that $\stwov\leq W/2$,
and $\sonev>W/2$ for every node $v$.

\begin{figure}[ht!]
%\centering
\begin{tikzpicture}[
		round/.style={circle,fill=black!10, thick, minimum size=3mm},scale=0.8,
		squarednode/.style={rectangle, draw=black!60, fill=white!5,  thick, rounded corners, minimum size=7mm},
		]

% drawing tree and path 
\small
\node[round, fill=black!20, draw=black!50] (rho) at (0, 1) {$\rho$};
\node[round,fill=black!20,  draw=black!50] (root) at (0, 0) {$r$};
\node[round] (v1) at (-2, - 1.5) {};
\node[round, fill=black!20, draw=black!50] (v2) at (2,  -1.5) {$v_1$};
\node[round] (v3) at (-3,  -3) {};
\node[round] (v4) at (-1,  -3) {};
\node[round, fill=black!20, draw=black!50] (v5) at (1,  -3) {$v_2$};
\node[round] (v6) at (3,  -3) {};
\node[round] (v7) at (-3.5,  -4.5) {};
\node[round] (v8) at (-2.5,  -4.5) {};
\node[round] (v9) at (-1.5,  -4.5) {};
\node[round] (v10) at (-0.5,  -4.5) {};
\node[round] (v11) at (0.5,  -4.5) {};
\node[round, fill=black!20, draw=black!50] (v12) at (1.5,  -4.5) {$v_3$};
\node[round] (v13) at (2.5,  -4.5) {};
\node[round] (v14) at (3.5,  -4.5) {};
\node at (-3.5, -5) {$\vdots$};
\node at (-2.5, -5) {$\vdots$};
\node at (-1.5, -5) {$\vdots$};
\node at (-0.5, -5) {$\vdots$};
\node at (0.5, -5) {$\vdots$};
\node at (1.5, -5) {$\vdots$};
\node at (2.5, -5) {$\vdots$};
\node at (3.5, -5) {$\vdots$};
\node (x15) at (-3.65,  -5.45) {};
\node (x16) at (-3.35,  -5.45) {};
\node (x17) at (-2.65,  -5.45) {};
\node (x18) at (-2.35,  -5.45) {};
\node (x19) at (-1.65,  -5.45) {};
\node (x20) at (-1.35,  -5.45) {};
\node (x21) at (-0.65,  -5.45) {};
\node (x22) at (-0.35,  -5.45) {};
\node (x24) at (0.65,  -5.45) {};
\node (x23) at (0.35,  -5.45) {};
\node (x26) at (1.65,  -5.45) {};
\node (x25) at (1.35,  -5.45) {};
\node (x28) at (2.65,  -5.45) {};
\node (x27) at (2.35,  -5.45) {};
\node (x30) at (3.65,  -5.45) {};
\node (x29) at (3.35,  -5.45) {};

% \node[round] (v15) at (-3.75,  -6) {};
% \node[round] (v16) at (-3.25,  -6) {};
% \node[round] (v17) at (-2.75,  -6) {};
% \node[round] (v18) at (-2.25,  -6) {};
% \node[round] (v19) at (-1.75,  -6) {};
% \node[round] (v20) at (-1.25,  -6) {};
% \node[round] (v21) at (-0.75,  -6) {};
% \node[round] (v22) at (-0.25,  -6) {};
% \node[round] (v24) at (0.75,  -6) {};
% \node[round] (v23) at (0.25,  -6) {};
% \node[round] (v26) at (1.75,  -6) {};
% \node[round, fill=black!20, draw=black!50]  (v25) at (1.25,  -6) {};
% \node[round] (v28) at (2.75,  -6) {};
% \node[round] (v27) at (2.25,  -6) {};
% \node[round] (v30) at (3.75,  -6) {};
% \node[round] (v29) at (3.25,  -6) {};

\draw[black, very thick] (rho) -- (root) -- (v2) -- (v5) -- (v12);
%\draw[black, very thick] (x25) -- (v25);

\draw[black!13] (root) -- (v1) -- (v3) -- (v7);
\draw[black!13] (v2) -- (v6);
\draw[black!13] (v6) -- (v14);
\draw[black!13] (v1) -- (v4); 
\draw[black!13] (v3) -- (v8);
\draw[black!13] (v4) -- (v9);
\draw[black!13] (v4) -- (v10);
\draw[black!13] (v5) -- (v11);
\draw[black!13] (v6) -- (v13);

% \draw[black!13] (x15) -- (v15);
% \draw[black!13] (x16) -- (v16);
% \draw[black!13] (x17) -- (v17);
% \draw[black!13] (x18) -- (v18);
% \draw[black!13] (x19) -- (v19);
% \draw[black!13] (x20) -- (v20);
% \draw[black!13] (x21) -- (v21);
% \draw[black!13] (x22) -- (v22);
% \draw[black!13] (x23) -- (v23);
% \draw[black!13] (x24) -- (v24);
% \draw[black!13] (x26) -- (v26);
% \draw[black!13] (x27) -- (v27);
% \draw[black!13] (x28) -- (v28);
% \draw[black!13] (x29) -- (v29);
% \draw[black!13] (x30) -- (v30);

% Drawing sizes:
\node[left=of root, left=1pt] {$\left(\mathcolor{blue}{2^{2^{\dep} }}, \mathcolor{red}{0}\right)$};
\node[left=of v2, left=1pt] {$\left(\mathcolor{blue}{2^{(2^\dep + 2^{\dep-1})}}, \mathcolor{red}{0}\right)$};
\node[left=of v5,left=0pt] {$\left(\mathcolor{blue}{2^{(2^\dep + 2^{\dep-2})}}, \mathcolor{red}{0}\right)$};
\node[left=of v12,left=0pt] {$\left(\mathcolor{blue}{2^{(2^\dep + 2^{\dep-2}+2^{\dep-3})}}, \mathcolor{red}{0}\right)$};

% Drawing rewards:
\node[right=of root,right=1pt] {$\left(\mathcolor{blue}{W}, \textcolor{red}{1}\right)$};
\node[right=of v2,right=1pt] {$\left(\mathcolor{blue}{W - 2^{2^\dep}}, \mathcolor{red}{1 - \frac{1}{\sqrt{\dep}}}\right)$};
\node[right=of v5,right=0pt]  {$\left(\mathcolor{blue}{W - 2^{2^\dep}}, \mathcolor{red}{1 - \frac{1}{\sqrt{\dep}}}\right)$};

\node[right=of v12,right=0pt] {$\left(\mathcolor{blue}{W - 2^{2^\dep}-2^{(2^\dep+2^{\dep-2})}}, \mathcolor{red}{\left(1 - \frac{1}{\sqrt{\dep}}\right)^2}\right)$};

% Drawing Level labels
\node at (-4.7, 1) {\textbf{Level}};
\node at (-5, 0) {$\dep$}; 
\node at (-4.7, -1.5) {$\dep - 1$};
\node at (-4.7, -3) {$\dep - 2$};
\node at (-4.7, -4.5) {$\dep - 3$};
\node at (-4.7, -5) {$\vdots$};
\end{tikzpicture}
\caption{The (\textcolor{blue}{$\stwov$}, \textcolor{red}{$\rtwov$}),
  (\textcolor{blue}{$\sonev$}, \textcolor{red}{$\ronev$}) pairs are shown respectively on the left
  and right of each highlighted vertex in the tree.}
\label{figure:adapgap}
%\caption{The sizes $(\sonev,\stwov,\sthreev)$ and rewards $(\ronev,\rtwov,\rthreev)$ are indicated on the left and right, respectively, of each highlighted node in the tree. }
\end{figure}

%It is 
Importantly, note that any policy for this instance can obtain positive reward from at
most one item. This is because for any $v\in V$, $\sonev>W/2$.  
Therefore we can assume that any policy terminates upon observing a size $\sonev$ for
any visited vertex $v$. 
The binary tree is built so that a certain adaptive policy (see the proof of
Theorem~\ref{adapt-lbthm}) 
%explained in the next section,
can always reach a leaf-node if no positive reward has been collected in previous levels.
%reaching the leaf. 
%It is easy to see that this 
The construction of the tree prevents any path from going upward from a
node to its parent, as this will cause the length budget to run out.
%since the length budget will run out for the first node that does
%this.  
But more importantly, %as we will see in Section~\ref{adapgap-nonadap}, 
%the tree is built to prevent 
the instance is set up to preclude a policy from going to a left child of a node $v$ if
its instantiated size is $\stwov$ in the sense that if this happens 
%it will be impossible to 
then one cannot collect positive reward from this point on (Lemma~\ref{cheatlem}).
%the subtree rooted at the left child of $v$. 

The rest of this section is devoted to proving the following two results, from which
the adaptivity-gap lower bound immediately follows, 
%immediately from Theorems~\ref{adapt-lbthm} and~\ref{nonadapt-ubthm}, 
since $\dep=\Omega(\log B)$ and $\dep=\Omega(\log\log W)$ for
the above \csko instance.

\begin{theorem} \label{adapt-lbthm}
There is an adaptive policy for the above \csko instance that obtains $\Omega(1)$ expected reward.
\end{theorem}

\begin{theorem} \label{nonadapt-ubthm}
Any nonadaptive policy for the above \csko instance obtains expected reward at most
$\frac{2}{\sqrt{\dep}}$.
\end{theorem}  

%\subsection{There is an adaptive policy with expected reward $\Omega(1)$}
\subsubsection*{Proof of Theorem~\ref{adapt-lbthm}: lower bound on \boldmath $\OPT$.}
%We prove in this section that there is an adaptive policy with expected reward $\Omega(1)$
%for this instance of \csko.
Consider the following adaptive policy $\A$: the policy moves to node $r$ from $\rt$, and
then proceeds as follows. Let $v$ be the current node visited, which is $r$ initially.
%depending on the realization of the current node $v$'s size
%(where $v=r$ initially). 
%Let $\mathcal{A}$ be the adaptive policy that, starting at the root, decides which vertex
%to visit next based on the following criteria for the currently visited vertex $v$. 
If $v$ is a leaf, then the policy ends after the instantiation of $v$. Otherwise, the next
node visited by $\A$ is: the left child of $v$, if $\nsize=\sthreev$, and the right child
of $v$ if $\nsize=\stwov$; if $\nsize=\sonev$, then $\A$ stops and does not visit
any other nodes.   

Let $P^*$ denote the (random) path traversed by $\mathcal{A}$, which we may view as a
rooted path in $T$.
Let $\vl$ be the last vertex visited by $\A$, i.e., $\vl$ is the end-node of $P^*$ other
than $r$ and $P^*=\pathh(\vl)$. 
%First observe 
It is easy to see that %$\mathcal{A}$ 
$P^*$ does not exceed the knapsack and travel budgets, and
we lower bound the reward collected by $P^*$ in Lemma~\ref{lemma-adaptivereward}. 

\begin{claim}
We have $d(P^*)\leq B$ and $\nsize[](P^*):=\sum_{v\in P^*}\Sv\leq W$ with probability $1$.
\end{claim}

\begin{proof}
%$P^*$ is a rooted path in $T$, and 
The length of any rooted path in $T$ is at most
$\sum_{\ell=2}^\dep 2^{\ell-2}= 2^{\dep-1}-1=B$. 
%Since $P^*$ is a sequence of vertices in decreasing order of their levels in $\T$, it follows that 
%    \begin{equation*}
%        \di(P^*)\leq\sum_{\ell=2}^\dep2^{\ell-2}= 2^{\dep-1}-1= B.
%    \end{equation*}
% that $\mathcal{A}$ does not overflow the knapsack budget $W$. 
By construction,
for every node $w$ on $P^*$ before $\vl$, we have $\Sw\neq\sonew$, and
$\Sw=\stwow$ precisely if $w$ is a right-branching node of $P^*$. 
Therefore 
$\sum_{w\prec_{P^*}\vl}\Sw=\sum_{w\in \rightt(P^*)}\stwow=W-\nsize[\vl]^{(1)}$, and
since $\nsize[\vl]\leq\nsize[\vl]^{(1)}$, we obtain that $\nsize[](P^*)\leq W$. 
%\end{equation*}
\end{proof}

%We prove that $\mathcal{A}$ has expected reward $\Omega(1)$.
\begin{lemma}\label{lemma-adaptivereward} 
The expected reward collected by $P^*$ %$\mathcal{A}$ 
is at least $\frac{1-e^{-1}}{4}$.
\end{lemma}

\begin{proof}
%Note that $P^*=\pathh(\vl)$. 
Let $\Rew=\nrewd[](P^*):=\sum_{v\in P^*}\nrewd$ denote the reward obtained by $P^*$.
%$\mathcal{A}$. 
Note that $\vl$ is the only vertex from which $P^*$ %$\mathcal{A}$ 
can collect positive reward.
%we have that 
So $\E{\Rew}=\Pr\bigl[\nsize[\vl]=\nsize[\vl]^{(1)}\bigr]\cdot\E{\nrewd[\vl]^{(1)}}$.
Observe that if $\nsize[\vl]\neq\nsize[\vl]^{(1)}$, then $\vl$ is a leaf node, and hence
the event $\{\nsize[\vl]\neq\nsize[\vl]^{(1)}\}$ occurs precisely when $\mathcal{A}$ visits $\dep$
vertices, one on each level of $\T$, and none of them instantiate to size $\nsize^{(1)}$. 
Since vertex sizes are independent across different vertices, we have
$\Pr\bigl[\nsize[\vl]\neq\nsize[\vl]^{(1)}\bigr]=\bigl(1-\frac{1}{\dep}\bigr)^\dep\leq e^{-1}$,
and so $\E{\Rew}\geq (1-e^{-1})\E{\nrewd[\vl]^{(1)}}$.

We have $\nrewd[\vl]^{(1)}=\bigl(1-\frac{1}{\sqrt{\dep}}\bigr)^{|\rightt(P^*)|}$, and since
$\bigl(1-\frac{1}{\sqrt{\dep}}\bigr)^x$ is a convex function, we obtain that 
$\E{\nrewd[\vl]^{(1)}}\geq\bigl(1-\frac{1}{\sqrt{\dep}}\bigr)^{\E{|\rightt(P^*)|}}$. 
Observe that $v\in P^*$ gets included in $\rightt(P^*)$ precisely when $\nsize$
instantiates to $\stwov$, which happens with probability $\frac{1}{\sqrt{\dep}}$. 
So we can upper bound $\E{|\rightt(P^*)|}$ by $\dep\cdot\frac{1}{\sqrt{\dep}}=\sqrt{\dep}$.
It follows that 
$\E{\nrewd[\vl]^{(1)}}\geq\bigl(1-\frac{1}{\sqrt{\dep}}\bigr)^{\sqrt{\dep}}\geq\frac{1}{4}$,
where the last inequality uses the fact that $1-x\geq 4^{-x}$ for $x\leq 0.5$.
\end{proof}

\subsubsection*{Proof of Theorem~\ref{nonadapt-ubthm}: upper bound for non-adaptive policies}
Let $\sigma$ be some non-adaptive policy, %Then the expected reward of $\sigma$ is at most $\frac{2}{\sqrt{\dep}}$.
%\end{lemma}
%\begin{proof}
which we may again view as a rooted path in $T$, since we can always move first to $r$.
We may assume that %any non-adaptive policy 
$\sg$ %corresponds to a rooted path in $T$ that 
visits vertices in decreasing order of their levels, since any 
backtracking from a node $v$ to its ancestor would cause one to exceed the travel budget.
We say that an execution of $\sigma$ ``cheats'' if, for some visited node $v$, $\nsize$
instantiates to $\stwov$, and $\sigma$ proceeds to visit a vertex in the subtree of
$T$ rooted at the left-child of $v$.  

\begin{lemma} \label{claimcheat} \label{cheatlem}
%Let $\sigma$ be a non-adaptive policy. Then 
$\sigma$ does not collect any positive reward after cheating.
\end{lemma}

\begin{proof}
Suppose $\sg$ cheats at some vertex $u$.
%Let $u\in \sigma$ be a vertex in which $\sigma$ cheats. 
%Let $U\subset V$ denote the vertices on the binary tree rooted at the left-child of $u$.  
Let $v$ be any node in the tree rooted at the left-child of $u$.
The residual knapsack budget after visiting $u$ is at most $W-\stwou$. 
It suffices to show that $\sonev>W-\stwou$. %for any $v\in U$. 
Since $\sonev=W-\sum_{w\in \rightt(path(v))}\stwow$, 
%for any $v\in V$, 
this amounts to showing that $\stwou>\sum_{w\in A}\stwow$,
%
%\begin{equation}\label{size2eq}
%\stwou>\sum_{w\in \rightt(\pathh(v))}\stwow,\;\mbox{ for every $v\in U$.}
%\end{equation}
where $A=\rightt(\pathh(v))$.
We argue that $\stwow<\stwou$ for every $w\in A$, and the $\stwow$'s are
distinct for $w\in A$. This, coupled with the fact
that $\stwou$ and the $\stwow$'s are all powers of $2$, implies the above inequality.

Recall that for a node $z$, we have 
$\nsize[z]^{(2)}=2^{2^{\level(z)}}\cdot\prod_{w\in\rightt(\pathh(z))}2^{2^{\level(w)}}$. 
Let $A=\{a_1,a_2,\ldots,a_{|A|}\}$, where the nodes are ordered in increasing order of
their distance from $r$. Then, for any $i\geq 2$, we have
$\nsize[a_i]^{(2)}=2^{2^{\level(a_i)}}\cdot\nsize[a_{i-1}]^{(2)}$, showing that each
$\nsize[a_i]^{(2)}$ is a distinct power of $2$, and $\nsize[a_i]^{(2)}$ increases with $i$.
Note that $u\notin\rightt(\pathh(v))$
and $\rightt(\pathh(u))\sse A$. So for $z=a_{|A|}$, we have
$\nsize[z]^{(2)}=\prod_{w\in A-\rightt(\pathh(u))}2^{2^{\level(w)}}\cdot\prod_{w\in\rightt(\pathh(u))}2^{2^{\level(w)}}$
and $\prod_{w\in A-\rightt(\pathh(u))}2^{2^{\level(w)}}<2^{2^{\level(u)}}$. It follows that $\nsize[z]^{(2)}<\stwou$.
%
\begin{comment}
for any $v\in V$, and observe that for any $v\in
U$, 
Observe that $\pathh(u)$ is a subpath of $\pathh(v)$ and $u\notin\rightt(\pathh(v))$.  
Hence, all of the terms inside the sum on the RHS of \ref{size2eq} are distinct powers of two, each of which is less than $\stwou$ (a power of two itself). It follows that $\stwou>\sum_{w\in \rightt(\pathh(v))}\stwow$, as desired.
%    Maybe make the latter argument more explicit.
\end{comment}
\end{proof}

Recall that we view $\sg$ also as a rooted path in $T$. We argue that the total
expected reward obtained from $\rightt(\sg)$ and $\leftt(\sg)$ are both at most
$\frac{1}{\sqrt{\dep}}$, which will finish the proof.

\begin{lemma}
%Let $\sigma$ be a non-adaptive policy. 
The expected reward obtained from $\rightt(\sigma)$ is at most $\frac{1}{\sqrt{\dep}}$.
\end{lemma}

\begin{proof}
We have
\begin{equation*}
%\begin{split}
\E{\nrewd[](\rightt(\sigma))}
\leq \sum_{w\in \rightt(\sigma)}\Pr[\Sw=\sonew]\cdot
\Bigl(1-\tfrac{1}{\sqrt{\dep}}\Bigr)^{|\rightt(\pathh(w))|} 
%\leq \sum_{w\in \rightt(\sigma)}\frac{1}{\dep}\cdot\Bigl(1-\tfrac{1}{\sqrt{\dep}}\Bigr)^{\dep-\level(w)} \\
\leq\sum_{i=1}^\dep\frac{1}{\dep}\cdot\Bigl(1-\tfrac{1}{\sqrt{\dep}}\Bigr)^{i-1}\leq\frac{1}{\sqrt{\dep}}.
%\qedhere
%\end{split}
\end{equation*}
The second inequality follows because the exponents $|\rightt(\pathh(w))|$, as $w$ ranges
over $\rightt(\sg)$, are distinct elements of $\dbrack{\dep-1}$. 
\end{proof}

\begin{lemma}
%Let $\sigma$ be a non-adaptive policy. 
The expected reward obtained from $\leftt(\sigma)$ is at most $\frac{1}{\sqrt{\dep}}$.
\end{lemma}

\begin{proof}
For any $v\in\sigma$, 
%let $\Piv$ be the random variable that indicates the reward collected from $v$ in $\sigma$.
let $\mathcal{E}_v$ denote the event that $\Sw=\sthreew$ for all $w\in\leftt(\pathh(v))-v$.
%$w\in\leftt(\sigma)$ such that $w\prec_\sigma v$. 
By Lemma~\ref{claimcheat} and since $\sonew>W/2$ for every $w\in V$, we have that
$\sigma$ can collect positive reward from $v$ only if event $\mathcal{E}_v$ occurs.
So the expected reward obtained from $v$ is at most
$\Pr[\Ec_v]\cdot\Pr[\Sv=\sonev|\Ec_v]\cdot\ronev\leq\frac{1}{\dep}\cdot\Pr[\Ec_v]$,
where we use the fact that $\ronev\leq 1$, and
$\Pr[\Sv=\sonev\,|\,\Ec_v]=\Pr[\Sv=\sonev]$ since $v$'s realization is
independent of the realizations of other nodes.
Again, since different nodes are independent, we have 
$\Pr[\Ec_v]=\prod_{w\in\leftt(\pathh(v)),\;w\neq v}\Pr\bigl[\Sw=\sthreew\bigr]
=\bigl(1-\frac{1}{\sqrt{\dep}}-\frac{1}{\dep}\bigr)^{|\leftt(\pathh(v))|-1}$.
So the total expected reward obtained from $\leftt(\sg)$ is at most 
\begin{equation*}
\sum_{v\in\leftt(\sg)}\frac{1}{\dep}\cdot
\Bigl(1-\tfrac{1}{\sqrt{\dep}}-\tfrac{1}{\dep}\Bigr)^{|\leftt(\pathh(v))|-1}
%\leq\sum_{v\in\leftt(\sg)}\frac{1}{\dep}\cdot
%\Bigl(1-\tfrac{1}{\sqrt{\dep}}-\tfrac{1}{\dep}\Bigr)^{\dep-\level(v)}
\leq\sum_{i=1}^{\dep}\frac{1}{\dep}\cdot
\Bigl(1-\tfrac{1}{\sqrt{\dep}}-\tfrac{1}{\dep}\Bigr)^{i-1}
\leq \frac{1}{\dep}\cdot\frac{1}{\frac{1}{\sqrt{\dep}}+\frac{1}{\dep}}\leq\frac{1}{\sqrt{\dep}}.
%\qedhere
\end{equation*}
Again, the first inequality follows because
$|\leftt(\pathh(v))|-1$, for $v\in\leftt(\sg)$, are distinct elements of $\dbrack{\dep-1}$.
\end{proof}

\section{Approximation algorithms for \boldmath \csko} \label{sectioncskoapprox} \label{csko-approx}
We now devise non-adaptive approximation algorithms for \csko.
In Section~\ref{csko-quasipoly}, we develop an
$O(\log\log{W})$-approximation algorithm %for \csko 
%correlated knapsack orienteering (\csko) 
with $(n+\log B)^{O(\log W\log\log W)}$ (i.e., quasi-polynomial) running time, which will
prove Theorem~\ref{theorem-approxcsko}, and in Section~\ref{csko-poly}, we obtain a polytime
$O(\log W)$-approximation algorithm, thereby proving Theorem~\ref{poly-approxcsko}. 

\subsection{Quasi-polytime \boldmath $O(\log\log W)$-approximation algorithm}
%  Theorem~\ref{approx-cskothm}}  
\label{csko-quasipoly} 
There are two chief components underlying our algorithm.
First, we isolate a key structural result (Theorems~\ref{structhm} and~\ref{strucdthm})
showing that from an optimal adaptive policy, one can extract a suitable path $\spath$
and certain ``portal'' vertices on this path, such that the subpaths of $\spath$ between
these portal vertices satisfy various nice properties. 
Second, we show how to exploit this structural result algorithmically.
The above properties allow us to reduce the problem, roughly speaking at the expense of an
$O(\log\log W)$-factor loss, to that of finding the portal vertices, and suitable paths
between these portal vertices that satisfy certain knapsack constraints on the total
expected truncated size $\E{\min\{\Sv,2^j\}}$ of nodes on these paths. 
We ``guess'' (i.e., enumerate over all possible choices of)
these portal vertices and some auxiliary information, and set up a configuration LP
\eqref{cskolp} to
find paths between these portal vertices. This configuration LP can be solved
near-optimally, and we show that a fractional solution can be rounded incurring only an  
$O(1)$-factor loss in the objective and in the constraints. Finally, we argue that 
this leads to an $O(\log\log W)$-approximation non-adaptive policy.

Our approach is similar in spirit to the one in~\cite{BansalN14} for \cso, and in
Section~\ref{cso-approx}, we show that our approach also yields an 
$O(\log\log B)$-approximation for \cso, which improves upon the guarantee in~\cite{BansalN14}
by an $O\bigl(\frac{\log\log B}{\log\log\log B}\bigr)$-factor. While we borrow various
ingredients from~\cite{BansalN14}, the key difference between our approach and theirs is
that we extract much more information from the adaptive policy in terms of so-called portal
vertices, which enables us to round an underlying configuration LP incurring only a
{\em constant-factor} violation in the knapsack constraints; in contrast, 
%the rounding 
this step in~\cite{BansalN14} incurs an $O\bigl(\frac{\log\log B}{\log\log\log B}\bigr)$-factor 
violation of the constraints, and this savings is the source of our improved guarantee.

%\vspace*{-1ex}
\subsubsection{\boldmath Structural results for \csko}
%\paragraph{Notation.}
%Before stating our structural result, 
\begin{comment}
We first define and recall some notation. %and conventions.
We always think of the nodes on a rooted path $P$ as being ordered in increasing order of
their distance from $\rt$ along the path. 
%
%We extend the relation $\prec_P$ to sets of
%vertices: given two sets $S,T\sse V\cup\{\rt\}$ of vertices on $P$, we say that
%%$S\prec_P T$ if $u\prec_P v$ for every $u\in S, v\in T$, and 
%$S\preceq_P T$ if $u\preceq_P v$ for every $u\in S, v\in T$. Note that 
%%$S\prec_P T$ implies that $S\cap T=\es$, and 
%$S\preceq_P T$ implies that $|S\cap T|\leq 1$ and if $S\cap T=\{w\}$, then $w$ is the last
%node in $S$ and the first node in $T$.
%
We will interchangeably think of a path as an edge-set, or a sequence of nodes; the
meaning will be clear from the context.
\end{comment}
Recall that $\dbrack{k}:=\{0,1,\ldots,k\}$ for an integer $k\geq 0$.
Recall that, for a path $P$ and nodes $a,b\in P$, we use $P_{a,b}$ to denote the $a$-$b$
portion of $P$. 
%and we view the nodes on $P_{a,b}$ as ordered from $a$ to $b$ (i.e., in increasing
%order of their distance from $a$ along $P_{a,b}$).
%We will often need to consider the nodes of a path $P$ excluding one of its end-points:
%for a $u$-$v$ path $P$, which we view as going from $u$ to $v$, define 
%We use $\intp{P}$ to denote the internal nodes of $P$, i.e., the nodes of $P$ excluding its
%end-points. 
If $P$ is a $u$-$v$ path, 
its {\em regret} is $\dreg(P):=d(P)-d(u,v)$, and the
{\em two-point regret of $P$} with respect to a node $a\in P$ is
$\dreg(P,a):=d(P)-d(u,a)-d(a,v)=\dreg(P_{u,a})+\dreg(P_{a,v})$. 
For an index $j\in\{0,1,\ldots,L:=\ceil{\log W}\}$, recall that we define
$\Xvj:=\min\{\nsize,2^j\}$ and $\muvj:=\E{\Xvj}$.
For any vertex $v\in V$, let 
$\piv(t):=\E{\nrewd\cdot\bon_{\nsize\leq W-t}}=\sum_{t^\prime=0}^{W-t}\Pr[\Sv=t^\prime]\cdot\E{\Rv\,|\,\Sv=t^\prime}$
%\sum_{t^\prime=t}^W\Pr[\Sv=t^\prime]\cdot\Ex[\Rv\,|\,\Sv=t^\prime]$ 
denote the
expected reward obtained from $v$ if its processing starts at time $t$.  
Note that $\piv(t)=0$ for any $t>W$. Also, note that $\pi_{\rt}(t)=0$ for all $t$.
We may assume that $\pi_v(0)\leq\OPT/4$ for every $v\in V$, as otherwise, we can obtain
$\Omega(\OPT)$ reward by going to a single node.

Throughout, let $K=3\log(6\log W)+12$, $L=\ceil{\log W}$, $N_1=2(K+1)$. 
Define $\nd_{-1}:=\rt$. %and $\ndset_{-1}:=\{\rt\}$.

\begin{theorem} \label{structhm} 
There exists a rooted path $\spath$ with $d(\spath)\leq B$, 
vertices $\nd_0\preceq\nd_1\preceq\ldots\preceq\nd_k$ on $\spath$ for some $k\leq L$, 
and, for each $j\in\dbrack{k}$, a vertex-set 
$\ndset_j\sse\spath_{\nd_{j-1},\nd_j}$ containing nodes $\nd_{j-1}$, $\nd_j$, with 
%$\ndset_j\supseteq\{\nd_{j-1},\nd_j\}$,
$|\ndset_j|\leq N_1$, whose vertices are ordered by the order they appear on $\spath$,
satisfying the following properties. 
%For each $j\in[k]$, we order the nodes in $\ndset_j$ in the order they are visited by
%$\spath$  
%
\begin{enumerate}[label=(\alph*), topsep=0.1ex, noitemsep, leftmargin=*]
%\item $\ndset_j\supseteq\{\nd_{j-1},\nd_j\}$ and $|\ndset_j|\leq N_1$ for all $j\in\dbrack{k}$.
%and $\ndset_{j-1}\preceq\ndset_{j}$ for all $j\in[k]$. 
%and for all $j=0,1,\ldots,k-1$.
\item %(Reward) 
$\sum_{j=0}^k\sum_{v\in\spath_{\nd_{j-1},\nd_j}-\nd_j}\pi_v(2^j-1)
%+\pi_{\nd_k}(2^{k+1}-1)
\geq\OPT/4$.  
\item %(Prefix-size) 
$\mu^j(\spath_{\rt,\nd_j}-{\nd_j})\leq (K+1)2^j$ for all $j\in\dbrack{k}$.
\item %(Size) 
For every $j\in\dbrack{k}$ and
%considering nodes in $\ndset_j$ in the order they are visited on $\spath$, 
consecutive nodes $a,b\in\ndset_j$,
we have $\mu^j(\spath_{a,b}-b)\leq 2^j$.
\end{enumerate} 
\end{theorem}  

To avoid detracting the reader, we defer the proof of Theorem~\ref{structhm} to the end of
Section~\ref{csko-approx}, but we note
that once we have $\spath$ and the $\nd_j$ vertices satisfying parts (a) and (b), we can
obtain the $\ndset_j$-sets by simply splitting each $\spath_{\nd_{j-1},\nd_j}$ into at
most $2(K+1)$ segments.
As mentioned earlier, in our quasi-polytime algorithm, %our approach now is to 
we utilize Theorem~\ref{structhm} to construct a
good rooted path, by using enumeration to guess %certain quantities, namely
$\bigcup_j\ndset_j$, %in combination with 
and an LP to then obtain suitable paths between
consecutive nodes of $\bigcup_j\ndset_j$.  
%For the latter task, and 
In order to ensure that the total path length is at most the travel budget $B$, 
we will also need to obtain some information about the lengths
$d(\spath_{a,b})$ %of the portion of $\spath$ 
for consecutive nodes $a,b$ in $\bigcup_j\ndset_j$. 
Naively guessing these lengths would yield incur a %an additional
large $B^{O(LN_1)}$-factor in the running time; to do better, and reduce the dependence to
$(\log B)^{O(LN_1)}$, we instead guess the two-point regret of each $\spath_{a,b}$ with
respect to a ``mid-point'' node, within a factor of $2$, which suffices (see
Fig.~\ref{strucfig}). We refine Theorem~\ref{structhm} to incorporate these estimates as follows.

\begin{theorem}[Main structural result] \label{strucdthm} 
%Define $\nd_{-1}:=\rt$. %and $\ndset_{-1}:=\{\rt\}$.
%Objects (i)--(iii) exist, and satisfy properties~\ref{distance}--\ref{size} listed
%below. \\  
%We can define the following quantites, which satisfy properties (P1)--(P5) listed
%below.  
%(i) An ordered sequence of $k\leq L$ vertices $\nd_0,\ldots,\nd_k$. \\
%(ii) For each $j\in\dbrack{k}$, an ordered sequence of at most $N_1$ nodes $\ndset_j$
%whose first and last nodes are $\nd_{j-1}$ and $\nd_j$ respectively. \\ 
Let the node-sequence $\nd_0,\ldots,\nd_k$, where $k\leq L$, and for each
$j\in\dbrack{k}$, the ordered node sequence $\ndset_j$ of at most $N_1$ nodes, be as
given by Theorem~\ref{structhm}.
Define $\fullset:=\bigcup_{j=0}^k\ndset_j$, which we call ``portal nodes'',
where the ordering of nodes in $\fullset$ is %(the ordered sequence) 
$\ndset_0$, $\ndset_1$, \ldots, $\ndset_k$; 
for $a\in\fullset$, $a\neq\nd_k$, let $\nxt(a)$ be the next node in $\fullset$ after $a$. 
%\\
%(iii) auxiliary node $m_v$ and integer $r_v\geq 0$ for each 
%$v\in\fullset:=\bigcup_{j=0}^k\ndset_j$, $v\neq\nd_k$; and
%(iii) %For each pair of consecutive nodes $a,b\in\fullset$, %$a\neq\nd_k$, 
%an $a$-$b$ path $\spath_{a,b}$, 
For each $a\in\fullset-\nd_k$, there exists an $a$-$\nxt(a)$ path $\spath_{a,\nxt(a)}$,
auxiliary node $m_a$, and integer $\gm_a\geq 0$, such that the following properties hold.  
%Properties (P1)--(P5) are as follows.
%
\begin{enumerate}[label={\textnormal{(P\arabic*)}}, %ref={\ref{strucdthm}\,(\alph*)},
topsep=0.2ex, noitemsep, leftmargin=*]
\item \textnormal{(Distance)} 
$d(\spath_{a,b})\leq\dist_{a}:=2^{\gm_a}-1+d(a,m_a)+d(m_a,b)$
for %every $j\in\dbrack{k}$ and 
every pair of consecutive nodes $a,b\in\fullset$.
\label{distance}

\item \textnormal{(Total-length)} $\sum_{a\in\fullset-\nd_k}\dist_{a}\leq B$.
\label{totlen}

\item \textnormal{(Reward)}
%$\sum_{j=0}^k\sum_{v\in(\bigcup_{a\in\ndset_j-\nd_j}\spath_{a,\nxt(a)})-\nd_j}\pi_v(2^j-1)
%+\pi_{\nd_k}(2^{k+1}-1)\bigr]\geq\OPT/4$.  
$\sum_{j=0}^k\sum_{a\in\ndset_j-\nd_j}\sum_{v\in\spath_{a,\nxt(a)}-\nxt(a)}\pi_v(2^j-1)
%+\pi_{\nd_k}(2^{k+1}-1)
\geq\OPT/8$.  
\label{reward}

\item \textnormal{(Prefix-size)}
%$\sum_{a\in(\bigcup_{h=0}^j\ndset_h)-\nd_j}\mu^j(\spath_{a,\nxt(a)})-\mu^j_{\nd_j}\leq (K+1)2^j$ 
%$\sum_{h=0}^j\sum_{a\in\ndset_h-\nd_h}\sum_{v\in\spath_{a,\nxt(a)}-\nxt(a)}\mu^j_v\leq (K+1)2^j$ 
$\sum_{h=0}^j\sum_{a\in\ndset_h-\nd_h}\mu^j\bigl(\spath_{a,\nxt(a)}-\nxt(a)\bigr)\leq (K+1)2^j$ 
for all $j\in\dbrack{k}$. 
\label{prefix}

\item \textnormal{(Size)} 
$\mu^j(\spath_{a,b}-b)\leq 2^j$ for every $j\in\dbrack{k}$ and pair of consecutive nodes
$a,b\in\ndset_j$. 
\label{size}  
\end{enumerate} 
\end{theorem} 

\begin{figure}[ht!]
%\centering
\begin{tikzpicture}[
		round/.style={circle,fill=black, thick, minimum size=1mm},scale=0.8,
		squarednode/.style={rectangle, draw=black!60, fill=white!5,  thick, rounded corners, minimum size=5mm},spy using outlines]

\small
% zoom-in box
 
% portal nodes
\node[squarednode] (rho) at (0, 0) {};
\node at (0, -.6) {$\nd_{-1} = \rt$};
\node[squarednode] (v0) at (1.5, 0) {};
\node at (1.5, -.6) {$\nd_0$}; 
\node[squarednode] (v1) at (2.5, 0) {};
\node at (2.5, -.6) {$\nd_1$}; 
\node[squarednode] (v2) at (4.5, 0) {};
\node at (4.5, -.6) {$\nd_2$}; 
\node[squarednode] (v3) at (6, 0) {};
\node at (6, -.6) {$\nd_3$}; 
\node (v3.5) at (7,0) {$\cdots$};
\node[squarednode, fill=black] (vj1) at (8, 0) {};
\node at (8, -.6) {$\nd_{j-1}$}; 
\node[squarednode, fill=black] (vj) at (14, 0) {};
\node at (14, -.6) {$\nd_{j}$}; 
\node (vj2) at (15, 0) {$\cdots$};

% minor nodes 
\node[round, fill=black, scale=.7] at (8.5, 0) {};
\node[round, scale=.7] at (8.9, 0) {};
\node[round, scale=.7] (a0) at (10, 0) {};
\node (a) at (10, -.4) {$a$};
\node[round, scale=.7] (nexta) at (11, 0) {};
\node at (11, -.4) {$\mathrm{next}(a)$};
\node[round, scale=.7] at (12.3, 0) {}; % to make bigger, change scale
\node[round, scale=.7] at (12.7, 0) {};
\node[round, scale=.7] at (13.4, 0) {};

% edges 
\draw[black] (rho) -- (v0) -- (v1) -- (v2) -- (v3) -- (6.65, 0);
\draw[black] (7.3, 0) -- (vj1) -- (vj) -- (vj2);

% Manual spy layer
\draw[blue] (9.5, -.8) rectangle ++(2.3, 1.4); %small box. Syntax: (bottom left coordinate) rectangle ++(width, height)
\node at (10.55, 1) {\textcolor{blue}{$Q^*_{a, \mathrm{next}(a)}$}}; % Q*a label
\draw[blue] (1, 1.1) rectangle ++(4, 3); % large box
\node at (3, 4.5) {\textcolor{black}{$\mu^j\left(Q^*_{a, \mathrm{next}(a)} - \mathrm{next}(a)\right) \leq 2^j$}}; % mu label 
\draw[blue] (5, 2.67) -- (9.53, .25);
% inside of zoom box
\node[round] (ma) at (2.9, 2.5) {};
\node at (2.9, 1.9) {$m_a$};
\node[round, minimum size=4mm] (A) at (1.5, 2.5) {};
\node[round, minimum size=4mm] (NEXTA) at (4.2, 2.5) {};
\node at (1.5, 1.9) { $a$};
\node at (4.2, 1.9) { $\mathrm{next}(a)$};

\draw (1, 2.5) -- (A) -- (ma);
\draw[dotted] (ma) -- (NEXTA);
\draw (NEXTA) -- (5, 2.5);

\draw[draw=red, snake it] (2.9, 2.6) arc  (0:170:.7cm);
\draw[draw=red, style=double,  thick, snake it] (4.2, 2.65) arc  (0:180:.64cm);
\draw[draw=black, very thick, snake it] (4.2, 2.65) arc  (0:180:.64cm);
% end of zoom box
 
% BRACE
% To change: adjust the y-coordinate of the second points (change everywhere where you see -1.5)
\draw[dotted] (rho) -- (0, -1); 
\draw[dotted] (vj) -- (14, -1); 
\draw [brace, very thick] (0,-1)  --  node [position label, below=.5] {\large $\mu^j\left(Q^*_{\rt, \nd_j} - \nd_j\right) \leq (K+1)2^j$}   (14, -1);
\end{tikzpicture}
\caption{The portal nodes $\fullset$ and paths between consecutive portal nodes. The
  solid nodes depict $\fullset_j$.}
\label{strucfig}
%\caption{The sizes $(\sonev,\stwov,\sthreev)$ and rewards $(\ronev,\rtwov,\rthreev)$ are indicated on the left and right, respectively, of each highlighted node in the tree. }
\end{figure}

\begin{proof}
We show how Theorem~\ref{structhm} leads to the stated result.
Let $P$ be the rooted path given by Theorem~\ref{structhm}.
%We take the $\nd_j$ nodes and $\ndset_j$ node-sets for $j=0,1,\ldots,k\leq L$, to be as
%given by Theorem~\ref{structhm}, where the ordering of nodes in each $\ndset_j$ is the
%order in which they appear on $P$. 
For a node $a\in\fullset-\nd_k$, we will define $\spath_{a,\nxt(a)}$
to consist of a subset of nodes of $P_{a,\nxt(a)}$. So the (Prefix-size)
property~\ref{prefix}, and (Size) property~\ref{size} follow from parts (b), (c) of
Theorem~\ref{structhm} respectively.  
%(We have 
%$\sum_{h=0}^j\sum_{a\in\ndset_h-\nd_j}\mu^j(\spath_{a,\nxt(a)})
%\leq\sum_{h=0}^j\sum_{a\in\ndset_h-\nd_j}\mu^j(P_{a,\nxt(a)})
%=\mu^j(\spath_{\rt,\nd_j})$.)

Consider some $j\in\dbrack{k}$, and fix a node $a\in\fullset_j-\nd_j$. 
Let $b=\nxt(a)$ and $b'$ be the node just preceding $b$
on $P$. First define $m_a\in P_{a,b'}$ such that 
$\sum_{v\in P_{a,m_a}}\pi_v(2^j-1)$ and $\sum_{v\in P_{m_a,b'}}\pi_v(2^j-1)$ are both at
least $\frac{1}{2}\cdot\sum_{v\in P_{a,b'}}\pi_v(2^j-1)$.
%$\sum_{v\in P_{m_ab'}}\pi_v(2^j-1)\geq\frac{1}{2}\cdot\sum_{v\in P_{ab'}}\pi_v(2^j-1)$ and
Let $\gm_a\geq 0$ be the largest integer such that
$2^{\gm_a}-1\leq\dreg(P_{a,b},m_a)=d(P_{a,b})-d(a,m_a)-d(m_a,b)=\dreg(P_{a,m_a})+\dreg(P_{m_a,b})$. 
By definition then, we have that $\dist_{a}\leq d(P_{a,b})$, and since $d(P)\leq B$, the
(Total-length) property~\ref{totlen} follows.
%To define $\spath_{a,\nxt(a)}$, $m_a$, and $\gm_a$
We specify the path $\spath_{a,b}$ by specifying the sequence of nodes (starting at $a$,
ending at $b$) comprising it.
%with first and last nodes $a$ and $b$ respectively.
We set $\spath_{a,b}=P_{a,m_a},b$ if $\dreg(P_{a,m_a})\leq\dreg(P_{m_a,b})$ and 
$\spath_{a,b}=a,P_{m_a,b}$ otherwise. 
Since 
$\sum_{v\in\spath_{a,b}-b}\pi_v(2^j-1)\geq\frac{1}{2}\cdot\sum_{v\in P_{a,b}-b}\pi_v(2^j-1)$, 
the (Reward) property~\ref{reward} follows from part (a) of Theorem~\ref{structhm}.

Finally, we argue that $d(\spath_{a,b})\leq\dist_{a}$ proving the (Distance)
property~\ref{distance}. 
This is equivalent to showing that $\dreg(\spath_{a,b},m_a)\leq 2^{\gm_a}-1$.
This follows because by construction, we have 
$\dreg(\spath_{a,b},m_a)=\min\bigl\{\dreg(P_{a,m_a}),\dreg(P_{m_a,b})\}\leq\dreg(P_{a,b},m_a)/2$,
and by definition of $\gm_a$, we have $\dreg(P_{a,b},m_a)<2^{\gm_a+1}-1$, so 
$\dreg(P_{a,b},m_a)\leq 2^{\gm_a+1}-2$.
\end{proof} 

%\vspace*{-2ex}
\subsubsection{Configuration LP and non-adaptive %\boldmath $O(\log\log W)$-approximation
  algorithm} 
In the sequel, to avoid cumbersome notation, we assume that we have found, by enumeration,
nodes $\nd_0,\ldots,\nd_k$, where $k\leq L$, ordered node-sets $\ndset_j$ for
$j\in\dbrack{k}$, and length bounds $\dist_{a}$ for every pair of consecutive nodes 
$a,b\in\fullset:=\bigcup_{j=0}^k\ndset_j$, as stipulated by Theorem~\ref{strucdthm}. 
(We also need to enumerate for $\{m_a,\gm_a\}_{a\in\fullset-\nd_k}$; we do not use
these quantities directly, but these are used to specify the $\dist_{a}$ length bounds.)
That is, these objects are compatible with suitable $\spath_{a,b}$ paths
such that~\ref{distance}--\ref{size} hold. 
Clearly, this enumeration takes $(n\log B)^{O(N_1L)}=(n\log B)^{O(\log W\log\log W)}$
time, which is the source of the running time in Theorem~\ref{approx-cskothm}.

%We call the nodes in $\fullset$, ``portal nodes.''
We formulate a configuration LP to find $a$-$b$ paths, for every pair of consecutive
nodes $a,b\in\fullset$, satisfying properties \ref{distance}, \ref{reward}--\ref{size}.
To this end, fix some $j\in\dbrack{k}$ and 
$a\in\ndset_j-\nd_j$, and let $b=\nxt(a)$. The valid $a$-$b$
paths (i.e., the configurations) are the solutions to the following (deterministic) 
{\em point-to-point knapsack orienteering} (\knapo) problem: the end-nodes are $a$, $b$,
the length budget is $\dist_{a}$, the knapsack weights are $\mu^j_v$ 
for all $v\in V-b$ and $\mu^j_b=0$, and the knapsack-budget is $2^j$. 
%Observe that any solution to this \knapo instance
Let $\I_{a}$ denote the set of all feasible solutions to this \knapo instance.

The configuration LP has variables $x^a_\tau$, for every $a\in\fullset-\nd_k$ and
$\tau\in\I_a$, indicating the $a$-$\nxt(a)$ paths that are chosen. 
%and variables $y^j_v$
%for every $v\in V$, indicating that node $v$ lies on an $a$-$\nxt(a)$ path for some
%$a\in\ndset_j-\nd_j$. 
%For notational convenience, set $y^j_{\nd_k}=0$ for all $j=0,1,\ldots,k$.
%Note that nodes in $\fullset$ are always on some $a$-$\nxt(a)$ path.
%Recall that $\intp{P}$ denotes the internal nodes of path $P$.
%
\begin{alignat}{3}
\max & \quad & \sum_{j=0}^k\,\sum_{a\in\ndset_j-\nd_j}\,\sum_{\tau\in\I_a} 
x^a_\tau\cdot\Bigl(\sum_{v\in\tau-\nxt(a)}\pi_v&(2^j-1)\Bigr) \tag{CKO-P} \label{cskolp}
\\
\text{s.t.} & \quad & \sum_{\tau\in\I_a}x^a_\tau & = 1 \qquad
&& \forall a\in\fullset-\nd_k \label{config} \\
%&& y^j_v & = \sum_{a\in\ndset_j-\nd_j}\sum_{\tau\in\I_a:v\in\tau-\nxt(a)}x^a_\tau \qquad
%&& \forall v\in V,\,j\in\dbrack{k} \label{nodecov} \\
&& \sum_{a\in\fullset-\nd_k}\,\sum_{\tau\in\I_a:v\in\tau-\nxt(a)}x^a_\tau & \leq 1 
\qquad && \forall v\in V \label{covbnd} \\
&& \sum_{h=0}^j\,\sum_{a\in\ndset_h-\nd_h}\,\sum_{\tau\in\I_a}
x^a_\tau\cdot\mu^j\bigl(\tau-\nxt(a)\bigr) & \leq (K+1)2^j \qquad &&
\forall j\in\dbrack{k} \label{prefixlp} \\
&& x & \geq 0. \notag
\end{alignat} 
Constraints \eqref{config} encodes that we select an $a$-$b$ path for every consecutive
pair of nodes $a,b\in\fullset$, and
constraints \eqref{covbnd} ensure that each node $v$ lies
on at most one of these $a$-$b$ paths; %and that the portal nodes are visited;  
constraint \eqref{prefixlp} encodes the (Prefix-size) property \ref{prefix}. 
(Note that if $\nd_{h-1}=\nd_h$, then $\ndset_h=\{\nd_h\}$, so we do not have any term
for index $h$ in the objective function, and on the LHS of \eqref{prefixlp}.)
%For notational convenience, set 
%Observe that constraints \eqref{config} and \eqref{nodecov} set $y^j_a=1$ for all
%$j\in\dbrack{k}$, $a\in\ndset_j-\nd_j$; also, $y^j_{\nd_k}=0$ for all $j=0,1,\ldots,k$
%due to \eqref{nodecov}.

To gain some intuition, notice that
Theorem~\ref{strucdthm} shows that there is a feasible integral solution to \eqref{cskolp}
of objective value at least $\OPT/8$: we set $x^a_\tau=1$ for $\tau=\spath_{a,\nxt(a)}$
for every $a\in\fullset-\nd_k$. Properties~\ref{distance} and~\ref{size} show that
$\spath_{a,\nxt(a)}\in\I_a$;  
%The $y^j_v$ variables are set by constraints \eqref{nodecov};
%We set the $y^j_v$s so that constraints \eqref{nodecov}
%hold at equality. %and \eqref{covbnd} holds. 
%that is, we have $y^j_v=1$ for all
%$v\in\bigl(\bigcup_{a\in\ndset_j-\nd_j}\spath_{a,\nxt(a)}\bigr)-\nd_j$, 
%$j\in\dbrack{k}$. %and $y^k_{\nd_k}=1$.) 
property~\ref{prefix} shows that \eqref{prefixlp} holds,
and~\ref{reward} shows that the objective value is at least $\OPT/8$. 
%Properties~\ref{distance} and~\ref{size} show that

We can solve \eqref{cskolp} approximately, given an approximation algorithm for \knapo,
since this can be used to obtain an approximate separation oracle for the dual of
\eqref{cskolp}. %We prove the following in Appendix~\ref{append-csko}.

\newcommand{\optcskolp}{\ensuremath{\OPT_{\text{\textsf{\ref{cskolp}}}}}} %{\OPT_{\mathsf{LP}}}
\begin{claim} \label{cskolpval}
The optimal value of \eqref{cskolp}, $\optcskolp$, is at least $\OPT/8$.
\end{claim}

\begin{lemma} \label{cskolpsolve}
Given an $\al$-approximation algorithm for \knapo, we can compute in polytime a solution
$\bx$ to \eqref{cskolp} of objective value at least $\optcskolp/\al$.
%Also, we may assume that all constraints \eqref{nodecov} are tight for $(\bx,\by)$.
\end{lemma}

The proof of Lemma~\ref{cskolpsolve} involves a more-or-less routine application of the
ellipsoid method, but is somewhat long and technical. We therefore defer its proof to
Section~\ref{lpsolve}. 

%\vspace*{-2ex}
%\subparagraph*{LP rounding, and approximation algorithm.}
We now describe how to round the solution $\bx$ obtained by Lemma~\ref{cskolpsolve}
losing an $O(1)$-factor in the objective, and in the violation of constraints
\eqref{prefixlp}, and to obtain an $O(K)$-approximate non-adaptive policy for \csko from
the rounded solution.
The rounding algorithm is simply randomized rounding, and Chernoff bounds will
suffice to establish the above guarantee. Here is where we crucially exploit
property~\ref{size}. 
%(which we ensure by guessing many more portal vertices compared to~\cite{BansalN14}). 

\SetAlgoProcName{Algorithm}{Procedure}
\begin{procedure}[ht!]
\caption{CSKO-Alg(): \qquad 
\textnormal{// Rounding $(\bx,\by)$ and obtaining a non-adaptive policy}
\label{cskolpround} \label{cskoalg}}
\SetKwComment{simpc}{// }{}
\SetCommentSty{textnormal}
\DontPrintSemicolon

Independently, for each $a\in\fullset-\nd_k$, letting $b=\nxt(a)$, do the following: 
pick an $a$-$b$ path by choosing $\tau\in\I_a$ with probability $\bx^a_\tau/2$, and
choosing the ``direct'' path $a,b$ with the remaining probability
$0.5$; let $\rpath_{a,b}$ denote the path picked. \;
\label{randround}

If for any $j\in\dbrack{k}$, we have
$\sum_{h=0}^j\sum_{a\in\ndset_h-\nd_h}\mu^j\bigl(\rpath_{a,\nxt(a)}-\nxt(a)\bigr)>5(K+1)2^j$,
then \Return the empty policy that does not visit any node. \;
\label{prefixchk}

Consider the concatenated sequence of nodes $\{\rpath_{a,\nxt(a)}\}_{a\in\fullset-\nd_k}$ 
%retaining only one copy of each portal node %vertex in $\fullset$ 
(where $\fullset$ is ordered as in Theorem~\ref{strucdthm}). If a non-portal node
is repeated in this sequence, then shortcut the $\rpath_{a,\nxt(a)}$ paths so as to
retain only the first occurrence of each node. 
Let $\rpath'_{a,\nxt(a)}$ denote the shortcut version of $\rpath_{a,\nxt(a)}$ (which is still an
$a$-$\nxt(a)$ path). Let $\rpath'$ be the rooted path given by the node-sequence
$\{\rpath'_{a,\nxt(a)}\}_{a\in\fullset-\nd_k}$, where we retain only one copy of each
portal node. 
\;
\label{shortcut}

%Let $\rpath'$ denote the resulting rooted path. 
Sample each $v\in\rpath'-\rt$ independently with probability $\frac{1}{10(K+1)}$ 
to obtain the %node-sequence, and corresponding 
rooted path, $\rpath''$.
%Let 
%$\rpath''$ be the reultinng ordered sequence of nodes and corresponding rooted path.
\Return the non-adaptive policy $\rpath''$. 
\label{poloutput}
\end{procedure}

\vspace*{-2ex}
\subparagraph*{Analysis.}
We first give an overview.
The key observation is that since for any $j\in\dbrack{k}$, any $h\leq j$, any
$a\in\ndset_h-\nd_h$, 
and any $\tau\in\I_a$, we have $\mu^j\bigl(\tau-\nxt(a)\bigr)\leq 2^j$, we obtain that
$\sum_{h=0}^j\sum_{a\in\ndset_h-\nd_h}\mu^j\bigl(\rpath_{a,\nxt(a)}-\nxt(a)\bigr)$ is
the sum of a collection of independent {\em $2^j$-bounded} random variables,%
\footnote{This the key difference from~\cite{BansalN14}. They guess only the
$\nd_j$ nodes, and so in their case, the corresponding sum gets decomposed into the sum of
$(K+1)2^j$-bounded random variables, and so an application of Chernoff bounds incurs an
additional $\frac{\log k}{\log\log k}=\frac{\log\log W}{\log\log\log W}$-factor.}
%each bounded by $2^j$, and 
whose expectation is $O\bigl((K+1)\cdot 2^j\bigr)$, due to constraint \eqref{prefixlp}. 
It follows from Chernoff bounds that the
probability that this sum exceeds $5(K+1)2^j$, for any fixed index $j$, is
$\exp{-\Omega(K)}$, and so by a union bound, step~\ref{prefixchk} succeeds with high
probability (Lemma~\ref{prbad}).

To bound the reward obtained, %focus on the non-portal vertices (since the
%portal vertices contribute the same to the LP and the rounded solution). 
consider a node $v$ and index $j\in\dbrack{k}$, and 
define $\by^j_v:=\sum_{a\in\ndset_j-\nd_j}\sum_{\tau\in\I_a:v\in\tau-\nxt(a)}\bx^a_\tau$. 
(Note that $\sum_{h=0}^k\by^h_v\leq 1$.)
We say that %a non-portal vertex 
$v$ is ``visited by segment $j$'' if $v\neq\nd_j$ and
$v\in\bigcup_{a\in\ndset_j-\nd_j}\rpath_{a,\nxt(a)}$; we say that $v$ is 
``retained by segment $j$'' if $v\neq\nd_j$ and $v$ remains on
$\bigcup_{a\in\ndset_j-\nd_j}\rpath'_{a,\nxt(a)}$ after the shortcutting in
step~\ref{shortcut}. Note that the latter events are disjoint, for different $j$s. 
(Note that for a portal node in $\ndset_j-\nd_j$, both events happen with probability $1$.)   
Clearly, $v$ is retained by segment $j$ only if it is visited by segment $j$.
%for every $v\in V$ and $j\in\{0\}\cup[k]$. Note that $\by^j_a=1$
%
For convenience of analysis, we will view step~\ref{shortcut} as being executed even if 
step~\ref{prefixchk} fails, so we can talk about the event ``$v$ retained by segment $j$'' 
regardless of the outcome of step~\ref{prefixchk}.
%and independently of each other. 
It is not hard to argue that 
$\Pr[\text{$v$ is retained by segment $j$}]=\Omega(\by^j_v)$, but we need some care to
show that this holds even when we condition on the event that step~\ref{prefixchk}
succeeds, as subtle dependencies between events arise here. 
Nevertheless, we show that this indeed holds (Lemma~\ref{nodecovlem}).
%(Note that portal vertices always survive the shortcutting and belong to $\rpath'$, if
%step~\ref{prefixchk} succeeds.) 

Finally, given that step~\ref{prefixchk} succeeds and the rounded path $\rpath'$ satisfies
\eqref{prefixlp} with $O\bigl((K+1)2^j\bigr)$ on the RHS, due to the random sampling in 
step~\ref{poloutput}, we can argue that, for any node $v$ retained by segment
$j$, %or portal node $v\in\bigcup_{h=0}^j\ndset_h-\nd_j$, 
the non-adaptive policy processes $v$ by time $2^j-1$ with probability $\frac{1}{O(K)}$
(Lemma~\ref{nanodelem}). 
%shows that 
Thus, the expected reward of the non-adaptive policy is 
$\frac{1}{O(K)}\cdot\sum_{j=0}^k\sum_{v\in V}\by^v_j\cdot\pi_v(2^j-1)\geq\frac{\OPT}{O(K)}$. 
%\bigl(\text{objective value of }(\bx,\by)\bigr)\geq\frac{\OPT}{O(K)}$. 

We now delve into the details. %It will be convenient to define the following events. 
For an index $j\in\dbrack{k}$, let $\Bad_j$ be the event that 
$\sum_{h=0}^j\sum_{a\in\ndset_h-\nd_h}\mu^j\bigl(\rpath_{a,\nxt(a)}-\nxt(a)\bigr)>5(K+1)2^j$. 
So $\Bad:=\bigvee_{j=0}^k\Bad_j$ is the event that step~\ref{prefixchk} fails; let $\cBad$
denote the complement of $\Bad$. 
Recall that $K=3\log(6\log W)+12$.
The following basic claims will be useful.

\begin{claim} \label{mudec}
For any $j\geq 0$, and any $v\in V$, we have
$\frac{\mu^j_v}{2^j}\geq\frac{\mu^{j+1}_v}{2^{j+1}}$. 
\end{claim}

\begin{proof} %{Claim~\ref{mudec}}
This follows because $\E{\min\{\nsize,2^{j+1}\}}\leq
2\cdot\E{\min\{\nsize,2^j\}}$. 
\end{proof}

\begin{claim} \label{helper}
Let $\tht_1,\tht_2,\ldots,\tht_n\in[0,1]$. Then
(a) %If $\sum_{i\in[n]}\tht_i\leq 1$, then
$1-\prod_{i=1}^n(1-\tht_i)\geq(1-e^{-1})\min\bigl\{1,\sum_{i=1}^n\tht_i\bigr\}$.
(b) If $\sum_{i\in[n]}\tht_i\leq 0.5$, then $\prod_{i=1}^n(1-\tht_i)\geq 0.5$.
\end{claim}

\begin{proof} %{Claim~\ref{helper}}
For part (a), $1-\prod_{i=1}^n(1-\tht_i)\geq 1-e^{-\sum_i\tht_i}$. 
If $\sum_i\tht_i\geq 1$, then this is at least $1-e^{-1}$. Otherwise, by concavity of
$1-e^{-x}$, we have $1-e^{-\sum_i\tht_i}\geq (1-e^{-1})\sum_i\tht_i$.

For part (b), we use the fact that $1-x\geq 4^{-x}$ for all $x\in[0,0.5]$. 
\end{proof}

\begin{comment}
\begin{claim} \label{prefixbnd}
For any $j\in\dbrack{k}$, we have
$\sum_{h=0}^j\sum_{a\in\ndset_h-\nd_h}\sum_{\tau\in\I_a}\bx^a_\tau\mu^j\bigl(\tau-\nxt(a)\bigr)
=\sum_{h=0}^j\sum_{v\in V}\mu^j_vy^h_v$.
\end{claim}
\end{comment}

\begin{lemma} \label{prbad}
$\Pr[\Bad_j]\leq e^{-(K+1)}$ for all $j\in\dbrack{k}$.
Hence, $\Pr[\Bad]\leq 1/\poly(\log W)$.
\end{lemma}

\begin{proof}
%We bound $\Pr[\Bad_h]$ for any $h\in\dbrack{k}$ and use the union bound. 
The second statement follows from a straightforward union bound, since the number of
$j$ indices is at most $1+\log W$.
So fix an index $j\in\dbrack{k}$. %We upper bound $\Pr[\Bad_j]$ and then use the union bound.
For any $a\in\bigcup_{h=0}^j(\ndset_h-\nd_h)$, define the random variable 
$Z_a=\mu^j\bigl(\rpath_{a,\nxt(a)}-\nxt(a)\bigr)/2^j$. This is a $[0,1]$-random
variable, since if $a\in\ndset_h-\nd_h$, then since $\rpath_{a,\nxt(a)}\in\I_a$, we have 
$\mu^h\bigl(\rpath_{a,\nxt(a)}-\nxt(a)\bigr)\leq 2^h$ and so 
$\mu^j\bigl(\rpath_{a,\nxt(a)}-\nxt(a)\bigr)\leq 2^j$ (by Claim~\ref{mudec}).
Also, $\sum_{h=0}^j\sum_{a\in\ndset_h-\nd_h}\E{Z_a}\leq (K+1)$ due to \eqref{prefixlp}.
\begin{comment}
This follows from:
%
\begin{equation*}
\begin{split}
2^j\cdot\sum_{h=0}^j & \sum_{a\in\ndset_h-\nd_h}\E{Z_a} =
0.5\cdot\sum_{h=0}^j\sum_{a\in\ndset_h-\nd_h}\sum_{\tau\in\I_a}\bx^a_\tau\mu^j\bigl(\tau-\nxt(a)\bigr)
+0.5\cdot\sum_{h=0}^j\mu^j(\ndset_h-\nd_h) \\
& = 0.5\cdot\sum_{h=0}^j\sum_{v\in V}\mu^j_v\by^h_v
+0.5\cdot\sum_{h=0}^j\sum_{a\in\ndset_h-\nd_h}\mu^j_a\by^h_a 
\leq \sum_{h=0}^j\sum_{v\in V}\mu^j_v\by^h_v \leq (K+1)2^j
\end{split}
\end{equation*}
where we use Claim~\ref{prefixbnd} to obtain the second equality.
\end{comment}
%
%\begin{equation*}
%\begin{split}
%2^j\cdot\sum_{h=0}^j\sum_{a\in\ndset_h-\nd_h}\E{Z_a} & =
%0.5\cdot\sum_{h=0}^j\sum_{a\in\ndset_h-\nd_h}\sum_{\tau\in\I_a}\bx^a_\tau\sum_{v\in\tau-\nxt(a)}\mu^j_v
%+0.5\cdot\sum_{h=0}^j\sum_{a\in\ndset_h-\nd_h}\mu^j_a \\
%& = \sum_{h=0}^j\sum_{a\in\ndset_h-\nd_h}\mu^j_a
%+0.5\cdot\sum_{h=0}^j\sum_{v\in V-\fullset}\mu^j_v\cdot\sum_{a\in\ndset_h-\nd_h}\sum_{\tau\in\I_a:v\in\tau}\bx^a_\tau
%\\
%& = \sum_{h=0}^j\sum_{a\in\ndset_h-\nd_h}\mu^j_a
%+0.5\cdot\sum_{h=0}^j\sum_{v\in V-\fullset}\mu^j_v\cdot\by^h_v
%\leq \sum_{h=0}^j\sum_{v\in V}\mu^j_v\by^h_v \leq (K+1)2^j.
%\end{split}
%\end{equation*}
%
The $Z_a$s are independent $[0,1]$ random variables. 
So by Chernoff bounds,
we have 
$\Pr[\Bad_j]=\Pr[\sum_{h=0}^j\sum_{a\in\ndset_h-\nd_h}Z_a>5(K+1)] 
\leq\bigl(\frac{e^4}{5^5}\bigr)^{K+1}\leq e^{-(K+1)}$. 
%(Theorem~\ref{chernoff}), 
%we have $\Pr[\Bad_j]\leq\bigl(\frac{e^4}{5^5}\bigr)^{K+1}\leq e^{-(K+1)}$. 
%It follows that $\Pr[\Bad]\leq e^{-(K+1)+\log\log (2W)}\leq 1/\poly(\log W)$.
\end{proof}

\begin{lemma} \label{nodecovlem}
%(a) If $v$ is a portal vertex, $\Pr[v\in\rpath']=\Pr[\cBad]\geq $; \ 
%(b) For a non-portal node 
For any node $v\in V$ and any $j\in\dbrack{k}$, we have
$\Pr[\{\text{$v$ is retained by segment $j$}\}\wedge\cBad]\geq\frac{\by^j_v}{16}$.
\end{lemma}

\begin{proof}
If $v\in\ndset_j-\nd_j$, then this probability is simply $\Pr[\cBad]$, and the statement
follows from Lemma~\ref{prbad}. So suppose $v\notin\fullset$.
(Observe then that $v\in\tau$ for some $\tau\in\I_a$, $a\in\fullset-\nd_k$ iff 
$v\in\tau-\nxt(a)$.) 
Let $\Vis_h$ denote the event ``$v$ is visited by segment $h$'' for $h\in\dbrack{k}$.
Let $\Ret$ denote the event ``$v$ is retained by segment $j$.''
As noted earlier, we need to proceed somewhat carefully here 
%a bit careful in bounding this probability 
because the events $\Ret$ and $\cBad$ are not independent; moreover upper bounding the
conditional probability $\Pr[\Bad_f\,|\,\Ret]$ via Chernoff bounds is problematic, because
this conditioning can cause the $Z_a$ random variables used in the proof of
Lemma~\ref{prbad} to no longer be independent.
We proceed as follows. 
First, note that
$$
\Pr[\Vis_j]=\Bigl(1-\prod_{a\in\ndset_j-\nd_j}\bigl(1-0.5\sum_{\tau\in\I_a:v\in\tau}\bx^a_\tau\bigr)\Bigr)
\geq
\bigl(1-e^{-1}\bigr)\cdot\sum_{a\in\ndset_j-\nd_j}\sum_{\tau\in\I_a:v\in\tau}0.5\bx^a_\tau
\geq\frac{\by^j_v}{4}.
$$
%where the first inequality is because 
%$\sum_{a\in\ndset_j-\nd_j}\sum_{\tau\in\I_a:v\in\I_a}x^a_\tau=y^j_v\leq 1$.
We upper bound $\Pr[\Bad_f\,|\,\Vis_j]$ for every $f\in\dbrack{k}$, and
obtain $\Pr[\Bad\,|\,\Vis_j]\leq 1/\poly(\log W)$. As we shall see, here we can still proceed by
defining suitable independent random variables and applying Chernoff bounds.
We also show that $\Pr[\Ret\,|\,\Vis_j]\geq 0.5$. 
We then obtain that
\begin{equation*}
\begin{split}
\Pr[\Ret\wedge\cBad] & = \Pr[\Ret\wedge\Vis_j\wedge\cBad]
=\Pr[\Vis_j\wedge\Ret]-\Pr[\Vis_j\wedge\Ret\wedge\Bad]
\geq \Pr[\Vis_j\wedge\Ret]-\Pr[\Vis_j\wedge\Bad] \\
& = \Pr[\Vis_j]\cdot\Bigl(\Pr[\Ret\,|\,\Vis_j]-\Pr[\Bad\,|\,\Vis_j]\Bigr)
\geq\frac{\by^j_v}{16}.
\end{split}
\end{equation*}
We now obtain the stated bounds on $\Pr[\Bad\,|\,\Vis_j]$ and $\Pr[\Ret\,|\,\Vis_j]$. 
For the latter, note that conditioned on $\Vis_j$, $v$ is retained by segment $j$ if $v$ is 
not visited by any segment $h<j$. We have 
$\Pr[\Vis_h]\leq 0.5\sum_{a\in\ndset_h-\nd_h}\sum_{\tau\in\I_a:v\in\tau}\bx^a_\tau\leq 0.5\by^h_v$,
for any index $h$. So $\Pr[\Ret\,|\,\Vis_j]=\prod_{h=0}^{j-1}(1-\Pr[\Vis_h])\geq 0.5$ by
Claim~\ref{helper}, since $\sum_{h=0}^{j-1}\Pr[\Vis_h]\leq 0.5$.

Finally, to bound $\Pr[\Bad_f\,|\,\Vis_j]$, we first note that if $f<j$, then $\Bad_f$ and
$\Vis_j$ are independent, and the upper bound follows from Lemma~\ref{prbad}. 
So suppose $f\geq j$.
As in the proof of Lemma~\ref{prbad}, 
define $Z_a=\mu^f\bigl(\rpath_{a,\nxt(a)}-\nxt(a)\bigr)/2^f$ for all
$a\in\bigcup_{h=0}^f(\ndset_h-\nd_h)$.
So $\Pr[\Bad_f\,|\,\Vis_j]=\Pr\bigl[\sum_{h=0}^f\sum_{a\in\ndset_h-\nd_h}Z_a>5(K+1)\,|\,\Vis_j\bigr]$. 
Observe that $\sum_{a\in\ndset_j-\nd_j}Z_a\leq|\ndset_j|\leq 2(K+1)$. So 
\begin{equation*}
\begin{split}
\Pr\Bigl[\sum_{h=0}^f\sum_{a\in\ndset_h-\nd_h}Z_a>5(K+1)\,|\,\Vis_j\Bigr]
& \leq \Pr\Bigl[\sum_{0\leq h\leq f: h\neq j}\sum_{a\in\ndset_h-\nd_h}Z_a>3(K+1)\,|\,\Vis_j\Bigr] 
\\
& = \Pr\Bigl[\sum_{0\leq h\leq f: h\neq j}\sum_{a\in\ndset_h-\nd_h}Z_a>3(K+1)\Bigr]
\end{split}
\end{equation*}
where we can remove the conditioning on $\Vis_j$ in the last expression since for $h\neq j$ and
$a\in\ndset_h-\nd_h$, the random variable $Z_a$ is independent of $\Vis_j$.
Now we can proceed using Chernoff bounds as in the proof of Lemma~\ref{prbad} to obtain 
that $\Pr[\Bad_f|\Vis_j]\leq\bigl(\frac{e^2}{27}\bigr)^{K+1}\leq e^{-(K+1)}$ and hence 
$\Pr[\Bad\,|\,\Vis_j|\leq 1/\poly(\log W)$. 
\end{proof}

\begin{lemma} \label{nanodelem}
Consider any node $v\in V-\nd_k$. 
Suppose that $v$ is retained by segment $j$ in step~\ref{shortcut}. 
Then 
$\Pr[\text{non-adaptive policy $\rpath''$ processes $v$ by time $2^j-1$}]\geq\frac{1}{20(K+1)}$, 
where the probability is over both the random sampling in step~\ref{poloutput} to obtain
$\rpath''$ and the random execution of $\rpath''$.
\end{lemma}

\begin{proof}
Recall that $\mu^j_w=\E{X^j_w}$ and $X^j_w=\min\{\nsize[w],2^j\}$ for every node $w$.
The given probability is 
$$
\Pr\bigl[v\in\rpath''\bigr]\cdot
\Pr\bigl[\text{total processing time of nodes in $\rpath''$ before $v$}\leq 2^j-1\,|\,v\in\rpath''\bigr]
$$
where, throughout, all probabilities are conditioned on the state after
step~\ref{shortcut}. Since $v$ is retained by segment $j$, in particular, we have
$v\in\rpath'$; so $\Pr\bigl[v\in\rpath''\bigr]=\frac{1}{10(K+1)}$.

Let $T$ be the set of nodes in $\rpath'$ before $v$, and $A\sse T$ be the random set of
nodes from $T$ included in $\rpath''$. (Note that $T$ is fixed since we are conditioning on
the state after step~\ref{shortcut}.) Note that $A$ is independent of the event
$\{v\in\rpath''\}$. So we have
\begin{equation*}
\begin{split}
\Pr\bigl[\text{total} &\text{ processing time of nodes in $\rpath''$ before $v$}\leq 2^j-1\,|\,v\in\rpath''\bigr]
= \Pr\Bigl[\sum_{w\in A}\nsize[w]\leq 2^j-1\Bigr] \\
& \geq 1-\Pr\Bigl[\sum_{w\in A}\nsize[w]\geq 2^j\Bigr]
= 1-\Pr\Bigl[\sum_{w\in A}\min\{\nsize[w],2^j\}\geq 2^j\Bigr]
\geq 1-\frac{\E{\sum_{w\in A}X^j_w}}{2^j}.
\end{split}
\end{equation*}
Since step~\ref{prefixchk} succeeds, we know that $\mu^j(T)\leq 5(K+1)2^j$. Since
$\Pr[w\in A]=\frac{1}{10(K+1)}$ for each $w\in T$, we obtain that 
$\E{\sum_{w\in A}X^j_w}=\frac{\mu^j(T)}{10(K+1)}\leq 2^{j-1}$.
%we have $\bigcup_{h=0}^j\bigcup_{a\in\ndset_h-\nd_h}\bigl(\rpath_{a,\nxt(a)}-\nxt(a)\bigr)$,
Putting everything together, we obtain the desired statement.
\end{proof}

\begin{proofof}{Theorem~\ref{approx-cskothm}}
%Similar to the proof of Lemma~\ref{nanodelem}, the expected reward obtained from node
%$\nd_k$ is at least $\frac{\pi_k(2^{k+1}-1)}{20(K+1)}$.
Combining Lemmas~\ref{nodecovlem} and~\ref{nanodelem}, and since for any $v\in V-\nd_k$, the
events ``$v$ is retained by segment $j$'' are disjoint across different $j$s, the 
expected reward obtained from a node $v$ is at least $\frac{\sum_{j=0}^k\pi_v(2^j-1)\by^j_v}{320(K+1)}$.
So the total expected reward %across all nodes 
obtained by $\rpath''$ 
%the non-adaptive policy returned by Algorithm~\ref{cskoalg} 
is at least 
$\frac{1}{320(K+1)}\cdot\bigl(\text{objective value of $\bx$}\bigr)=\OPT/O(K)$.

The running time is polynomial in the time needed to enumerate the quantities in
Theorem~\ref{strucdthm}, which is 
$\poly\bigl((n\log B)^{O(\log W\log\log W)}\bigr)=O\bigl((n+\log B)^{O(\log W\log\log W)}\bigr)$.
\end{proofof}

\subsection{Polynomial-time \boldmath $O(\log W)$-approximation algorithm}
%  Theorem~\ref{poly-approxcsko}} 
\label{csko-poly}
The polytime algorithm also proceeds by gleaning some structural insights from
%the decision tree $\T$ representing 
an optimal adaptive policy that enable one to reduce the problem to rooted knapsack
orienteering, losing an $O(\log W)$-factor. 
%Roughly speaking, %similar to $\nd_j$ portal vertices, we argue that 
%from the decision tree $\T$ representing an optimal adaptive policy, 
Recall that $L=\ceil{\log W}$.
It is not hard to see that with a factor-$L$ loss, we can focus on vertices whose
processing starts in the interval $[2^j-1,2^{j+1}-1)$, for some index
$j\in\dbrack{L}$. Now we consider a (rooted) \knapo-instance with start node $\rt$, 
travel budget $B$, and knapsack budget $2^{j+1}$, where the knapsack weight of each vertex 
$v$ is the expected truncated size $\mu^j_v:=\E{\min\{\Sv,2^j\}}$ and its reward is $\pi_v(2^j-1)$.
%the reward of a vertex $v$ is the expected reward obtained when 
We show that the optimal adaptive policy yields a {\em fractional solution to the
LP-relaxation \eqref{knapo-lp} for this \knapo-instance}, of objective value
$\OPT/O(L)$. Invoking Theorem~\ref{knapo-roundthm}, we can then obtain a \knapo solution
of objective value $\OPT/O(L)$, and we show that the resulting rooted path can be utilized
to obtain an $O(L)$-approximation non-adaptive policy.

\begin{theorem} \label{knaporedn}
There exists an index $j\in\dbrack{L}$ such that, for the \knapo-instance with start
node $\rt$, travel budget $B$, knapsack budget $2^{j+1}$, knapsack weights  
$\{\mu^j_v\}_{v\in V}$, and rewards $\{\pi_v(2^j-1)\}_{v\in V}$, the optimal value of the
LP-relaxation \eqref{knapo-lp}, is at least $\OPT/(L+1)$.
\end{theorem}

\begin{proof}
Let $\T$ be the decision tree corresponding to an optimal adaptive policy.
For clarity, in this proof, nodes will always refer to the tree $\T$, and vertices will
refer to $V$. 
For nodes $u,v\in\T$, let $u\prec v$ denote that $u$ is an ancestor of $v$ in
$\mathcal{T}$. 
For a node $v\in\T$, the probability of reaching $v$ is the probability that the
execution of the adaptive policy follows the $\rho\leadsto v$-path in $\T$; 
%and is equal to the product of the probabilities labelling the edges of this rooted path; 
let $\iv:=\sum_{u\prec v}\Su$ denote the total size observed before reaching $v$, which is
also the time when $v$ is processed. 
%(Note that $\iv$ is a deterministic quantity, since
%we are looking at a node of $\T$.)
%in $\mathcal{T}$.
%and (iii) define the reward obtained from $v$ to be the expected reward obtained from $v$
%by the policy given that $v$ is reached, which is equal to $\pi_v(\iv)$.
Clearly, $\OPT=\sum_{v\in\T}\Pr[\text{$v$ is reached}]\cdot\pi_v(\iv)$. 
%is equal to the total reward of nodes in $\T$.
%Recall that
%for $v\in V$, and any $j\geq 0$, $\Xvj:=\min\{\Sv,2^j\}$ and $\muvj:=\Ex[\Xvj]$.  
%For any
%$v\in V$ and $j\in \{0,1,\hdots,\lceil \log W\rceil\}$ let $\Xvj:=\min\{S_v,2^j\}$ and
%$\muvj:=\Ex[\Xvj]$.  
%The following inequality is crucial in the next section:
%\begin{equation}
 %   \frac{\muvj}{2^j}\leq \frac{\mu_v^h}{2^h},\;\forall j\geq h.
%\end{equation}

For each index $\ell\in\dbrack{L}$, define
$\Nc_\ell:=\{v\in\T: \iv\in[2^\ell-1,2^{\ell+1}-1)\}$. Clearly, the $\Nc_\ell$s partition
the node-set of $\T$ (though some of these sets could be $\es$), 
%and the expected reward obtained by the adaptive policy is the total 
so there is some index $j$ such that 
$\sum_{v\in\Nc_j}\Pr[\text{$v$ is reached}]\cdot\pi_v(\iv)\geq\OPT/(L+1)$.
Let $\T_j$ be the subgraph of $\T$ induced by nodes in $\Nc_j$. 

Now consider the LP-relaxation \eqref{kolp} for the rooted \knapo-instance described in
the theorem statement. We map $\T_j$ 
to a solution to \eqref{kolp} as follows. Let $\Pc_j$ denote the collection of rooted
paths obtained by prepending $\rt$ to each maximal path of $\T_j$. Each
$Q\in\Pc_j$ corresponds to a rooted path in the metric space. For any such path $Q$,
%For $Q\in\Pc_j$, 
let $\Pr[Q]$ be the probability that the execution of the adaptive policy follows the 
rooted path in $\T$ that has $Q-\rt$ as its suffix.
%which is the product of the probabilities labelling the edges of the rooted path in $\T$
%that has $Q-\rt$ as its suffix; 
%let $\vfar(Q)$ be the furthest node from $\rt$ on $Q$. 
Note that $\sum_{Q\in\Pc_j}\Pr[Q]\leq 1$. 
%Note that $\sum_{Q\in\Pc_j}\Pr[Q]\leq 1$.
%Note that each node of $Q$ is labeled by a distinct vertex of $V$, so we
%can view $Q$ as a path in the metric space.
For each $Q\in\Pc_j$, consider the $\{0,1\}$-solution, where 
letting $v=\argmax_{u\in Q}d(\rt,u)$, %be the furthest node from $\rt$ on $Q$, 
the $x^v_a$ and $z^v_u$
variables encode the indicator vectors of the edge-set and vertex-set of $Q$
respectively, and all other variables are set to $0$; i.e., take $\bx^{Q,v}=\chi^{E(Q)}$,
$\bz^{Q,v}=\chi^{V(Q)}$, and $\bx^{Q,w}=\bz^{Q,w}=0$ for all $w\neq v$. 
Observe that this %yields a feasible solution to \eqref{kolp}. 
satisfies the orienteering constraints \eqref{pref-visit}--\eqref{unit} by construction,
since $Q$ is a rooted path and $v$ is the furthest node from $\rt$ on $Q$. %=\vfar(Q)$. 

We claim that the weighted sum of the $(\bx^Q,\bz^Q)$ solutions, using the $\Pr[Q]$ 
probabilities as weights, also satisfies the knapsack constraint \eqref{knbudget}, and is 
thus a feasible solution   
to \eqref{kolp}. We prove the claim shortly, but first show how this yields the theorem.
Let $x=\sum_{Q\in\Pc_j}\Pr[Q]\cdot\bx^Q$ and $z=\sum_{Q\in\Pc_j}\Pr[Q]\cdot\bz^Q$. 
For a vertex $u\in V$, let $T_u$ be the random variable denoting the time when $u$ is
processed by the adaptive policy; we set $T_u:=\infty$ if $u$ is not processed.%
\footnote{Note that $T_u$ is similar to $i_u$, except that $i_u$ is defined for a 
{\em node} $u$ of the adaptive-policy tree $\T$, and is therefore a deterministic quantity.}
For every $u\in V$, define $z_u=\sum_{v\in V}z^v_u$, and note that
$z_u=\Pr\bigl[T_u\in[2^j-1,2^{j+1}-1)\bigr]$. 
The objective value of $(x,z)$ is 
$\sum_{u\in V}\Pr\bigl[T_u\in[2^j-1,2^{j+1}-1)\bigr]\cdot\pi_u(2^j-1)
\geq\sum_{w\in\Nc_j}\Pr[\text{$w$ is reached}]\cdot\pi_w(i_w)\geq\OPT/(L+1)$. 
%$\sum_{Q\in\Pc_j}\Pr[Q]\cdot\sum_{w\in Q}\pi_w(2^j-1)\geq\sum_{Q\in\Pc_j}\Pr[Q]\pi_w(i_w)
%\sum_{w\in\Nc_j}\Pr[\text{$w$ is reached}]\cdot\pi_w(\iw)\geq\OPT/(L+1)$. 

We now prove the claim to finish the proof.
We need to show that 
$\sum_{u\in V}\mu^j_u\cdot\Pr\bigl[T_u\in [2^j-1,2^{j+1}-1)\bigr]\leq 2^{j+1}$.
This follows from the same type of argument as used for stochastic
knapsack~\cite{DeanGV08,GuptaKMR11}. 
Recall that $X^j_u=\min\{\Su,2^j\}$ for any $u\in V$. 
%Let $Z$ be the random rooted path resulting from the execution of the adaptive policy,
%and 
Let $\vl$ be the random vertex processed last by the adaptive policy at some time in
$[2^j-1,2^{j+1}-1)$. 
%The total size of vertices other than $\vl$ that are processed in
%$[2^j-1,2^{j+1}-1)$ can be at most $2^j$
We have 
$\sum_{u\in V: u\neq\vl}\Su\cdot\bon_{T_u\in[2^j-1,2^{j+1}-1)}\leq 2^j$, so
$\sum_{u\in V}X^j_u\cdot\bon_{T_u\in[2^j-1,2^{j+1}-1)}\leq 2^{j+1}$.
Taking expectations, we obtain that
$\sum_{u\in V}\E{X^j_u\,|\,T_u\in[2^j-1,2^{j+1}-1)}\cdot
\Pr\bigl[T_u\in[2^j-1,2^{j+1}-1)\bigr]\leq 2^{j+1}$.
This proves the claim since for any $u\in V$, we have  
$z_u=\Pr\bigl[T_u\in[2^j-1,2^{j+1}-1)\bigr]$, and $X^j_u$ is independent of $T_u$, so 
$\E{X^j_u\,|\,T_u\in[2^j-1,2^{j+1}-1)}=\mu^j_u$.
\end{proof}

\begin{proofof}{Theorem~\ref{poly-approxcsko}}
Theorem~\ref{knaporedn} leads to the following simple algorithm. For the index $j$ in the
theorem statement, we solve \eqref{kolp} and round it to an {\em integer solution} $P$
losing a factor of $5$ (see Theorem~\ref{knapo-round}). We sample each non-root node in $P$
independently with probability $\frac{1}{4}$, and return the resulting rooted path $P''$. 
To analyze this, for any $v\in P$, we have that the probability that the non-adaptive
policy $P''$ processes $v$ by time $2^j-1$ is at least $\frac{1}{8}$. 
%Recall that $X^j_w=\min\{\Sw,2^j\}$ for any $w\in V$. 
The claim follows because 
$\Pr\bigl[\sum_{w\prec_{P''}v}\Sw\geq 2^j\bigr]=\Pr\bigl[\sum_{w\prec_{P''}v}X^j_w\geq 2^j\bigr]$,
which is at most 
\begin{equation*}
%\begin{split}
\frac{\E{\sum_{w\prec_{P''}v}X^j_w}}{2^j}
=\frac{1}{4}\cdot\frac{\E{\sum_{w\prec_{P}v} X^j_w}}{2^j}
\leq\frac{1}{4}\cdot\frac{\sum_{w\in P}\mu^j_w}{2^j}\leq\frac{1}{2}.
%\end{split}
\end{equation*}
The probability of the stated event is therefore at least 
$\Pr[v\in P'']/2\geq 1/8$. Therefore, the expected reward obtained is at least 
$\frac{1}{8}\cdot\sum_{v\in P}\pi_v(2^j-1)\geq\frac{\OPT}{L+1}\cdot\frac{1}{5}\cdot\frac{1}{8}=\OPT/O(L)$.
\end{proofof}

\begin{comment}
\begin{remark}
We point out that the idea of extracting a good \knapo-LP solution from a policy can also be 
leveraged to obtain a -approximation for nonadaptive stochastic orienteering, which
substantially improves upon the approximation factor of $500$ obtained
by~\cite{GuptaKNR}. Gupta et al.~\cite{GuptaKNR} in fact remark in this context that
``obtaining a significantly smaller constant factor seems to require additional
techniques.''
\end{remark}
\end{comment}

\subsection*{Proof of Theorem~\ref{structhm}} \label{append-structhm}

Recall that $L=\ceil{\log W}$, and $K=3\log(6\log W)+12$.
We will prove the following more-compact statement.

\begin{theorem} \label{simpstructhm} 
There exists a rooted path $\spath$ with $d(\spath)\leq B$, and 
vertices $\nd_0\preceq\nd_1\preceq\ldots\preceq\nd_k$ on $\spath$ for some $k\leq L$, such that: 
\begin{enumerate}[label=(\alph*), topsep=0.1ex, noitemsep, leftmargin=*]
\item \label{simprewd} %(Reward) 
$\sum_{j=0}^k\sum_{v\in\spath_{\nd_{j-1},\nd_j}-\nd_j}\pi_v(2^j-1)\geq\OPT/4$; and  
\item \label{simpprefsize} %(Prefix-size) 
$\mu^j(\spath_{\rt,\nd_j}-{\nd_j})\leq (K+1)2^j$ for all $j\in\dbrack{k}$.
\end{enumerate} 
\end{theorem}  

We first show that Theorem~\ref{simpstructhm} easily implies Theorem~\ref{structhm},
since we can simply subdivide each $\spath_{\nd_{j-1},\nd_j}$ path into at most $2(K+1)$ subpaths
satisfying having $\mu^j$-weight at most $2^j$. %part (c) of Theorem~\ref{structhm}. 
More precisely, for each $j\in\dbrack{k}$, let $P=\spath_{\nd_{j-1},\nd_j}-\nd_j$. 
So we have $\mu^j(P)\leq (K+1)2^j$ due to part~\ref{simpprefsize} above. Also, we have 
$\mu^j_v\leq 2^j$ for all $v\in V$.
Now we greedily split $P$ into at most $K+1$ segments, where each segment $Z$
is a minimal contiguous subsequence of the node-sequence of $P$ satisfying
$\mu^j(Z)\geq 2^j$; clearly, if $v$ is the last node of $Z$, then we have $\mu^j(Z-v)\leq 2^j$.
%If $\mu^j(Z)>2^j$, we can in turn split $Z$ into at most two segments,
%each of $\mu^j$-weight at most $2^j$: one comprising all nodes of $Z$ except the end-node,
%and one comprising the end-node of $Z$.
%
\begin{comment}
Now greedily split $P$ as follows: initialize $i=1$ and $v_1=\nd_{j-1}$. 
If $v_{2i-1}\neq\nd_j$, set $v_{2i}$ to be the first node $w$ on $P$ after
$v_{2i-1}$ such that $\mu^j(P_{v_{2i-1},w})>2^j$; if such a $w$ does not exist, then set
$v_{2i}=\nd_j$. If $v_{2i}\neq\nd_j$, then let $v_{2i+1}$ be the node right after $v_{2i}$ on $P$, advance
the counter $i$ and repeat; otherwise end the loop. 
%If there is no such node $w$, then set $v_{2i}=\nd_j$, and end the loop. 
Let $i^*$ be the final value of the counter $i$. Note that $i^*\leq K+1$, since each time
the counter $i$ is incremented, we discard $P_{v_{2i-1},v_{2i}}$ and %whose total
$\mu^j(P_{v_{2i-1},v_{2i}})>2^j$ by construction. 
Also, note that $\mu^j(P_{v_{2i-1},v_{2i}}-v_{2i})\leq 2^j$ for all $i\in[i^*]$, and
$\mu^j(P_{v_{2i},v_{2i+1}}-v_{2i+1})=\mu^j_{v_{2i}}\leq 2^j$, for all $i\in[i^*-1]$.
\end{comment}
%
Let $\ndset_j$ consist of the end-points of all the $Z$ segments so obtained.
%all the designated $v_i$-nodes created by this process, 
Clearly, we have $\{\nd_{j-1},\nd_j\}\sse\ndset_j$, $|\ndset_j|\leq 2(K+1)$. 

\medskip
In the sequel, we focus on proving Theorem~\ref{simpstructhm}.
We follow a similar strategy as in~\cite{BansalN14} and show that $\spath$ 
%a suitable path
can be extracted from the decision tree $\T$ representing an optimal adaptive policy. We
also argue that this immediately implies an $O(K)$ upper bound on the adaptivity gap,
albeit in a non-constructive fashion.
%\subsection{Upper bound on the adaptivity gap of \csko}
%In this section we show an upper bound of $\mathcal{O}(\log\log W)$ on the adaptivity gap
%of \csko. We follow the proof of the $\mathcal{O}(\log\log B)$ upper bound on the
%adaptivity gap for correlated stochastic orienteering, presented in \cite{BansalN14}. 
%Let $\mathcal{T}$ denote the decision tree for an optimal adaptive policy. We make the
%following remarks about $\mathcal{T}$: 
%Observe that: 
As noted earlier: (i) $\T$ is rooted at $\rho$; (ii) each node of $\T$ is labeled by a vertex
of $V$; and (iii) branches of a node of $\T$ labeled $v\in V$ correspond to the
different size and reward instantiations of $v$, and we consider each edge from the node
to its children as being labeled by the probability of the corresponding instantiation.
Also, note that while different nodes of $\T$ may be labeled by the same vertex of $V$, on
any rooted path of $\T$, each node is labeled by a distinct vertex of $V$. 
Due to (ii), when refer to a node $v\in\T$, we mean that the node is labeled $v\in V$.
For clarity, in this section, nodes will always refer to the tree $\T$, and vertices will
refer to $V$. 
%We consider each edge from a node $v\in\T$ to its
%children as being labeled by the probability of the corresponding realization of $v$.
%\begin{itemize}
%    \item Its nodes are labelled by the vertices of $G$.
%    \item It is rooted at $\rho$.
%    \item Its branches at a certain node $v$ correspond to the different size and reward instantiations of $v$.
%    \item Any vertex $v\in V$ may appear on multiple nodes of $\mathcal{T}$, but $v$ can appear at most once on any r%ooted path of $\mathcal{T}$.
%\end{itemize}
For any two nodes $u,v\in\mathcal{T}$ we use $u\prec v$ to denote that $u$ is an ancestor
of $v$ in $\mathcal{T}$. Recall that for $v\in V$, and any $j\geq 0$, $\Xvj:=\min\{\Sv,2^j\}$ and
$\muvj:=\Ex[\Xvj]$.  
For any $v\in\mathcal{T}$, let $\iv=\sum_{u\prec v}\Su$ denote the total size
observed before reaching $v$ in $\mathcal{T}$.
%(b) define the reward obtained from $v$ to be the expected reward obtained from $v$
%by the policy, given that $v$ is reached, which is equal to $\pi_v(\iv)$; and
%For any
%$v\in V$ and $j\in \{0,1,\hdots,\lceil \log W\rceil\}$ let $\Xvj:=\min\{S_v,2^j\}$ and
%$\muvj:=\Ex[\Xvj]$.  
%The following inequality is crucial in the next section:
%\begin{equation}
 %   \frac{\muvj}{2^j}\leq \frac{\mu_v^h}{2^h},\;\forall j\geq h.
%\end{equation}
The probability of reaching a node $v\in\T$ 
is the probability that the execution of the
adaptive policy follows the $\rho\leadsto v$-rooted path in $\T$, and is equal to the
product of the probabilities labeling the edges of this rooted path.

Lemma~\ref{lemmaconditions} below is from~\cite{BansalN14}, and as therein, leads to
Lemma~\ref{pathle}.  

\begin{lemma}[Lemma 3.2 in \cite{BansalN14}] \label{lemmaconditions}
Consider any $M\geq 12$, and any $j\in\dbrack{L}$. 
%and $j\in\{0,1,...,\lceil \log{W}\rceil\}$. Then, 
The probability of reaching a node $u\in\mathcal{T}$ with 
$\sum_{w\preceq u}X^j_w\leq 2^{j+1}$, and $\sum_{w\preceq u}\muwj> M\cdot 2^j$ is at most $e^{-M/3}$.
%satisfying the following:
%     \begin{itemize}
%        \item $\sum_{w\preceq u}X^j_w\leq 2^{j+1}$, and
%        \item $\sum_{w\preceq u}\muwj> K\cdot 2^j$
%    \end{itemize}
%    is at most $e^{-K/3}$.
\end{lemma}

%As in \cite{BansalN14}, the above lemma leads to the following.

\begin{lemma} \label{pathle} \label{pathlem}
%Assume that $K\geq 3\cdot \log{(6\log{W})}+12$. Then, 
There exists some node $s\in\mathcal{T}$ such that the $\rho\leadsto s$-path $\sigma$ in
$\T$ %from the root to $s$ in $\mathcal{T}$ 
satisfies the following:
%\begin{itemize}
    %\item 
(a) $\sum_{v\in \sigma}\piv(\iv)\geq OPT/2$; and
    %\item 
(b) for each $v\in\sigma$ and $j\in\dbrack{L}$, we have 
$\sum_{w\preceq v}X^j_w>2^{j+1}$ or $\sum_{w\preceq v}\muwj\leq K\cdot 2^j$.
%\end{itemize}
\end{lemma}

\begin{proof} %of}{Lemma~\ref{pathlem}}
For each $j\in\dbrack{L}$, let $A_j$ be the nodes that satisfy both
the conditions in Lemma~\ref{lemmaconditions} for $j$. 
Let $A=\bigcup_{j=0}^LA_j$.
Let $F$ be the set of nodes $v\in A$ such that no ancestor of $v$ is in $A$, i.e.,
$F=\{v\in A: \nexists w\in A\text{ s.t. }w\prec v\}$.
By a union bound, the probability of reaching any node in $A$ %$\bigcup_{j=0}^L N_j$ 
is at most $(L+1)e^{-K/3}\leq \frac{1}{2}$. 
Let $\T^\prime$ be the component of $\T-F$ containing $\rho$. %obtained after deleting $F$.
For each $v\in F$, the expected reward obtained from the subtree of $\T_v$ rooted at
$v$ is at most $\Pr[v\text{ is reached}]\cdot\OPT$, since the adaptive policy that goes
from $\rho$ directly to $v$ and then follows $\T_v$, can gather expected reward at most
$\OPT$. Therefore, the total expected reward
obtained from subtrees of $\T$ rooted at nodes in $F$ is at most $\OPT/2$. 
%all nodes in $A$. 
%in containing only the nodes 
%that have no predecessors, nor are included, in $\cup_{j=0}^{\lceil\log{W}\rceil}N_j$.
%By the optimality of $\T$, the expected reward obtained after reaching any node of $\T$
%can be at most $\OPT$.  
It follows that the expected reward of $\T^\prime$ is at least $OPT/2$.
%The proof follows by a standard averaging argument over the possible paths that can be
%traversed in $\T^\prime$. 

The expected reward of $\T^\prime$, which is 
$\sum_{v\in\T^\prime}\Pr[v\text{ is reached}]\cdot\pi_v(\iv)$, 
can also be expressed as
$\sum_{w\text{ leaf of }\T^\prime}\Pr[w\text{ is reached}]\cdot\sum_{v\preceq w}\pi_v(i_v)$. 
Since $\sum_{w\text{ leaf of }\T^\prime}\Pr[w\text{ is reached}]\leq 1$, there is some
leaf $s$ of $\T^\prime$ for which $\sum_{v\preceq s}\pi_v(i_v)\geq\OPT/2$.
So taking $\sg$ to be the $\rho\leadsto s$-path of $\T^\prime$, part (a) holds.
Part (b) follows because for every node $v\in\T^\prime$, one of the conditions of
Lemma~\ref{lemmaconditions} does not hold. 
\end{proof}

\begin{proof}[Finishing up the proof of Theorem~\ref{simpstructhm}]
%We prove Lemma~\ref{pathlem} shortly, but first show that this yields
%Theorem~\ref{simpstructhm}.
Let $\spath$ be the path $\sg$ given by Lemma~\ref{pathlem}. Since $\sg$ is a valid
execution of the adaptive policy, we have $d(\spath)\leq B$, and $i_s\leq W$.
For any integer $j\geq 0$, define $\nd_j$ to be the first node $v$ on $\spath$ such that
$\iv\geq 2^{j+1}-1$, if such a node exists. Let $k$ be the smallest integer for which such
a node does not exist; set $\nd_k=s$, and note that $k\leq\ceil{\log W}$. 

For any $j\in\dbrack{k}$ and any $v\in\spath_{\nd_{j-1},\nd_j}-\nd_j$, we have 
$i_v\geq 2^j-1$, so $\pi_v(2^j-1)\geq\pi_v(i_v)$. 
Therefore, by Lemma~\ref{pathlem} (a) and since $\pi_{\nd_k}(0)\leq\OPT/4$, we obtain that
$\sum_{j=0}^k\sum_{v\in\spath_{\nd_{j-1},\nd_j}-\nd_j}\pi_v(2^j-1)\geq\sum_{v\in\sg}\pi_v(i_v)-\OPT/4\geq\OPT/4$.
This proves part (a) of the theorem statement.
For part (b), fix any $j\in\dbrack{k}$, and let $w, v$ be the nodes immediately
preceding $\nd_j$ on $\spath$, in that order. Then, 
$i_v=\sum_{v'\preceq w}X^j_{v'}<2^{j+1}-1$ and so we must have  
$\sum_{v'\preceq w}\mu_{v'}^j\leq K\cdot 2^j$, by part (b) of Lemma~\ref{pathlem}. It
follows that $\mu^j(\spath_{\rho\nd_j}-\nd_j)=\sum_{v'\preceq w}\mu_{v'}^j+\mu^j_v\leq(K+1)\cdot 2^j$.
\end{proof}

\section{\boldmath Refined approximation guarantees and hardness results for \csko} 
\label{csko-refine} 
In this section, we perform a fine-grained-complexity study of \csko. 
Motivated by the fact that our adaptivity-gap lower bound for \csko utilizes
distributions of support-size $3$, whereas the adaptivity-gap lower-bound example for
stochastic orienteering~\cite{BansalN14} considers {\em weighted Bernoulli distributions},
we investigate the complexity of \csko when 
%the (size, reward) distribution of each vertex is supported
%on at most $2$ points, 
we have distributions supported on at most $2$ points---we call this special case
\tcsko---as also the further special case where the vertex-size distributions are 
weighted Bernoulli distributions. %(\bercsko).

In stark contrast with stochastic orienteering, we show %in Section~\ref{csko-twosup} 
that the {\em adaptivity gap is a constant} for \tcsko. Moreover, we obtain non-adaptive
$O(1)$-approximation algorithms that run in polynomial time for weighted Bernoulli
distributions (Theorem~\ref{bercsko-thm}), and in time $(n+\log B)^{O(\log W)}$ for general
\tcsko (Corollary~\ref{tcsko-equiv}).  
%distributions of support-size two.

\begin{comment}
In Section~\ref{csko-nonadapt}, we %investigate the complexity 
consider the problem of obtaining approximation
guarantees for \csko with respect to the {\em non-adaptive optimum}, i.e., the optimal
value achieved by non-adaptive policies. 
%and devise a quasi-polytime $O(1)$-approximation algorithm.  
\end{comment}

The chief insight underlying the above results 
%when considering this fine-grained complexity view for \csko 
is that %the complexity of above tasks is closely related to that of the following 
one can isolate a novel {\em deterministic}
\vrp, that we call {\em orienteering with knapsack deadlines} (\orientkd), that governs
the complexity of \tcsko. %the above tasks. 
In \orientkd, we are given an 
(rooted or \ptp) orienteering instance, along with nonnegative knapsack weights
$\{\wt_v\}_{v\in V}$ and {\em knapsack deadlines} $\{\knapd_v\}$. A path $P$ with start
node $a$ is feasible, if it is feasible for the orienteering instance, and 
$\sum_{u\in P_{a,v}}\wt_u\leq\knapd_v$ for every node $v\in P$; 
%(That is, viewing the knapsack weights as sizes of jobs associated with the nodes, we want
%that each job $v\in P$ should complete by time $\knapd_v$.) 
the goal is to find a
feasible path $P$ that obtains the maximum reward.
%Besides being an interesting deterministic problem in and of itself, 

We show that, up to constant factors, \orientkd is {\em equivalent to \tcsko in terms of
approximability} (Corollary~\ref{tcsko-equiv}). 
%(Theorems~\ref{orientkd-tcsko} and~\ref{tcsko-orientkd}). 
%an $\al$-approximation guarantee for  \orientkd yields an $O(\al)$-approximation for
%\tcsko (Theorem~\ref{}) and vice versa (Theorem~\ref{}).
%and (general) \csko relative to the non-adaptive optimum, and vice versa.
%In Section~\ref{okd-alg}, 
%This equivalence result also yields as {\em corollaries}: 
%As consequences of this equivalence result, we obtain: 
%(a) a quasi-polytime $O(1)$-approximation algorithm for \tcsko, since
%we devise a quasi-polytime $O(1)$-approximation algorithm for \orientkd in 
%Section~\ref{okd-alg};
%(b) an $O(1)$-approximation for weighted Bernoulli size distributions, since in this
%case, the \orientkd instance that one needs to solve is a \knapo instance;
The %quasi-polytime 
$O(1)$-approximation for \tcsko, and the polytime $O(1)$-approximation
for weighted Bernoulli distributions both fall out as direct consequences of this
equivalence: 
the former, because we devise an $(n+\log B)^{O(\log W)}$-time $O(1)$-approximation
algorithm for \orientkd in 
Section~\ref{okd-alg}; the latter, because the \orientkd instance that one needs to solve
in this case is in fact a \knapo instance.
%In Section~\ref{csko-nonadapt}, %presents a hardness result, %for \csko, 
%yields 
Another corollary %of this equivalence 
is a hardness result for \csko
%that obtaining an approximation-factor for \csko relative the non-adaptive optimum
%is at least as hard as approximating \orientkd.
showing that an $\al$-approximation for \csko relative to the {\em non-adaptive optimum}
implies an $O(\al)$-approximation for \orientkd (Theorem~\ref{nacsko-hard}); this
follows because such an approximation guarantee for \csko implies an
$O(\al)$-approximation for \tcsko (since the adaptivity gap for \tcsko is $O(1)$).

\begin{comment}
That is, we show
that an $\al$-approximation algorithm for \orientkd yields a non-adaptive
$O(\al)$-approximation guarantee for \tcsko, and (general) \csko with respect to
the non-adaptive optimum. Moreover the converse is also true: an $\al$-approximation
algorithm (even an adaptive one) for \tcsko, or with respect to the non-adaptive optimum
for \csko, yields an $O(\al)$-approximation for \orientkd.
\end{comment}

\vspace*{-1ex}
\subparagraph*{Difficult instances of \boldmath \csko.}
We begin by distilling the key source of difficulty for \csko
(Lemma~\ref{csko-difficult}). 
This will prove to be useful when we study \tcsko, 
%and when we consider \csko with cancellations in Section~\ref{csko-cancel}, 
as it will allow us to focus on the core of
the problem. We define the size instantiation $\Sv$ of a vertex $v$ to be large if
$\Sv>W/2$, and small otherwise. We argue that the difficulty of \csko stems from instances
where most of the optimal reward comes from 
{\em vertices that instantiate to a large size with small probability}. 

To make this precise, we introduce some notation. 
For a vertex $v$, we can split its reward $\Rv$ as 
$\Rv=\excep+\trunc$, where $\excep:=\Rv\bon_{\Sv>W/2}$ and 
$\trunc:=\Rv\bon_{\Sv\leq W/2}$. %Letting $\I$ denote the original \csko instance, 
We can consider the modified \csko instances $\excep[\I]$ and $\trunc[\I]$, where the
rewards are given by $\{\excep\}_{v\in V}$ and $\{\trunc\}_{v\in V}$ respectively; 
so in $\excep[\I]$, we only collect non-zero reward from large instantiations, and in 
$\trunc[\I]$, we only collect non-zero reward from small instantiations. 
For $p\in[0,1]$, define $\excep[\I](p)$ to be the instance with vertex set 
$V(p):=\{v\in V: \Pr[\Sv>W/2]\leq p\}$ (note that $\rt\in V(p)$).
Thus, in instance $\excep[\I](p)$, we only consider vertices that instantiate to a large
size with probability at most $p$ (i.e., small probability), and collect reward only from
large instantiations.

\begin{lemma} \label{csko-difficult} \label{csko-hard}
%Let $\I$ be an \csko instance. 
Suppose we have an $\al$-approximation algorithm for \csko instances of the form
$\excep[\I](0.5)$. Then, we can obtain an $\bigl(\al+O(1)\bigr)$-approximation algorithm for
all \csko instances.
\end{lemma}

\begin{proof}
A \csko instance $\I$ can be decomposed into three instances, $\I_1=\trunc[\I]$,
$\I_2=\excep[\I](0.5)$, and $\I_3$ with vertex set $V_3:=\{v\in V:\Pr[\Sv>W/2]>0.5\}$ and
rewards $\{\excep\}_{v\in V_3}$. Clearly, by construction, a vertex $v$ yields positive
reward in at most one of these 3 instances for any size instantiation, and so 
$\OPT=\OPT(\I)\leq\OPT(\I_1)+\OPT(\I_2)+\OPT(\I_3)$. 
We show that one can obtain
(non-adaptive) polices that yield approximation guarantees of $\beta_1=O(1)$,
$\beta_3=O(1)$ for $\I_1$ and $\I_3$ respectively. Let $\beta_2=\al$. 
Therefore, if we can obtain an
$\al$-approximation to $\I_2$, then we can do the following: with probability 
$\frac{\beta_j}{\beta_1+\beta_2+\beta_3}$, we run the corresponding algorithm for instance
$\I_j$, for $j=1,2,3$. The expected reward obtained via this is at least
$\sum_{j=1}^3\frac{\beta_j}{\beta_1+\beta_2+\beta_3}\cdot\frac{\OPT(\I_j)}{\beta_j}
\geq\frac{\OPT}{\beta_1+\beta_2+\beta_3}=\frac{\OPT}{\al+O(1)}$.

%Consider $\I_3$. 
Any adaptive policy for $\I_3$ can collect positive
reward from at most one vertex, namely, the first vertex that instantiates to a large
size, after which the adaptive policy may as well stop. The expected number of nodes
visited by an adaptive policy is at most $\sum_{i\geq 1}i\cdot2^{-(i-1)}\leq 4$, so 
simply visiting the node in $V_3$ with largest expected reward, yields an
$O(1)$-approximation to $\OPT(\I_3)$.  

\begin{comment}
For $\I_1$, we show in Appendix~\ref{append-tcsko} that one can obtain an
$O(1)$-approximation by solving 
%an $O(1)$-approximate solution to
the \knapo instance with 
rewards $\{\E{\trunc}\}_{v\in V}$, travel budget $B$, knapsack weights
$\{\E{\min\{\Sv,W\}}\}_{v\in V}$, and knapsack budget $2W$. 
%yields an $O(1)$-approximation to $\OPT(\I_1)$. 
\end{comment}

For $\I_1$, consider the \knapo instance with 
rewards $\{\E{\trunc}\}_{v\in V}$, travel budget $B$, knapsack weights
$\{\wt_v=\E{\min\{\Sv,W\}}\}_{v\in V}$, and knapsack budget $2W$, and the LP-relaxation
\eqref{kolp} for this instance. 
We can argue as in the proof of Theorem~\ref{knaporedn} to show that the optimal value of
this LP is at least $\OPT(\I_1)$. 
For any rooted path $Q$ corresponding
to the execution of an adaptive policy, letting $v$ be the furthest node from $\rt$ on $Q$,
taking $x^v$ and $z^v$ be the indicator vectors of the edge-set and node-set of $Q$, and
all other variables to be $0$, satisfies \eqref{pref-visit}--\eqref{unit}. So the
$\Pr[Q]$-weighted convex combination of these solutions also satisfies these constraints,
where $\Pr[Q]$ is the probability that the execution of the adaptive policy results in the
path $Q$. Under this convex combination, constraint \eqref{knbudget} reads 
$\sum_{u\in V}\Pr[\text{$u$ is visited by the adaptive policy}]\cdot\E{\min\{\Su,w\}}\leq 2W$.
This is satisfied due to a standard argument used for stochastic
knapsack~\cite{DeanGV08,GuptaKMR11}.
The objective value of this solution is at least $\OPT(\I_1)$, since the expected reward
collected from a node $u$ conditioned on $u$ being visited by the adaptive policy is at
most $\E{\trunc[u]}$. 

Using Theorem~\ref{knapo-round}, we can %round an optimal solution to the LP to 
obtain a rooted path $Q$ with $\sum_{v\in Q}\E{\trunc}\geq\OPT(\I_1)/5$ and 
$\wt(Q)\leq 2W$. We sample 
each vertex on $Q-\rt$ independently with probability $\frac{1}{8}$ and return the
resulting rooted path $Q''$. For any $v\in Q$, the probability that the non-adaptive
policy $Q''$ processes $v$ by time $W/2$ is at least $\frac{1}{16}$. This is because, 
taking $T=Q''\cap (Q_{\rt,v}-v)$, this
probability is $\Pr[v\in Q'']\cdot\Pr\bigl[\sum_{u\in T}\Su\leq W/2\bigr]$,
and
%total expected weight of nodes in $Q''_{\rt,v}-v$ is at most $\frac{W}{4}$, and so 
%$\Pr[\wt(Q''_{\rt,v}-v)>W/2]$ is at most $0.5$. But 
%$\Pr[\wt(Q''_{\rt,v}-v)>W/2]=
\[
\Pr\Bigl[\sum_{u\in T}\Su>W/2\Bigr]
=\Pr\Bigl[\sum_{u\in T}\min\{\Su,W\}>W/2\Bigr]
\leq\frac{\E{\sum_{u\in T}\min\{\Su,W\}}}{W/2}\leq\frac{\wt(Q)/8}{W/2}\leq
0.5.
\]
Since we only obtain positive reward from $v$ if it instantiates to size at most $W/2$,
this implies that we collect expected reward 
$\E{\trunc}$ if this happens. So the total expected reward obtained is at least
$\frac{1}{16}\cdot\sum_{v\in Q}\E{\trunc}\geq\frac{\OPT(\I_1)}{80}$.

\end{proof}

\subsection{\boldmath \tcsko: \csko with distributions of support-size at most 2} 
\label{tcsko} \label{csko-twosup}
Recall that \tcsko denotes the special case of \csko where, for each vertex $v$, the
distribution of $\Sv$ is supported on at most $2$ values; we denote these two values
$\sonev$, $\stwov$, where 
$\sonev\geq\stwov$. By Lemma~\ref{csko-difficult}, to obtain an $O(1)$-approximation for
\tcsko, it suffices to consider the instance $\I_2=\excep[\I](0.5)$, and we focus on such
instances in the sequel. To keep notation simple, we
continue to use $V$ to denote the vertex set of $\I_2$. Then we may assume that the
(size, reward) distribution for each $v\in V$ is $(\sonev,R_v)$ with
probability $p_v$, and $(\stwov,0)$ with probability $1-p_v$, where
(i) $\sonev>W/2\geq\stwov$ and (ii) $p_v\leq 0.5$. Property (i) holds because if 
$\sonev\leq W/2$, then $v$ yields $0$ reward for $\I_2$, so may be discarded; if
$\stwov>W/2$, then $\Pr[\Sv>W/2]=1$, which means that $v$ would not be considered for
$\I_2$. Given (i) the reward when the size is $\stwov$ must be $0$, and (ii) holds because 
$p_v=\Pr[\Sv>W/2]$. 

We first argue that the adaptivity gap for such instances is $1$
(Theorem~\ref{tcsko-nogap}). Next, we show that the above \tcsko problem is equivalent to   
\orientkd, up to constant-factor approximation losses. 
Whenever we say ``equivalent'' below, we always mean equivalent up to
a multiplicative $O(1)$ factor. 
We actually consider another problem, 
{\em knapsack orienteering with knapsack deadlines} (\knapokd), which is the
knapsack-constrained version of \orientkd. We show the equivalence of \tcsko and \knapokd
(Theorems~\ref{okd-to-tcsko} and \ref{tcsko-to-okd});   
by Theorem~\ref{knapvrp}, \orientkd and its knapsack-constrained version \knapokd are
equivalent, %up to constant factors. 
so this shows the desired equivalence (Corollary~\ref{tcsko-equiv}).

\begin{theorem} \label{tcsko-nogap}
The adaptivity gap for \tcsko (instances of the form $\excep[\I](0.5)$) is $1$. 
\end{theorem}

\begin{proof}
Let $\T$ be the decision tree of an optimal adaptive policy. 
Consider the (rooted) path $\sigma$ of $\T$ corresponding to the $\stwov$ size
instantiations. %Observe that 
Then $\T$ cannot collect any reward outside of $\sigma$, since
the residual knapsack budget when we reach any node $v\in\T\setminus\sigma$ is less than $W/2$.  
Thus, the non-adaptive policy represented by $\sigma$ has the same expected reward as
$\T$. %This concludes the proof. 
\end{proof}

The following basic claim will guide us toward the translation between \tcsko and
\knapokd that will be used to show the equivalence.
%toward defining a suitable \knapokd instance from
%a \tcsko instance. %in order to show equivalence.

\begin{claim} \label{expreward}
Consider an instance of \tcsko, and a $\rt$-rooted path $\tau$ representing a non-adaptive
policy. Let $v\in\tau$ be such that $\sum_{w\prec_\tau v}\stwow\leq W-\sonev$. 
Then, the expected reward from $v$ is $p_v\Rv\prod_{w\prec_\tau v}(1-p_w)$.  
\end{claim}

\begin{proof}
%For any $v\in \tau$, 
Let $\Piv$ be the random variable denoting the reward collected from node $v$. 
Then $\Piv$ can be positive only if we have $\Sw=\stwow$ for every $w\prec_\tau v$, since
$\soneu>W/2$ for every node $u\in V$. 
Therefore, 
\begin{equation*}
\E{\Piv}=R_v\cdot\Pr\bigl[\Sv=\sonev\,|\,\Sw=\stwow\ \forall w\prec_\tau v\bigr]\cdot
\Pr[\Sw=\stwow\ \forall w\prec_\tau v]
= p_vR_v\prod_{w\prec_\tau v}(1-p_w)
\end{equation*}
%It follows from this observation and the law of total expectation that 
%   \begin{align*}
%        \Ex[\,\Piv\,]&=\Pr\Big[\bigwedge_{w\prec_\tau v}\{\sizee(w)=\stwow\}\Big]\cdot\Ex\Big[\,\Piv\,\Big|\,\bigwe%dge_{w\prec\tau v}\{\sizee(w)=\stwow\}\Big]\\
%        &=\Big(\prod_{w\prec_\tau v}(1-p_w)\Big)\cdot\Ex\Big[\Piv\,\Big|\,\bigwedge_{w\prec\tau v}\{size(w)=\stwow\%}\Big]=\Big(\prod_{w\prec_\tau v}(1-p_w)\Big)\cdot p_v\Rv,
%    \end{align*}
where the second equality follows since different vertices are independent.
%from the independence of the random variables accross different vertices, and the third
%equation follows from the assumption that $\sum_{w\prec_\tau v}\stwow\leq W-\sonev$ for
%every $v\in \tau$.  
\end{proof}

\begin{theorem} \label{okd-to-tcsko} \label{hardness1}
Given an $\al$-approximation algorithm for \knapokd, one can obtain an
$O(\al)$-approximation algorithm for \tcsko.
\end{theorem}

\begin{proof} %[Proof of lemma \ref{hardness1}:]
Consider an instance $\mathcal{I}$ of \tcsko.
%We may assume that $d(r,v)\leq B$, as otherwise we can simply discard $v$. 
%and 
%$\max_{v\in V}p_vR_v\leq e^{-0.5}\cdot\OPT(\I)$, as otherwise simply visiting the vertex
%$\argmax_{v\in V}p_vR_v$ yields an $O(1)$-approximation.
Utilizing Claim~\ref{expreward}, we let the rewards be $\{p_vR_v\}_{v\in V}$, and define
the knapsack weights and knapsack deadlines %in the underlying \orientkd instance 
to encode the condition 
$\sum_{\text{$w$ visited before $v$}}\stwow\leq W-\sonev$. The additional knapsack
constraint will encode that the total $p_v$-weight of the path should be at most $1$; this 
will ensure that the expected reward obtained for $\I$ from each vertex $v$ on the
\knapokd-solution is $\Omega(p_vR_v)$. 
Formally, define the instance $\mathcal{J}$ of \knapokd with the same metric and
travel budget as $\I$, rewards $\{\pi_vR_v\}_{v\in V}$, knapsack weights
$\{\stwov\}_{v\in V}$ and knapsack 
deadlines $\{W-\sonev+\stwov\}_{v\in V}$, and the additional knapsack constraint 
given by weights $\{p_v\}_{v\in V}$ and budget $1$.
We next argue that: 
%\begin{enumerate}[label=(\alph*), nosep, leftmargin=*]
(a) %\item 
if $\sg$ is a solution to $\J$, then the non-adaptive policy $\sg$
for $\I$ yields reward at least $\frac{1}{4}\cdot\sum_{v\in\sg}p_vR_v$;
(b) %\item 
$\OPT(\J)=\Omega\bigl(\OPT(\I)\bigr)$.
%\end{enumerate}
It follows that an $\al$-approximate solution to $\J$ yields an $O(\al)$-approximation for 
\tcsko. 

Part (a) follows directly from Claim~\ref{expreward}: by the feasibility of $\sg$
for $\J$, for every $v\in\sg$, we have 
$\sum_{w\prec_\sg v}\stwow+\stwov\leq W-\sonev+\stwov$, and so the expected reward
obtained from $v$ is 
$p_vR_v\cdot\prod_{w\prec_\sg v}(1-p_w)\geq p_vR_v\cdot 4^{-\sum_{w\prec_\sg v}p_w}\geq\frac{1}{4}\cdot p_vR_v$.
The first inequality is because $1-x\geq 4^{-x}$ for all $x\leq 0.5$, and recall that
$p_w\leq 0.5$ for all $w\in V$.

For part (b), by Theorem~\ref{tcsko-nogap}, there is no adaptivity gap for \tcsko.
Let $\tau$ be the path represented by an optimal non-adaptive policy to $\mathcal{I}$.
We may assume that $\sum_{w\prec_\tau v}\stwow\leq W-\sonev$ for every $v\in \tau$;
otherwise, we cannot collect positive reward from $v$, and can just delete $v$ from
$\tau$. Let $\tau_0$ be the maximal rooted subpath of $\tau$ having total $p_v$-weight at
most $1$. By construction, $\tau_0$ is a feasible solution to the \knapokd instance $\J$.
We argue that the expected reward obtained from $\tau_0$ for $\I$ is at least
$(1-e^{-0.5})\OPT(\I)$; 
hence, 
\begin{equation*}
\OPT(\J)\geq\sum_{v\in\tau_0}p_vR_v\geq\E{\text{reward of $\tau_0$ for $\I$}}\geq(1-e^{-0.5})\OPT(\I).
\end{equation*}
If $\tau_0=\tau$, the statement clearly holds. So suppose otherwise. 
Note that $\sum_{w\in\tau_0}p_w>0.5$, since $p_v\leq 0.5$ for every node $v$.
%Let $v_0$ be the vertex on $\tau$ immediately after the end-node of $\tau_0$. 
So for each $v\in\tau-\tau_0$, %with $v_0\preceq_\tau v$, 
we have 
$\prod_{w\prec_\tau v}(1-p_w)\leq e^{-0.5}\prod_{w\in\tau-\tau_0: w\prec_\tau v}(1-p_w)$. 
So by Claim~\ref{expreward}, we have
\begin{equation}
\OPT(\I) \leq \sum_{v\in\tau_0}p_vR_v\prod_{w\prec_\tau v}(1-p_w)
+e^{-0.5}\cdot\sum_{v\in\tau-\tau_0}p_vR_v\prod_{w\in\tau-\tau_0:w\prec_\tau v}(1-p_w).
\label{tau0}
\end{equation}
The term $\sum_{v\in\tau-\tau_0}p_vR_v\prod_{w\in\tau-\tau_0:w\prec_\tau v}(1-p_w)$
corresponds to the expected reward of the non-adaptive policy given by the rooted subpath
of $\tau$ that skips all the nodes in $\tau_0$ and moves directly to the first node of
$\tau-\tau_0$.  
So the second term on the RHS of \eqref{tau0} is at most $e^{-0.5}\cdot\OPT(\I)$. It
follows that the expected reward from $\tau_0$ is at least $(1-e^{-0.5})\OPT(\I)$. 
\end{proof}

%By essentially ``inverting'' the reduction in Theorem~\ref{okd-to-tcsko}, we obtain the
%following. 

\begin{theorem} \label{tcsko-to-okd} \label{hardness2}
Given an $\al$-approximation algorithm for \tcsko, one can obtain an
$O(\al)$-approximation algorithm for \knapokd.
\end{theorem}

\begin{proof} %{Theorem~\ref{tcsko-to-okd}} %[Proof of lemma \ref{hardness2}]
Consider an instance $\mathcal{J}$ of \knapokd given by a metric $(V,d)$, root $\rt\in V$,
travel budget $B$, rewards, knapsack weights and knapsack deadlines 
$\{\pi_v,\wt_v,\knapd_v\}_{v\in V}$ respectively, and the additional knapsack constraint
given by weights $\{\knwt_v\}_{v\in V}$ and budget $\knbudg$. By scaling, we may assume
that $\knbudg=1$, and $\knwt_v\in[0,1]$.
We may assume that, for every $v\in V$, the path $\rt,v$ is feasible, as otherwise we can
discard $v$. 
%and given this, we can assume that 
%any such path yields reward at most
%$\pi_v\leq e^{-0.5}\cdot\OPT(\J)$ for all $v\in V$, as otherwise, there is a trivial
%$O(1)$-approximation for \knapokd. 
We may further assume that $\knwt_v\leq 0.5$ for all $v\in V$. Any solution to $\J$ may
visit at most one vertex $v$ with $\knwt_v>0.5$, and so visiting the maximum-reward vertex
yields an $O(1)$-approximation to the reward obtained from such vertices. So for the
purposes of obtaining an $O(\al)$-approximation, we can concentrate on vertices $v$ with
$\knwt_v\leq 0.5$. 

We essentially ``invert'' the reduction in Theorem~\ref{okd-to-tcsko} to define the \tcsko
instance $\I$. The metric, travel budget, and root are the same as in $\J$. Take the
processing-time budget $W$ to be sufficiently large, e.g., 
$1+2\max_{v\in V}\knapd_v$.
%and $\{\piv,\monev,\mtwov,\qv\}_{v\in V}$. Recall that we assume
%that $\monev\leq T/2$ since just visiting the vertex $v$ of highest reward for which
%$\monev>T/2$ is an optimal solution for the sub-instance of $\mathcal{J}$ induced by such
%vertices. Define an instance $\mathcal{I}$ of \tcsko on the same metric and travelling
%budget as $\mathcal{J}$, with $W:=1+2\cdot\max_{v\in V}{\qv},$ and  
For every $v\in V$, set
\begin{equation*}
\sonev:=W-\knapd_v+\wt_v, \quad \stwov:=\wt_v, \quad 
p_v:=\knwt_v, \quad \Rv:=\frac{\piv}{p_v}
\end{equation*}
Note that the above \tcsko instance has the special structure that we have throughout
assumed: we have $\sonev>W/2\geq\stwov$ for every $v\in V$, since we have chosen $W$
sufficiently large, and $p_v\leq 0.5$ by the assumptions on $\J$.

Now, the theorem follows by simply observing that for the above instance $\J$, the \tcsko
instance that we create in the proof of Theorem~\ref{okd-to-tcsko} is 
{\em precisely $\I$}. It follows from the proof therein that 
$\OPT(\I)=\Theta\bigl(\OPT(\J)\bigr)$. Moreover, from any solution to $\I$, one can
extract a non-adaptive solution $\tau$ with the same expected reward, as described in the
proof of Theorem~\ref{tcsko-nogap}. We can assume that no subpath of $\tau$ yields larger
expected reward than $\tau$. The %portion of the proof of Theorem~\ref{okd-to-tcsko} that
proof of the upper bound on $\OPT(\I)$ in Theorem~\ref{okd-to-tcsko} 
then shows how to extract from $\tau$, a solution to $\J$ of reward
$\Omega\bigl(\E{\text{reward of $\tau$}}\bigr)$.
\end{proof}

\begin{corollary} \label{tcsko-equiv}
\tcsko is equivalent, in terms of approximation, to \orientkd, up to constant factors.
Thus, there is an $(n+\log B)^{O(\log W)}$-time $O(1)$-approximation algorithm for \tcsko.
\end{corollary}

\begin{proof}
The equivalence follows from Theorems~\ref{okd-to-tcsko} and~\ref{tcsko-to-okd}, 
and Theorem~\ref{knapvrp}, which shows that \orientkd and \knapokd are equivalent.
%show that \tcsko is equivalent to
%\knapokd, which in turn is equivalent to \orientkd by Theorem~\ref{knapvrp}.
The approximation guarantee then follows from Theorem~\ref{okd-quasipoly}, which
presents an $O(1)$-approximation for \orientkd.
\end{proof}

%\vspace*{-2ex}
%\subparagraph*{\boldmath $O(1)$-approximation for \csko with weighted Bernoulli size
%  distributions.} 
\begin{theorem}[\textbf{Weighted Bernoulli size distributions}] \label{bercsko-thm}
There is a polytime $O(1)$-approximation for \csko with weighted Bernoulli size
distributions. 
\end{theorem}

\begin{proof}
With weighted Bernoulli size distributions, we are considering the special case of \tcsko
with $\stwov=0$ for all $v\in V$. Recall that due to Lemma~\ref{csko-difficult}, we can
focus on the setting where $\sonev>W/2$, and $p_v:=\Pr[\Sv=\sonev]\leq 0.5$. 
The result now follows from Theorem~\ref{okd-to-tcsko} by observing that the \knapokd
instance $\J$ created in the proof of Theorem~\ref{okd-to-tcsko} is now simply a \knapo
instance. (So Theorem~\ref{okd-to-tcsko} shows a reduction from \csko with weighted
Bernoulli size distributions to \knapo.) 
%where the knapsack constraint ensures that the total $p_v$-weight of the path
%should be at most $1$. 
This is because $\stwov=0$ for all $v\in V$, so the
knapsack-deadlines are trivially satisfied by any rooted path. Since we have a polytime
$O(1)$-approximation for \knapo, the result follows.
\end{proof}

%\subsection{Hardness of approximating the non-adaptive optimum for \boldmath \csko}
%\label{csko-nonadap}

\begin{theorem}[\textbf{Hardness of approximating the non-adaptive optimum}]
\label{nacsko-hard}
Given an $\al$-approximation algorithm for \csko with respect to the non-adaptive optimum,
we can obtain an $O(\al)$-approximation algorithm for \orientkd.
\end{theorem}

\begin{proof}
An $\al$-approximation algorithm for \csko with respect to the non-adaptive optimum would
yield an $O(\al)$-approximation algorithm for \tcsko, due to
Theorem~\ref{tcsko-nogap}. Hence, by Theorem~\ref{tcsko-to-okd}, we would obtain an
$O(\al)$-approximation algorithm for \knapokd, and hence \orientkd.
\end{proof}

\subsection{Approximation algorithms for \boldmath \orientkd} \label{okd-alg}
We now devise approximation algorithms for \orientkd. We consider rooted \orientkd, but
the ideas carry over to the \ptp problem as well.
Recall that an instance of rooted \orientkd consists of a rooted orienteering instance
with metric $(V,d)$, root $\rt\in V$, length/travel budget $B$, rewards $\{\rewd_v\}_{v\in V}$,
along with nonnegative knapsack weights $\{\wt_v\}_{v\in V}$ and {\em knapsack deadlines}
$\{\knapd_v\}$. A rooted path $P$ is feasible, if $d(P)\leq B$ and 
$\wt(P_{\rt,v})\leq\knapd_v$ for every node $v\in P$;
we sometimes call $\wt(P_{\rt,v})$, the completion time of $v$. 
%(That is, viewing the knapsack weights as sizes of jobs associated with the nodes, we want
%that each job $v\in P$ should complete by time $\knapd_v$.) 
The goal is to find a
feasible path $P$ that obtains the maximum reward.
By scaling, we may assume that all $d(u,v)$ distances, knapsack weights, and deadlines are
integers. Let $W:=\max_v\knapd_v$. 

\begin{theorem} \label{okd-quasipoly} \label{quasipolytime-kdo} 
There is an $O(1)$-approximation algorithm for \orientkd with $(n+\log B)^{O(\log W)}$
running time. 
\end{theorem}

Theorem~\ref{okd-quasipoly} is the main result and focus of this section.
At the end of this section, 
%In Appendix~\ref{append-tcsko}, 
we also present a polytime
$O\bigl(\log(\frac{W}{\knapd_{\min}})\bigr)$-approximation algorithm, where $\knapd_{\min}$
is the minimum non-zero knapsack deadline (Theorem~\ref{okd-logthm}).

%\subsubsection*{Proof of Theorem~\ref{okd-quasipoly}}
%{\boldmath $O(1)$-approximation in almost quasi-polytime}
For brevity, we refer to $(n+\log B)^{O(\log W)}$ running time as ``almost quasi-polytime.'' 
We may assume that: (i) $\wt_v\leq\knapd_v$ for every $v\in V$, as otherwise, we can discard
$v$; and (ii) $\rewd_\rt=0$, since a $c$-approximation algorithm when we reset $\rewd_\rt$ to
$0$ implies a $c$-approximation with the original rewards; (iii) $\wt_\rt=0$, since we can
always set this and work with the modified deadlines 
$\{\knapd'_v:=\knapd_v-\wt_\rt\}_{v\in V}$.

%First, we provide a polynomial time $O\left(\log{W}\right)$-approximation, where
%$W:=\max_{v\in V}{\qv}$, proving theorem \ref{polytime-kdo}. Then, we present a
%$O(1)$-approximation algorithm with quasipolynomial running time, which proves theorem
%\ref{quasipolytime-kdo}. 

%\subsection{$O(1)$-approximation in quasi-polynomial time}

%In this section we present a constant-factor approximation algorithm for \kdo with quasi-polynomial running time. 
%$n^{O(\log{W})}$, where $W:=\max_{v\in V}{q_v}$
Our approach is similar to that used for deadline TSP in~\cite{FriggstadS17}.
Consider an optimal path $P^*$. Let $\opt=\rewd(P^*)$ be the optimal value of the \orientkd
instance. 
Our strategy is to guess a set of $O(\log W)$ portal
vertices along $P^*$, along with some auxiliary information to estimate the lengths of
the segments of $P^*$ between consecutive portal vertices, 
%estimate $d(P^*_{a,b})$ for every pair $a,b$, of consecutive portal vertices, 
%and the length of the segments of $P^*$ between consecutive portals,
and then (approximately) solve \knapo instances to obtain suitable paths between the
portal vertices. 
%corresponding to each segment.
%We assume for simplicity that $s_\rt=\rewd_\rt=0$.

Let $\zeta\geq 1$ be such that $\zeta^2\leq\zeta+1$; e.g., $\zeta=1.5$.
Define $u_{-1}:=\rt$. For $j\geq 0$, let $u_j$ be the first node on $P^*$ after $u_{j-1}$
such that $\wt(P^*_{u_{j-1},u_j}-u_{j-1})\geq\zeta^j$. 
Note that 
$\wt\bigl(P^*_{u_{j-1},u_j}-\{u_{j-1},u_j\}\bigr)\leq W_j:=\ceil{\zeta^j}-1$.
%Given $u_{i-1}$ define $u_{i}$ as the first successor of $u_{i-1}$ in $P^*$ such that the
%total size of the vertices between $u_{i-1}$ and $u_{i}$, excluding $u_{i-1}$ but
%including $u_{i}$, is at least $2^i$  
%i.e., $u_{i}$ is the first vertex in $P^*$ such that $\sum_{u_{i-1}\prec v\preceq
%  u_{i}}s_v\geq 2^i$. 
Let $k$ be the largest integer for which $u_k$ is well defined. 
If $u_k$ is the end-point of $P^*$ define $k^\prime:=k$; otherwise, define $u_{k+1}$ as
the end-point of $P^*$ other than $\rt$, and $k^\prime:=k+1$. 
Note that $k'\leq 1+\log_\zeta W=O(\log W)$.  
%It is immediate from this definition 
Also, for any index $j\in\dbrack{k^\prime}$, we have 
$\wt(P^*_{\rt,u_{j-1}})\geq\lb_j:=\sum_{i=0}^{j-1} \max\{\zeta^i,\wt_{u_i}\}$.
%and any vertex $v\in P^*_{u_{j-1},u_j}-u_{j-1}$,
%the total weight of all nodes on $P^*_{\rt,u_j}$ is at least
%$u_{i-1}\prec v\preceq u_{i}$ the total size of the predecessors (or knapsack visiting
%time) of $v$ in $P^*$ is at least 
%\begin{equation*}\label{lowervis}
%    \lb_j:=\sum_{i=0}^{j-1} \max\{2^i,s_{u_i}\},
%\end{equation*}

We guess the $u_j$-nodes and the lengths of the $P^*_{u_{j-1},u_j}$ segments.
%proceed to guess this set of portal vertices and the length of the segments of $P^*$
%between the consecutive portal vertices.  
%Let  $W=\min\left\{\max_{v\in V}q_v\,,\,\sum_{v\in V}s_v\right\}.$
%By our definition of the $u_i$'s, it follows that $k^\prime\leq \lceil\log W\rceil$.
As in the structural result for \csko (Theorem~\ref{strucdthm}), instead of guessing the
actual length $d(P^*_{u_{j-1},u_j})$, we guess the two-point regret of $P^*_{u_{j-1},u_j}$
with respect to a mid-point node $m_j$
within a factor of $2$, to reduce the choices for each index $j$ to $O(\log B)$ at the
expense of a factor-$2$ loss in the reward.
%
%for each segment between
%consecutive portals we can guess some vertex in the segment in order to reduce the number
%of guesses for the length of the segment from $B$ to $\log{B}$ (by potentially
%shortcutting the segment), at the expense of losing at most half of the reward of $P^*$. 
This yields the following structural result for \orientkd.

\begin{lemma}\label{guesskdo} \label{okd-structhm}
For some $k^\prime\leq 1+\ceil{\log_\zeta W}$, there are portal vertices 
$\{u_j\}_{j\in\dbrack{k^\prime}}$, and for each 
%consecutive pair of nodes $a,b\in\fullset$, 
index $j\in\dbrack{k'}$,
an auxiliary vertex $m_j$, integer $\gm_j\geq 0$, and a $u_{j-1}$-$u_j$ path 
$\spath_j$, such that the following hold.

%\begin{flushleft}
%\noindent
\begin{enumerate}[label=(\alph*), topsep=0.2ex, noitemsep, leftmargin=*]
\item \textnormal{(Distance)} For every $j\in\dbrack{k'}$, we have 
$d(\spath_j)\leq\dist_j:=2^{\gm_j}-1+d(u_{j-1},m_j)+d(m_j,u_j)$. 
% For every pair of consecutive nodes $a,b\in\fullset$, we
%have $d(P_{a,b})\leq\dist_a:=2^{\gm_{a}}-1+d(a,m_a)+d(m_a,b)$.
\label{okd-dist}

%\hspace*{-5ex}
%\noindent
\item \textnormal{(Total length)} $\sum_{j=0}^{k'}\dist_{j}\leq B$. 
\label{okd-len} %\qquad \qquad \qquad

\item \textnormal{(Reward)} 
$\sum_{j=0}^{k'}\rewd(\spath_j-u_{j-1})\geq\frac{\opt}{2}$. \label{okd-rewd} %\\

%\noindent
\item \textnormal{(Size)} For each $j\in\dbrack{k'}$, we have 
$\wt\bigl(\spath_j-\{u_{j-1},u_j\}\bigr)\leq W_j:=\ceil{\zeta^j}-1\leq\zeta^j$.
\label{okd-size} %\\

%\noindent
\item \textnormal{(Feasibility)} The concatenated path
$\spath:=\spath_o,\spath_1,\ldots,\spath_{k^\prime}$ is feasible for \orientkd. 
\label{okd-feas} %\\
%every vertex is reached by its knapsack deadline. 

%\noindent
\item \textnormal{(Lower bound)} For each $j\in\dbrack{k^\prime}$ and
$v\in\spath_j-u_{j-1}$, we have $\wt(\spath_{\rt,v})\geq\lb_j+\wt_v$, where
$\lb_j:=\sum_{h=0}^{j-1} \max\{\zeta^h,\wt_{u_h}\}$. \label{okd-lb}
\end{enumerate}
%\end{flushleft}
\end{lemma}

\begin{proof}
The $u_j$ portal vertices are as defined earlier, using the optimal path $P^*$.
The $\spath_j$ paths will be subpaths of the $P^*_{u_{j-1},u_j}$ paths, so
properties \ref{okd-size}--\ref{okd-lb} hold.  
Let $m_j\in P^*_{u_{j-1},u_j}$ and $\gm_j\geq 0$ be such that 
$\rewd(P^*_{u_{j-1},m_j}),\rewd(P^*_{m_j,u_j})\geq\rewd(P^*_{u_{j-1},u_j})/2$, and 
$2^{\gm_j}-1\leq\dreg(P^*_{u_{j-1},u_j},m_j)<2^{\gm_j+1}-1$. Then, as in the proof of
Theorem~\ref{strucdthm}, we take $\spath_j$ to be $P^*_{u_{j-1},m_j},u_j$ if
$\dreg(P^*_{u_{j-1},m_j})\leq\dreg(P^*_{m_j,u_j})$ and $P^*_{m_j,u_j}$ otherwise.
As in the proof of Theorem~\ref{strucdthm}, this shows properties
\ref{okd-dist}--\ref{okd-rewd}. 
\end{proof}

\SetAlgoProcName{Algorithm}{Procedure}
\begin{procedure}[ht!]
\SetKwInput{KwInput}{Input}
\SetKwInput{KwOutput}{Output}
\caption{OrientKD-Alg(): \qquad 
\textnormal{// almost quasi-polytime approximation for \orientkd}
\label{algorithmKDOquasi}}
\SetKwComment{simpc}{// }{}
\SetCommentSty{textnormal}
\DontPrintSemicolon
\KwInput{\orientkd instance; $k'$, portal vertices $\{u_j\}_{j\in\dbrack{k'}}$, distance
  bounds $\{\dist_j\}_{j\in\dbrack{k'}}$ satisfying the properties in
  Lemma~\ref{okd-structhm}; an $\al$-approximation algorithm $\Alg$ for \ptp-\knapo}
\KwOutput{Feasible \orientkd solution}

%\For{$i\gets 0$ \KwTo $k$}
Define $V_j:=\{v\in V: \knapd_v\geq\lb_j+\wt_v\}$, for all $j\in\dbrack{k'}$. \;
%Obtain paths $Q^0,Q^1,...,Q^k$ in the following way:\;

\For{$j\gets 0$ \KwTo $k'$}
{Obtain a $u_{j-1}$-$u_j$ path $Q^j$ by running $\Alg$ on the \ptp-\knapo
instance %the node-set is $V_j-\bigcup_{h=0}^{j-1}Q^h 
with start node $u_{j-1}$, end node $u_j$, length budget $\dist_j$, knapsack weights  
$\wt_v$ for all $v\in V-\{u_{j-1},u_j\}$ and $0$ for $v\in\{u_{j-1},u_j\}$, 
knapsack budget $W_j$, and rewards $\rewd_v$ for all 
$v\in V_j-\bigcup_{h=0}^{j-1} Q^h$ and $0$ for all other $v$.
By shortcutting, we may assume that $Q^j$ only visits nodes in
$\bigl(V_j-\bigcup_{h=0}^{j-1} Q^h\bigr)+u_{j-1}$.} \label{okdalg-iter}

For $\ell\in\{0,1,2\}$, let $Z^\ell$ be the path given by the node-sequence
$\rt,\{Q^j-u_{j-1}-u_j\}_{j\in\dbrack{k'}:j=\ell\bmod 3}$. 
Let $Z^{\max}\in\{Z^1,Z^2,Z^3\}$ be the path that collects maximum reward.
\; \label{okdalg-zpath}

\Return $Z^{\max}$ or the path $\rt,u_0,u_1,\ldots,u_{k'}$, whichever gathers
greater reward. 
\end{procedure}

\vspace*{-2ex}
\subparagraph*{Analysis.}
Let $P$ be the path obtained by concatenating the $Q^j$ paths for all
$j\in\dbrack{k'}$. We argue that $P$ obtains large reward, and that each $Z^\ell$ path is
feasible. The approximation guarantee of Theorem~\ref{okd-quasipoly} then follows. The
running time is determined by the time needed to enumerate the quantities in
Lemma~\ref{okd-structhm}, which is $(n+\log B)^{O(\log W)}$.

\begin{lemma}[Follows from~\cite{ChekuriK17}] \label{okdalg-totrewd}
We have $\rewd(P)\geq\opt/2(\al+1)$.
\end{lemma}

\begin{proof}
This follows from~\cite{ChekuriK17}, since the problem of finding a maximum
$\rewd\bigl(\bigcup_{j=0}^{k'}Q^j\bigr)$ collection of $Q^j$-paths is an instance of the
maximum-coverage problem with group budget constraints. We include a proof for
completeness. We have 
$\rewd(Q^j-u_{j-1})\geq\frac{1}{\al}\cdot\rewd\bigl(\spath_j-\bigcup_{h=0}^{j-1}Q^h\bigr)$,
since $\spath_j$ is a feasible solution to the \knapo instance considered in iteration
$j$ of step~\ref{okdalg-iter}, due to parts \ref{okd-dist}, \ref{okd-size} of
Lemma~\ref{okd-structhm}. Summing this over all $j\in\dbrack{k'}$ yields
$\rewd(P)\geq\frac{1}{\al}\cdot\bigl[\sum_{j\in\dbrack{k'}}\rewd(\spath_j-u_{j-1})-\rewd(P)\bigr]$,
since $\spath_j-u_{j-1}\sse V_j$ by parts~\ref{okd-feas}, \ref{okd-lb} of
Lemma~\ref{okd-structhm}.  
The statement now follows from Lemma~\ref{okd-structhm} \ref{okd-rewd}.
\end{proof}

\begin{lemma} \label{okdalg-feas}
Each path $Z^\ell$ in step~\ref{okdalg-zpath} is feasible.
\end{lemma}

\begin{proof}
We have $d(P)\leq B$ by construction, due to Lemma~\ref{okd-structhm} \ref{okd-len}, so
each $Z^\ell$ path satisfies the length budget. For each $j\in\dbrack{k'}$ and each 
$v\in Q^j-u_{j-1}$, we have 
\begin{equation}
\wt(P_{\rt,v})\leq\sum_{h=0}^{j-1}\Bigl(\wt(Q^h-u_{h-1}-u_h)+\wt_{u_h}\Bigr)+W_j+\wt_v
\leq\bigl(\lb_j+\wt_v\bigr)+\sum_{h=0}^{j-1}\wt_{u_h}+\zeta^j. \label{complv}
\end{equation}
The first inequality follows because $\wt(Q^j_{u_{j-1},v}-u_{j-1}-v)\leq W_j$, and the
second follows from the definition of $\lb_j$, 
since $\wt(Q^h-u_{h-1}-u_h)\leq W_h\leq\zeta^h$ for all $h=0,\ldots,j-1$.

Now fix $\ell\in\{0,1,2\}$. Consider an index $j\in\dbrack{k'}$ where $j=\ell\bmod 3$, and
some $v\in Q^j-u_{j-1}-u_j$. First suppose $j\geq 3$. We have 
\begin{equation*}
\begin{split}
\wt(Z^\ell_{\rt,v}) &\leq \wt(P_{\rt,u_{j-3}})+\zeta^j-\sum_{h=0}^{j-3}\wt_{u_h} \\
& \leq \lb_{j-3}+\zeta^{j-3}+\zeta^j \leq \lb_{j-3}+\zeta^{j-3}+\zeta^{j-2}+\zeta^{j-1}
\leq \lb_j.
\end{split}
\end{equation*}
The second inequality follows by applying \eqref{complv} to $v=u_{j-3}$; the third is
because $\zeta^j\leq\zeta^{j-1}+\zeta^{j-2}$; and the last inequality because from the
definition of $\lb_j$, we can infer that $\lb_j\geq\lb_{j-3}+\sum_{h=j-3}^{j-1}\zeta^h$.
Since $v\in V_j$, it follows that $\wt(Z^\ell_{\rt,v})\leq\knapd_v$.
Now suppose $j\leq 2$. Then $\wt(Z^\ell_{\rt,v})\leq W_j\leq\zeta^j\leq\lb_j\leq\knapd_v$.
\end{proof}

\begin{proofof}{Theorem~\ref{okd-quasipoly}}
Recall that we may assume that we have the right guess of the portal nodes. Then, the path
$Z'=\rt,u_0,\ldots,u_{k'}$ is clearly a feasible solution.
By Lemma~\ref{okdalg-feas}, each $Z^\ell$-path is also feasible. The $Z^\ell$s and $Z'$
cover $P$, so we return a feasible solution of reward at least
$\rewd(P)/4\geq\frac{\OPT}{8(\al+1)}$, where $\al=O(1)$.
The time needed to enumerate the portal nodes, and the auxiliary quantities mentioned in
Lemma~\ref{okd-structhm} is $(n+\log B)^{O(\log W)}$.
\end{proofof}

\subsubsection*{\boldmath Polytime $O\bigl(\log(\frac{W}{\knapd_{\min}})\bigr)$-approximation.}
Recall that $\knapd_{\min}=\min_{v\in V:\knapd_v>0}\knapd_v$ is the minimum non-zero
knapsack deadline. Let $N=\floor{\log\bigl(\frac{W}{\knapd_{\min}}\bigr)}$. Again, let
$P^*$ be an optimal \orientkd solution.
%The approximation follows from 
We use a simple bucketing idea, where we partition the vertices
into $N+2$ buckets based on their deadlines: for $j\in\dbrack{N}$, bucket $j$ consists of
the vertices $V_j:=\bigl\{v\in V: \frac{\knapd_v}{\knapd_{\min}}\in [2^j,2^{j+1})\bigr\}$, and
let $V_{-1}:=\{v\in V: \knapd_v=0\}$.
For each $j\in\{-1\}\cup\dbrack{N}$,
let $P^*_j$ be the rooted path obtained by shortcutting $P^*$ to retain only the nodes in
$V_j+\rt$.
Clearly, some bucket $P^*_j$ is responsible for at least a $\frac{1}{N+2}$-fraction of the
optimum. 
It is not hard to show that one can take a solution to the \knapo-instance with vertex-set
$V_j$ and knapsack budget $\knapd_{\min}\cdot 2^{j+1}$, and convert it to a feasible
\orientkd solution losing a constant-factor. This yields an $O(N)$-approximation. 

Formally, for each $j\in\dbrack{N}$, let $Q^j$ be an $\al$-approximate solution to the
\knapo instance with knapsack budget $\knapd_{\min}\cdot 2^j$, rewards $\pi_v$ for 
$v\in V_j$ and $0$ for all other $v$, and the same knapsack weights, metric, and length
budget as in the \orientkd instance. (Note that $\al\leq 4+\e$, for any fixed $\e>0$.)
By shortcutting, we may assume that $Q^j$ only visits nodes in $V_j+\rt$.
For $j=-1$, the above translates to solving an orienteering instance with $\pi_v$-rewards 
for nodes in $V_{-1}$ and $0$ rewards everywhere else, to obtain $Q^{-1}$, since all nodes
in $V_{-1}$ have $\knapd_v=0=\wt_v$.   

$Q^{-1}$ is clearly a feasible \orientkd solution. For $j\geq 0$, we can break up $Q^j$
into at most $3$ subpaths $Q^j_i$, $i=1,2,3$, with 
$\wt(Q^j_1), \wt(Q^j_3)\leq\knapd_{\min}\cdot 2^j$ and $Q^j_2$ consisting of a singleton
node. For each $i=1,2,3$, appending $\rt$ to $Q^j_i$ yields a feasible \orientkd solution,
since $\knapd_v\geq\max\{\wt_v,\knapd_{\min}\cdot 2^j\}$ for all $v\in Q^j-\rt$. So the
maximum-reward path among these three paths collects at least $\pi(Q^j)/3$ reward. 
Also, note that $P^*_j$ is a feasible solution to the \knapo instance for index $j$, 
so $\pi(Q^j)\geq\pi(P^*_j)/\al$. 

We return the best of all the feasible solutions $Q^{-1}$, $\{Q^j_i\}_{j\in\dbrack{N},i\in[3]}$.
From the above, we obtain that this solution achieves reward at least $\pi(P^*)/3(N+2)$.

\begin{theorem} \label{okd-logthm}
There is a polytime $O\bigl(\log(\frac{W}{\knapd_{\min}})\bigr)$-approximation algorithm
for \orientkd, where $W$ is the maximum knapsack deadline and $\knapd_{\min}$
is the minimum non-zero knapsack deadline.
\end{theorem}

\section{\boldmath \csko with cancellations} \label{csko-cancel}
In {\em \csko with cancellations} (\cskocancel), the input is the same as in \csko, 
%instance of \csko. 
%The difference with \csko is that 
but we are now allowed to {\em cancel} the
processing of the current vertex $v$ %in a solution 
at any (integer) timestep before its size and reward get fully realized; if $v$
is cancelled, then no reward is collected from $v$ and we cannot process $v$ again.
%cancelled vertex and that vertex can not be processed again; 
As with \csko, we only collect reward from vertices that complete by the processing-time 
horizon $W$.   
Gupta et al.~\cite{GuptaKMR11} showed that even for correlated knapsack (which is the
special case of \csko where all vertices are co-located), the optimal reward when we allow
cancellations can be substantially larger than the optimal reward without cancellations,
so we need to develop new algorithms to handle cancellations.

We obtain the same guarantees for \cskocancel %with cancellations 
as for \csko: that is, 
$O(\log\log W)$-approximation in $(n+\log B)^{O(\log W\log\log W)}$ time, 
%(i.e., almost quasi-polytime), 
and a polytime $O(\log W)$-approximation.

We proceed as follows. Recall %(from Section~\ref{csko-refine})
that for a vertex $v$, we define $\excep:=\Rv\bon_{\Sv>W/2}$
and $\trunc:=\Rv\bon_{\Sv\leq W/2}$. %Letting $\I$ denote the original \csko instance, 
%We can consider the modified \csko instances 
Let $\excep[\I]$ and $\trunc[\I]$ denote the \cskocancel instances where the
rewards are given by $\{\excep\}_{v\in V}$ and $\{\trunc\}_{v\in V}$ respectively. 
As observed by~\cite{GuptaKMR11}, cancellations do not help for the instance $\excep[\I]$,
i.e., the optimal reward is the same both with and without cancellations. This is because
if a policy cancels a vertex $v$ after it has run for some $t\leq W/2$ time steps, we can
modify the policy to not process $v$ at all, without decreasing the reward accrued from 
subsequently-processed vertices; if $v$ is cancelled after it has run for more than $W/2$
time steps, then both with and without cancellation, the policy cannot collect any further
reward. 

We show that we can obtain an $O(1)$-approximation for $\trunc[\I]$. Then, with
probability $0.5$, we can work on the instance $\trunc[\I]$, where we utilize this
$O(1)$-approximation, or the instance $\excep[\I]$, where we utilize the approximation
results for \csko.  
So %coupled with the approximation results for \csko (which we can use for $\excep[\I]$)
this yields:
%Thus, for the instance $\excep{I}$, we can obtain: 
an $O(\log\log W)$-approximation in quasi-polytime, and a polytime 
$O(\log W)$-approximation. 

%\vspace*{-1ex}
So we focus on obtaining an $O(1)$-approximation for \cskocancel instances of the form 
$\trunc[\I]$. 
%that is, where we only collect non-zero reward from small (i.e., size at most $W/2$)
%instantiations of vertices. 
Our approach is based on LP-rounding, by combining the LP-rounding approaches for
orienteering %We combine the LP-relaxation for orienteering
in~\cite{FriggstadS17} %with the LP-relaxation for 
and the correlated knapsack problem with cancellations in~\cite{GuptaKMR11}.
We combine the LP-relaxations for these two problems to obtain an LP \eqref{cskoclp},
whose optimal value yields an upper bound on the optimal reward. We round an optimal
solution to \eqref{cskoclp} in two phases. 
We first {\em extract a suitable knapsack orienteering instance from the LP solution}, and 
use Theorem~\ref{knapo-round} to obtain a good rooted path $Q$ for this \knapo
instance. Next, we focus on the vertices in $Q$, and extract a good LP solution to the
correlated knapsack problem restricted to vertices in $Q$. We utilize the
LP-rounding result in~\cite{GuptaKMR11} to round this solution to obtain a \cskocancel
solution that visits the vertices in $Q$ in order, potentially cancelling some vertices
along the way. Thus, we obtain a non-adaptive policy for \cskocancel.
In the second step, we crucially leverage an important aspect of the LP-rounding algorithm
in~\cite{GuptaKMR11} for correlated knapsack with cancellations, namely that it is 
{\em order oblivious}: it's guarantee does not 
depend on the order in which the vertices (i.e., items in correlated knapsack) are
considered. This flexibility %in their algorithm 
allows us to consider vertices in $Q$ in
the order they are visited, and thereby ensure that the travel-budget constraint is
satisfied. (We remark that for the correlated knapsack without cancellations, we do not have this
flexibility, %in correlated knapsack, 
when considering large instantiations; see
Appendix~\ref{append-corrknap}. This lack of flexibility is the main obstacle in obtaining
a good solution from large instantiations in \csko.)

\vspace*{-1ex}
\subparagraph*{LP relaxation.}
We use $\Ru(t)$ to denote the reward $\Ru$ when the size $\Su$ is $t$; this is $0$
if $\Pr[\Su=t]=0$. Note that $\Ru(t)=0$ for all $t>W/2$, since we are considering
$\trunc[\I]$. 
\begin{alignat}{3}
\max && \quad 
\sum_{u \in V}\sum_{t=1}^{W/2}z_{u,t}&\cdot\Pr[S_u=t\,|\,\Su\geq t]\cdot \Ru(t)
\tag{CKOC-LP} \label{cskoclp} \\
\text{s.t.}  && \quad 
%z^v_u = 
x^v\bigl(\delta^{\into}(u)\bigr) & \geq x^v\bigl(\delta^{\out}(u)\bigr) \qquad 
&& \forall u\in V-\rt,\,v\in V \tag{\ref{pref-visit}} \\
%&& x^v(\delta^{in}(u))&=z^v_u \qquad&&\forall u,v\in V\tag{OR2}\label{oc2}\\
&& x^v\bigl(\delta^{\into}(S)\bigr) & \geq z^v_u \qquad && 
\forall v\in V,\, S\subseteq V-\rt,\, u\in S \tag{\ref{subtour}} \\
&& z^v_u & = 0 && \forall u,v\in V: d_{\rt,u}>d_{\rt,v} \tag{\ref{dbnd}} \\
&& \sum_{a\in A}d_a\cdot x_a^v & \leq Bz_{v}^v, \quad 
x^{v}\bigl(\delta^{\out}(\rt)\bigr) = z^v_v\qquad && \forall v\in V 
\tag{\ref{dbudget}} \\
%&& x^{v}\bigl(\delta^{\out}(\rho)\bigr) & = z^v_v &&\forall v\in V\tag{OR6}\label{oc6}\\
&& \sum_{v\in V}z_v^{v} & = 1 \tag{\ref{unit}} \\
&& \sum_{v\in V}z_{u}^v & = z_{u,0} \qquad && \forall u\in V \label{link} \\
&& z_{u,t} &= s_{u,t}+z_{u,t+1} \qquad && \forall u\in V,\,t\in\dbrack{W} \tag{CK1} 
\label{ck1} \\
&& s_{u,t} & \geq \Pr[\Su=t\,|\,\Su\geq t]\cdot z_{u,t} \qquad && 
\forall u\in V,\,t\in\dbrack{W} \tag{CK2} \label{ck2} \\
&& \sum_{u\in V}\sum_{t=0}^W t\cdot s_{u,t} & \leq W \tag{CK3} \label{ckbudg} \\
&& x,z,s &\geq 0. \notag
\end{alignat}
\clonelabel{lp:csoc}{cskoclp}%
\clonelabel{kn1}{link}%
\clonelabel{kn2}{ck1}%
\clonelabel{kn3}{ck2}%
\clonelabel{kn4}{ckbudg}%
As noted earlier, this LP is a combination of the LPs for orienteering
in~\cite{FriggstadS17} and correlated knapsack with cancellations in~\cite{GuptaKMR11}.
The $x^v_a$ and $z^v_u$ variables, and constraints \eqref{pref-visit}--\eqref{unit}
are from the LP for rooted orienteering (and also present in LP \eqref{kolp} for
\knapo). They encode the arcs included, and vertices visited, respectively by the
rooted path, provided that $v$ is the furthest node from $\rt$ that is visited.
The $z_{u,t}$ and $s_{u,t}$ variables, and constraints \eqref{ck1}--\eqref{ckbudg} are from the LP for
correlated knapsack with cancellations. For any vertex $u$ and $t\geq 0$, variable $z^u_t$ 
encodes that $u$ is processed for a least $t$ time units, and $s_{u,t}$ encodes that $u$
is processed for {\em exactly} $t$ time units.
Thus, variable $z_{u,0}$ encodes that $u$ is visited, and constraint \eqref{link} links
the orienteering and correlated knapsack LPs.  
The objective function is from the LP in~\cite{GuptaKMR11} for correlated knapsack
with cancellations. 
%We only go up $t=W/2$ in the objective function, since $\Ru(t)=0$ for
%all $t>W/2$ as we are working with $\trunc[\I]$.

We remark that~\cite{GuptaKMR11} showed that one can replace \eqref{ck1}--\eqref{ckbudg}
with a polynomial-size formulation losing an $O(1)$-factor, which applies here as well.

\newcommand{\optcskoclp}{\ensuremath{\OPT_{\text{\ref{cskoclp}}}}\xspace}
Let $(\bx,\bz,\bs)$ be an optimal solution to \eqref{cskoclp} and $\optcskoclp$ be its
objective value.
As argued in the proof of Theorem~\ref{knaporedn}, constraints
\eqref{pref-visit}--\eqref{unit} are valid because any rooted path $Q$, corresponding to
an execution of an adaptive policy, satisfies these constraints, where the superscript $v$
in the non-zero variables is the furthest node from $\rt$ on $Q$.
Gupta et al.~\cite{GuptaKMR11} show (see Theorem 3.1 in~\cite{GuptaKMR11}) that
constraints \eqref{ck1}--\eqref{ckbudg} are valid 
for correlated knapsack with cancellations, where $z_{u,t}$ and $s_{u,t}$ are now
respectively the probabilities that $u$ is visited for at least $t$, and exactly $t$ time
steps. They also argue that the objective function provides an upper bound on the expected
reward obtained. So we have the following.

%We leave the proof of the following two lemmas to appendix \ref{appendix:roundlps}.

\begin{lemma} \label{lpvalid}
%The optimal value of \eqref{cskclp}, 
\optcskoclp is at least the optimal reward for the \cskocancel instance. 
\end{lemma}

\begin{theorem} \label{cskoc-small-thm}
We can round $(\bx,\bz,\bs)$ to obtain a non-adaptive $O(1)$-approximation algorithm for
\cskocancel. 
\end{theorem}

\begin{proof}
Consider the \knapo instance defined by the same metric and travel budget, with
knapsack weights $\wt_u:=\bigl(\sum_{t=0}^W t\cdot\bs_{u,t}\bigr)/\bz_{u,0}$ 
and rewards $\hpi_u:=\bigl(\sum_{t=1}^{{W}/{2}}\bz_{u,t}\cdot\Pr[S_u=t\,|\,\Su\geq t]\cdot \Ru(t)\bigr)/\bz_{u,0}$
for all $u\in V$, and knapsack budget $W$.
Then, due to \eqref{link} and\eqref{ckbudg}, one can infer that $\bigl(\bx,\{\bz^v_u\}_{u,v\in V}\bigr)$ is
a feasible solution to the LP-relaxation \eqref{kolp} for this \knapo instance, of value
$\optcskoclp$. By Theorem~\ref{knapo-round}, we can therefore obtain a rooted path $Q$
with $d(Q)\leq B$, $\wt(Q)\leq W$, and $\hpi(Q)\geq\optcskoclp/5$.

Now, we select a subsequence of $Q$ to visit 
%in the order they appear on $Q$, 
by solving a correlated knapsack with cancellations problem involving only vertices in
$Q$. Consider the following LP-relaxation for this problem in~\cite{GuptaKMR11}.
\begin{alignat}{3}
\max && \quad \sum_{u \in Q}\sum_{t=1}^{{W}/{2}}z_{u,t}&\cdot\Pr[S_u=t\,|\,\Su\geq t]\cdot \Ru(t)
\tag{CK-LP} \label{cklp} \\
\text{s.t.} && \quad z_{u,0} & = 1 \qquad && \forall u\in Q \label{cork1} \\
&& z_{u,t} & = s_{u,t}+z_{u,t+1} \qquad && \forall u\in Q,\,t\in\dbrack{W} \label{cork2} \\
&& s_{u,t} & \geq \Pr[\Su=t\,|\,\Su\geq t]\cdot z_{u,t} \qquad &&
\forall u\in Q,\,t\in\dbrack{W} \label{cork3} \\
&& \sum_{u\in Q}\sum_{t=0}^Wt\cdot s_{u,t} & \leq W \label{cork4} \\
&& z,s &\geq 0. \notag
\end{alignat}

Consider the solution where we set $\tz_{u,t}=\frac{\bz_{u,t}}{\bz_{u,0}}$ and
$\ts_{u,t}=\frac{\bs_{u,t}}{\bz_{u,0}}$ for all $u\in Q$, $t\in\dbrack{W}$. We claim that
this is feasible to \eqref{cklp}. Constraint \eqref{cork1} clearly holds, and
\eqref{cork2}, \eqref{cork3} hold due to \eqref{ck1}, \eqref{ck2}.
The LHS of \eqref{cork4} evaluates to $\wt(Q)$, so \eqref{cork4} holds since 
$\wt(Q)\leq W$. The objective value of $(\tz,\ts)$ is precisely $\hpi(Q)$. 
We can now use the LP-rounding algorithm for correlated knapsack with cancellations
in~\cite{GuptaKMR11} to process vertices, with cancellations, in the order they appear on
$Q$. This yields expected reward at least $\hpi(Q)/8\geq\optcskoclp/40$.
\end{proof}

\section{\boldmath $O(\log\log B)$-approximation for \cso} \label{cso-approx}
We show that our approach in Section~\ref{csko-approx} for \csko can be utilized to yield
the guarantees mentioned in Theorem~\ref{cso-approxthm}: that is,
an $O(\al\log\log B)$-approximation algorithm for \cso %correlated orienteering
in time $(n+\log B)^{O(\log B\log\log B)}\cdot\tim$, where $\tim$ is the
running time of the given $\al$-approximation algorithm for deadline TSP.  
%This improves upon the approximation guarantee in~\cite{BansalN14},
%which also runs in quasi-polytime, by an 
%$O\bigl(\frac{\log\log B}{\log\log\log B}\bigr)$-factor. 
The algorithm in%%Friggstad and Swamy
~\cite{FriggstadS21} for deadline TSP 
%devised an almost quasi-polytime algorithm for deadline TSP,
translates to an $O(1)$-approximation in $n^{O(\log B)}$ time, so this implies an
$O(\log\log B)$-approximation for \cso in quasi-polytime.
%This will prove Theorem~\ref{cso-approxthm}.
We also simplify the exposition significantly by making use of monotone-reward
TSP~\cite{FriggstadS21} as a subroutine.  

Throughout this section, let $K=3\log(6\log B)+12$, $L=\ceil{\log B}$, $N_1=2(K+1)$. 
Define $\nd_{-1}:=\rt$, and
$\piv(t):=\E{\nrewd\cdot\bon_{\nsize\leq B-t}}=\sum_{t^\prime=0}^{B-t}\Pr[\Sv=t^\prime]\cdot\E{\Rv\,|\,\Sv=t^\prime}$.
Let $\optcso$ denote the optimal reward for the \cso instance.
We may assume that $\pi_v(d_{\rt,v})\leq\optcso/4$ for every $v\in V$, as otherwise, we can
obtain $\Omega(\optcso)$ reward by going to a single node.
The algorithm in~\cite{BansalN14} is based on the following structural result, which we
have paraphrased (and corrected slightly) to conform to our notation. 

\begin{lemma}[\textnormal{Lemma 3.6 in~\cite{BansalN14}}]  \label{cso-simpstructhm} 
There exists a rooted path $P$ with $d(P)\leq B$, and 
vertices $\nd_0\preceq\nd_1\preceq\ldots\preceq\nd_k$ on $P$ for some $k\leq L$, such that: 
\begin{enumerate}[label=(\alph*), topsep=0.1ex, noitemsep, leftmargin=*]
\item %\label{simprewd} %(Reward) 
$\sum_{j=0}^k\sum_{v\in P_{\nd_{j-1},\nd_j}-\nd_j}\pi_v\bigl(d(P_{\rt,v})+2^j-1\bigr)\geq\optcso/4$; and  
\item %\label{simpprefsize} %(Prefix-size) 
$\mu^j(P_{\rt,\nd_j}-{\nd_j})\leq (K+1)2^j$ for all $j\in\dbrack{k}$.
\end{enumerate} 
\end{lemma}

We refine this by subdividing each $P_{\nd_{j-1},\nd_j}$ subpath into at most $2(K+1)$
segments each of $\mu^j$-weight at most $2^j$, and by guessing the two-point regrets of
these segments, 
%in a fashion similar to how Theorem~\ref{strucdthm} is obtained from
%(Theorem~\ref{structhm} and hence) Theorem~\ref{simpstructhm} 
to obtain a structural result analogous to Theorem~\ref{strucdthm}.

\begin{theorem}[Structural result for \cso] \label{cso-strucdthm} 
Let the rooted path $P$ and node-sequence $\nd_0,\ldots,\nd_k$, where $k\leq L$, be as in
Lemma~\ref{cso-simpstructhm}. 
For each $j\in\dbrack{k}$, there is a vertex-set 
$\ndset_j\sse P_{\nd_{j-1},\nd_j}$ containing $\nd_{j-1}$, $\nd_j$, with 
$|\ndset_j|\leq N_1$, whose nodes are ordered by the order they appear on $P$, 
and
%Let $\fullset:=\bigcup_{j=0}^k\ndset_j$, which we call ``portal nodes'',
for every pair of consecutive nodes $a,b\in\fullset:=\bigcup_{j=0}^k\ndset_j$, 
there is a path $\spath_{a,b}$, node $m_a$, integer $\gm_a\geq 0$, satisfying the
following properties. For $a,b\in\fullset$, let $\nxt(a)$ be the next node in
$\fullset$ after $a$, for $a\neq\nd_k$; let $b\prec a$ if $b$ comes before $a$. 
%in $\fullset$. 
%
\begin{enumerate}[label={\textnormal{(C\arabic*)}}, %ref={\ref{strucdthm}\,(\alph*)},
topsep=0.2ex, noitemsep, leftmargin=*]
\item \textnormal{(Distance)} $d(\spath_{a,b})\leq\dist_{a}:=2^{\gm_a}-1+d(a,m_a)+d(m_a,b)$
for every pair of consecutive nodes $a,b\in\fullset$.
\label{cso-distance}

\item \textnormal{(Total-length)} $\sum_{a\in\fullset-\nd_k}\dist_{a}\leq B$.
\label{cso-totlen}

\item \textnormal{(Reward)}
$\sum_{j=0}^K\sum_{a\in\ndset_j-\nd_j}\sum_{v\in\spath_{a,\nxt(a)}-\nxt(a)}
\pi_v\bigl(\sum_{b\prec a}\dist_b+d(\spath_{a,v})+2^j-1\bigr)
\geq\optcso/8$.  
\label{cso-reward}

\item \textnormal{(Prefix-size)}
$\sum_{h=0}^j\sum_{a\in\ndset_h-\nd_h}\mu^j\bigl(\spath_{a,\nxt(a)}-\nxt(a)\bigr)\leq (K+1)2^j$ 
for all $j\in\dbrack{k}$. 
\label{cso-prefix}

\item \textnormal{(Size)} 
$\mu^j(\spath_{a,b}-b)\leq 2^j$ for every $j\in\dbrack{k}$ and pair of consecutive nodes
$a,b\in\ndset_j$. 
\label{cso-size}  
\end{enumerate} 
\end{theorem} 

Note that only the (Reward) property \ref{cso-reward} is different from the (Reward)
property in Theorem~\ref{strucdthm} for \csko; the remaining properties are the same. We
now exploit this structural result in much the same way as in Section~\ref{csko-quasipoly}.

\vspace*{-1ex}
\subparagraph*{Configuration LP.}
Again, to avoid excessive notation, assume that we have found, by enumeration, the
``portal nodes'' $\fullset$ and length bounds $\dist_a$ for all $a\in\fullset-\nd_k$
satisfying Theorem~\ref{cso-strucdthm}.
To capture the reward obtained from an $a$-$\nxt(a)$ path,
the configurations $\I_a$ for a node $a\in\fullset-\nd_k$ will now consist of feasible
solutions to {\em \ptp knapsack monotone-reward TSP}, which is the {knapsack-constrained
version of \ptp-monotone-reward TSP}: 
they are simply all $a$-$\nxt(a)$ paths $\tau$ with 
$\mu^j\bigl(\tau-\nxt(a)\bigr)\leq 2^j$. The length budget $\dist_a$ will be captured
implicitly, by defining the reward function 
of a node $v\neq\nxt(a)$ to be 
$\pi_v\bigl(\sum_{b\prec a}\dist_b+d(\tau_{a,v})+2^j-1\bigr)$
if $d(\tau_{a,v})\leq\dist_a-d(v,\nxt(a))$ and $0$ otherwise, which is a non-increasing
function of 
$d(\tau_{a,v})$; for $v=\nxt(a)$, the reward function is identically $0$.
For notational convenience, define $\pi^{a,j}(\tau)$ to be the total reward obtained from
nodes in $\tau$ under the above rewards.

The configuration LP for \cso now has the same constraints as \eqref{cskolp}, but the
objective function changes to 
$\max\ \sum_{j=0}^k\sum_{a\in\ndset_j-\nd_j}\sum_{\tau\in\I_a}x^a_\tau\cdot\pi^{a,j}(\tau)$.
Let \csolp denote this LP, and \optcsolp denote its optimal value.
Similar to Claim~\ref{cskolpval} and Lemma~\ref{cskolpsolve}, we have the following. 

\begin{lemma} \label{csolpsolve}
We have $\optcsolp\geq\optcso/8$.
Given an $\al$-approximation algorithm for deadline TSP, we can compute in polytime
a \csolp-solution $\bx$ of objective value at least $\optcsolp/O(\al)$.
\end{lemma}

As with Lemma~\ref{cskolpsolve}, we defer the proof of the above lemma to
Section~\ref{lpsolve}. 

\vspace*{-2ex}
\subparagraph*{LP-rounding and non-adaptive algorithm.}
For any $a\in\fullset-\nd_k$, due to our definition of the node-reward functions in the
corresponding monotone-reward TSP instance, we may assume (by shortcutting) that every
path $\tau\in\I_a$ with $\bx^a_\tau>0$ has $d(\tau)\leq\dist_a$.
The LP-rounding algorithm and conversion to a non-adaptive policy are %proceed 
{\em exactly} as in Algorithm~\ref{cskoalg}. %Section~\ref{csko-quasipoly}. 
The analysis is similar, but we now analyze the reward on a path-by-path basis, 
considering the reward obtained from paths $\tau\in\I_a$, for each $a\in\fullset-\nd_k$. 
%also esentially the same, with cosmetic
%changes, mainly to terminology and notation.
%proceed in essentially the exact same way as in Section~\ref{csko-quasipoly}. 

%For $a,b\in\fullset$, we say that $b\prec a$, if $b$ comes before $a$.
Consider a vertex $v\in V$. 
%For $a\in\fullset-\nd_k$, define
%$\by^a_v:=\sum_{\tau\in\I_a:v\in\tau-\nxt(a)}\bx^a_\tau$. Note that
%$\sum_{a\in\fullset-\nd_k}\by^a_v\leq 1$.
We say that $v$ is ``visited by segment $a$'' if $v\in\rpath_{a,\nxt(a)}-\nxt(a)$, 
and $v$ is ``retained by segment $a$'' if $v\in\rpath'_{a,\nxt(a)}-\nxt(a)$; the
latter events are disjoint for different $a\in\fullset-\nd_k$. 
(If $v=a$, both events happen with probability $1$.)
Again, we view step~\ref{shortcut} as being executed regardless of the outcome of
step~\ref{prefixchk}. 
%(For $a\in\fullset-\nd_k$,
%both events happen with probabilty $1$.)
As before, for $j\in\dbrack{k}$, let $\Bad_j$ be the event that 
$\sum_{h=0}^j\sum_{a\in\ndset_h-\nd_h}\mu^j\bigl(\rpath_{a,\nxt(a)}-\nxt(a)\bigr)>5(K+1)2^j$,
and $\Bad:=\bigvee_{j=0}^k\Bad_j$ be the event that step~\ref{prefixchk} fails; let $\cBad$
denote the complement of $\Bad$. Lemma~\ref{prbad} continues to hold as is.

%The following lemmas are proved in Appendix~\ref{append-cso} by mimicking the proofs of
%Lemmas~\ref{nodecovlem} and~\ref{nanodelem} respectively.

\begin{lemma} \label{cso-pathcov}
Consider any $a\in\fullset-\nd_k$, $\tau\in\I_a$, and any $v\in\tau-\nxt(a)$.
Then we have 
$\Pr[\{\rpath_{a,\nxt(a)}=\tau,\,v\in\rpath'_{a,\nxt(a)}\}\wedge\cBad]\geq\frac{\bx^a_\tau}{8}$.
\end{lemma}

\begin{proof}
We follow a similar strategy as in the proof of Lemma~\ref{nodecovlem}.
Let $\Ec$ denote the event $\rpath_{a,\nxt(a)}=\tau$. We have $\Pr[\Ec]\geq 0.5\bx^a_{\tau}$.
Let $\Ret$ denote the event ``$v$ is retained by segment $a$'', i.e., $\{v\in\rpath'_{a,\nxt(a)}\}$.
Let $\Vis_b$ denote the event ``$v$ is visited by segment $b$'' for $b\in\fullset-\nd_k$.
The stated probability is equal to
\begin{equation*}
%\begin{split}
\Pr[\Ec\wedge\Ret]-\Pr[\Ec\wedge\Ret\wedge\Bad] 
\geq \Pr[\Ec\wedge\Ret]-\Pr[\Ec\wedge\Bad]
= \Pr[\Ec]\cdot\Bigl(\Pr[\Ret\,|\,\Ec]-\Pr[\Bad\,|\,\Ec]\Bigr).
%\end{split}
\end{equation*}
We show that $\Pr[\Ret\,|\,\Ec]\geq 0.5$ and $\Pr[\Bad\,|\,\Ec|\leq 1/\poly(\log B)$,
which yields the lemma. 

To bound $\Pr[\Ret\,|\,\Ec]$, if $v=a$, then this is $1$. So suppose otherwise. (Note that
then $v\neq\nxt(b)$ for any $b\in\fullset-\nd_k$.)
%Otherwise 
Since $v\in\tau$, conditioned on $\Ec$, $v$ is retained by segment $a$ if $v$ is
not visited by any segment $b\prec a$. We have  
$\Pr[\Vis_b]\leq 0.5\sum_{\tau'\in\I_b:v\in\tau'}\bx^b_{\tau'}$ %\leq 0.5\by^b_v$
for any $b$, so $\Pr[\Ret\,|\,\Ec]=\prod_{b\prec a}(1-\Pr[\Vis_b])\geq 0.5$ by
Claim~\ref{helper} since $\sum_{b\prec a}\Pr[\Vis_b]\leq 0.5$.

We bound $\Pr[\Bad_f\,|\,\Ec]$, let $j$ be the index such that $a\in\ndset_j$.
We first note that if $f<j$, then $\Bad_f$ and $\Ec$ are independent, and the upper
bound follows from Lemma~\ref{prbad}. So suppose $f\geq j$.
As in the proof of Lemma~\ref{prbad}, 
define $Z_b=\mu^f\bigl(\rpath_{b,\nxt(b)}-\nxt(b)\bigr)/2^f$ for all
$b\in\bigcup_{h=0}^f(\ndset_h-\nd_h)$.
So $\Pr[\Bad_f\,|\,\Ec]=\Pr\bigl[\sum_{h=0}^f\sum_{b\in\ndset_h-\nd_h}Z_b>5(K+1)\,|\,\Ec\bigr]$. 
Observe that $Z_a\leq 1$. So 
$\Pr[\Bad_f\,|\,\Ec]\leq
\Pr\bigl[\sum_{h=0}^f\sum_{b\in\ndset_h-\nd_h: b\neq a}Z_b>5(K+1)-1\bigr]$ 
and by Chernoff bounds, the latter probability is at most 
$\bigl(\frac{e^3}{4^4}\bigr)^{K+1}\leq e^{-(K+1)}$;
hence $\Pr[\Bad\,|\,\Ec|\leq 1/\poly(\log W)$. 
\end{proof}

%\begin{lemma} \label{cso-nodecovlem}
%For any node $v\in V$ and any $a\in\fullset-\nd_k$, we have
%$\Pr[\{\text{$v$ is retained by segment $a$}\}\wedge\cBad]\geq\frac{\by^a_v}{16}$.
%\end{lemma}
Mimicking the proof of Lemma~\ref{nanodelem}, we obtain the following.

\begin{lemma} \label{cso-nanodelem}
Consider any node $v\in V-\nd_k$. 
Suppose that $v$ is retained by segment $a\in\ndset_j-\nd_k$ in step~\ref{shortcut}, 
for some index $j\in\dbrack{k}$. 
Then 
$\Pr[v\in\rpath'',\,\sum_{w\in\rpath'':w\prec_{\rpath'}v}\Sw\leq 2^j-1]\geq\frac{1}{20(K+1)}$
%$\Pr[\text{total size of nodes on $\rpath''$ before $v$}\leq 2^j-1$}]\geq\frac{1}{20(K+1)}$, 
where the probability is over both the random sampling in step~\ref{poloutput} 
%to obtain  $\rpath''$ 
and the random execution of $\rpath''$, and is conditioned on the state after
step~\ref{shortcut}. 
\end{lemma}  

\begin{proof} %{Lemma~\ref{cso-nanodelem}}
Recall that $\mu^j_w=\E{X^j_w}$ and $X^j_w=\min\{\nsize[w],2^j\}$ for every node $w$.
Let $S$ be the set of nodes in $\rpath'$ before $v$, and $T\sse S$ be the random set of
nodes from $S$ included in $\rpath''$. (Note that $S$ is fixed since we are conditioning on
the state after step~\ref{shortcut}.) 
The given probability is 
$\Pr\bigl[v\in\rpath''\bigr]\cdot\Pr\bigl[\sum_{w\in T}\Sw\leq 2^j-1\,|\,v\in\rpath''\bigr]$.

%where, throughout, all probabilities are conditioned on the state after
%step~\ref{shortcut}. 
Since $v$ is retained by segment $j$, in particular, we have
$v\in\rpath'$; so $\Pr\bigl[v\in\rpath''\bigr]=\frac{1}{10(K+1)}$.
Note that $T$ is independent of the event $\{v\in\rpath''\}$. So we have
\begin{equation*}
\begin{split}
\Pr\Bigl[\sum_{w\in T}\nsize[w]\leq 2^j-1\Bigr] 
= 1-\Pr\Bigl[\sum_{w\in T}\nsize[w]\geq 2^j\Bigr]
& = 1-\Pr\Bigl[\sum_{w\in T}\min\{\nsize[w],2^j\}\geq 2^j\Bigr] \\
& \geq 1-\frac{\E{\sum_{w\in T}X^j_w}}{2^j}.
\end{split}
\end{equation*}
Since step~\ref{prefixchk} succeeds, we know that $\mu^j(S)\leq 5(K+1)2^j$. Since
$\Pr[w\in T]=\frac{1}{10(K+1)}$, we obtain that 
$\E{\sum_{w\in T}X^j_w}=\frac{\mu^j(S)}{10(K+1)}\leq 2^{j-1}$.
%we have $\bigcup_{h=0}^j\bigcup_{a\in\ndset_h-\nd_h}\bigl(\rpath_{a,\nxt(a)}-\nxt(a)\bigr)$,
Putting everything together, we obtain the desired statement.
\end{proof}

\begin{proofof}{Theorem~\ref{cso-approxthm}} 
%In order to solve \csolp using an $\al$-approximation algorithm for deadline TSP, we
We define a random variable $\Rewd^{a,\tau}_v$, for $a\in\fullset-\nd_k$, $\tau\in\I_a$,
and $v\in\tau$, which gives a pessimistic estimate of the
expected reward obtained from $v$
%the nodes in $\tau\cap\rpath''$ 
if $\tau$ is sampled in step~\ref{randround}, and $v\in\rpath''$.
For $a\in\fullset-\nd_k$, define $\visit_a:=\sum_{b\prec a}\dist_b$.
Let
%\begin{equation*}
$\Rewd^{a,\tau}_v=\pi_v\bigl(\visit_a+d(\tau_{a,v})+2^j-1\bigr)$ %\quad
%\text
if $\rpath_{a,\nxt(a)}=\tau$, $\cBad$ happens, $v\in \rpath''$, and
$\sum_{w\in\rpath'':w\prec_{\rpath'}v}\Sw\leq 2^j-1$; %\qquad  
%\text{
otherwise set $\Rewd^{a,\tau}_v=0$.
%\bon_{\rpath_{a,\nxt(a)}=\tau}\cdot\bon_{v\in\rpath'_{a,\nxt(a)}}\cdot\bon_{\cBad}
%\cdot\bon_{v\in\rpath''}\bon_{\sum_{w\in\rpath'':w\prec_{\rpath'}v}\Sw\leq 2^{j-1}}
%\cdot\pi_v\bigl(\visit_a+d(\tau_{a,v})+2^{j-1}\bigr).
%\label{cso-rewd}
%\end{equation*}
Observe that $\sum_{a\in\fullset-\nd_k}\sum_{\tau\in\I_a}\sum_{v\in\tau}\Rewd^{a,\tau}_v$
is a lower bound on the reward returned, since 
$\visit_a+d(\tau_{a,v})+2^j-1$ is an upper bound on the start time of $v$, as
$d(\rpath_{b,\nxt(b)})\leq\dist_b$ for all $b\in\fullset-\nd_k$.
Taking expectations, we obtain that
\begin{equation}
\begin{split}
\E{&\Rewd^{a,\tau}_v} =%\sum_{v\in\tau}
\pi_v\bigl(\visit_a+d(\tau_{a,v})+2^j-1\bigr)\cdot
\Pr\bigl[\{\rpath_{a,\nxt(a)}=\tau,\,v\in\rpath'_{a,\nxt(a)}\}\wedge\cBad\bigr]\cdot \\
&\quad \Pr\Bigl[v\in\rpath'',\,{\textstyle \sum_{w\in\rpath'':w\prec_{\rpath'}v}}\Sw\leq 2^j-1\,\Bigl\vert\,
\{\rpath_{a,\nxt(a)}=\tau,\,v\in\rpath'_{a,\nxt(a)}\}\wedge\cBad\Bigr]
\end{split}
\label{cso-exprewd}
\end{equation}
By Lemma~\ref{cso-pathcov}, we have 
$\Pr[\{\rpath_{a,\nxt(a)}=\tau,\,v\in\rpath'_{a,\nxt(a)}\}\wedge\cBad]\geq\frac{\bx^a_{\tau}}{8}$.
For the second probability term on the RHS of \eqref{cso-exprewd}, we can remove the
conditioning on $\rpath_{a,\nxt(a)}=\tau$, since both events depends only on the state
after step~\ref{shortcut}, that is on $\rpath'$ and the event $\Bad$.
Lemma~\ref{cso-nanodelem} shows that this probability is at least $\frac{1}{20(K+1)}$.

%So the total reward obtained is at least
%$\sum_{a\in\fullset-\nd_k}\sum_{\tau\in\I_a}\sum_{v\in\tau}\frac{\bx^a_{\tau}}{8}\cdot$
Substituting these quantities in 
$\sum_{a\in\fullset-\nd_k}\sum_{\tau\in\I_a}\sum_{v\in\tau}\E{\Rewd^{a,\tau}_v}$, we
obtain that the total expected reward obtained
is at least (objective value of $\bx$)/$O(K)$.
\end{proofof}

\section{\boldmath Solving the LPs (\ref{cskolp}) and \csolp} \label{lpsolve}

\begin{proofof}{Lemma~\ref{cskolpsolve}}
Let $\Alg$ be the given $\al$-approximation algorithm for \knapo.
Since \eqref{cskolp} has an exponential number of variables, we consider the dual LP,
which is as follows.
\begin{alignat}{3}
\min & \quad 
& \sum_{a\in\fullset-\nd_k}\beta_a+\sum_{v\in V}\tht_v+\sum_{j=0}^k(K+1)&2^j\cdot\kp_j
\tag{CKO-D} \label{cskodual} \\
\text{s.t.} & \quad & 
\sum_{v\in\tau-\nxt(a)}\Bigl(\pi_v(2^j-1)-\tht_v-\sum_{h=j}^k\mu^h_v\kp_h\Bigr) & \leq \beta_a \qquad
&& \forall j\in\dbrack{k},\,a\in\ndset_j-\nd_j,\,\tau\in\I_a \label{dualko} \\
%&& \gm^h_v+\tht_v+\sum_{j=h}^k\mu^j_v\kp_j & \geq \pi_v(2^h-1) \qquad
%&& \forall h\in\dbrack{k}\,v\in V \label{dualother} \\
&& \tht,\kp & \geq 0. \label{dualnoneg}
\end{alignat}
The $\beta_a$, $\tht_v$, $\kp_j$ dual variables correspond respectively to 
primal constraints \eqref{config}, \eqref{covbnd}, and \eqref{prefixlp}, and
constraints \eqref{dualko} correspond to to the $x^a_\tau$ primal variables. 
%and constraints \eqref{dualother} correspond to the $y^h_v$ primal variables.
The dual \eqref{cskodual} has a polynomial number of variables, and an exponential number of
\eqref{dualko} constraints. Observe that a separation oracle for these constraints
amounts to solving, for each $j\in\dbrack{k}$, $a\in\ndset_j-\nd_k$, the \knapo-instance
whose feasible solutions are given by $\I_a$ and where the rewards are 
$\gm^j_v:=\pi_v(2^j-1)-\tht_v-\sum_{h=j}^k\mu^h_v\kp_h$ for all $v\in V-\nxt(a)$, 
and $\gm^j_{\nxt(a)}:=0$.
Using standard ideas 
(see, e.g.~\cite{FriggstadS14}), we show that we can use $\Alg$ to approximately separate
over these constraints, and hence utilize the ellipsoid method in conjunction with LP duality
to obtain the desired $\al$-approximate solution to \eqref{cskolp}.

It will be convenient to note that for any $\tht,\kp\geq 0$, 
%$(\beta,\gm,\tht,\kp)$ satisfying \eqref{dualother}, \eqref{dualnoneg}, 
we have 
\begin{equation}
\sum_{j=0}^k\sum_{a\in\ndset_j-\nd_j}\gm^j_a+\sum_{a\in\fullset-\nd_k}\tht_a
+\sum_{j=0}^k(K+1)2^j\cdot\kp_j\geq 0. \label{dualpos}
\end{equation}
This follows by adding the inequality 
$\gm^j_v+\tht_v+\sum_{h=j}^k\mu^h_v\kp_h\geq 0$ (which follows from the definition of
$\gm$-values) for all $j=0,1,\ldots,k$ and all $v\in\ndset_j-\nd_j$, and noting that 
$\sum_{h=0}^j\sum_{a\in\ndset_h-\nd_h}\mu^j_a\leq (K+1)2^j$ for all $j\in\dbrack{k}$;
the latter follows since we are assuming that we have the right guess for the portal
vertices $\fullset$ and the $\dist_{a}$ length bounds, satisfying
Theorem~\ref{strucdthm}, and in particular property~\ref{prefix}. (More precisely, we
check that this condition holds for our current guess, and if not, we reject
this guess.)

\smallskip
Define %$\K(\nu)$ 
\[
\K(\nu)\ :=\ \Bigl\{(\beta,\tht,\kp):\ \eqref{dualko}, \eqref{dualnoneg},\ 
\sum_{a\in\fullset-\nd_k}\beta_a+\sum_{v\in V}\tht_v+\sum_{j=0}^k(K+1)2^j\cdot\kp_j\leq\nu\Bigr\}
\]
to be the set of feasible solutions to \eqref{cskodual} of objective
value at most $\nu$. Clearly, the optimal value of \eqref{cskodual} and \eqref{cskolp} is
the smallest $\nu$ for which $\K(\nu)\neq\es$. We may focus on $\nu\geq 0$ since the
optimal value of \eqref{cskolp} is clearly nonnegative. (Again, since we have the right
guesses, \eqref{cskolp} is feasible, and hence has an optimal solution.)
%(Note that we are assuming we have
%the right guesses for the portal vertices $\fullset$ and the $\dist_{a}$ length bounds,
%so that \eqref{cskolp} is feasible, and hence has an optimal solution.)
For a given $\nu\geq 0$ and 
$(\beta,\tht,\kp)$, we can use $\Alg$ to either show that
%$(\beta,\gm,\tht,\kp)\in\K(\al\nu)$ 
$\K(\al\nu)\neq\es$, or find a hyperplane separating $(\beta,\tht,\kp)$
from $\K(\nu)$ as follows. We first check that $\tht,\kp\geq 0$, 
%constraints \eqref{dualother}, \eqref{dualnoneg} holds, 
and 
$\sum_{a\in\fullset-\nd_k}\beta_a+\sum_{v\in V}\tht_v+\sum_{j=0}^k(K+1)2^j\cdot\kp_j\leq\nu$,
and if not return the corresponding violated inequality as the
separating hyperplane. 

Next, for every 
$j\in\dbrack{k}$ and $a\in\ndset_j-\nd_j$, we run $\Alg$ on the \knapo-instance
mentioned above. 
Although the $\gm^j_v$ rewards may be negative, since $\I_a$ is closed under taking
subpaths starting and ending at $a$, $\nxt(a)$, we can transform to an 
instance with nonnegative rewards as follows. Considering the direct $a$, $\nxt(a)$ path,
we must have $\beta_a\geq\gm^j_a$; if not, we can return the corresponding inequality of
\eqref{dualko} as a separating 
hyperplane. Now, we can consider rewards $\gm'_v=\max\{0,\gm^j_v\}$ for all
$v\notin\{a,\nxt(a)\}$ and $\gm'_a=\gm'_{\nxt(a)}=0$, and check if $\Alg$ returns a
path $\tau\in\I_a$ of $\gm'$-reward larger than $\beta_a-\gm^j_a$. If so, then there is a subpath
$\tau'\in\I_a$ of $\tau$ such that $\sum_{v\in\tau'-\nxt(a)}\gm^j_v>\beta_a$, and we
return the corresponding inequality of \eqref{dualko} as the separating hyperplane. 

If, for every index $j$ and $a\in\ndset_j-\nd_j$, this does not happen, then we claim that 
$\K(\al\nu)\neq\es$. To see this, for every $j\in\dbrack{k}$,
$a\in\ndset_j-\nd_j$, we know that every $\tau\in\I_a$
satisfies $\sum_{v\in\tau-\nxt(a)}\gm^j_v\leq\beta'_a:=\gm^j_a+\al(\beta_a-\gm^j_a)$ since
$\Alg$ is an $\al$-approximation algorithm for \knapo.
Therefore, $(\beta',\tht,\kp)$ is a feasible solution to \eqref{cskodual}. Its
objective value is
\begin{equation*}
\begin{split}
\al\Bigl(\sum_{a\in\fullset-\nd_k}\beta_a&+\sum_{v\in V}\tht_v+\sum_{j=0}^k(K+1)2^j\cdot\kp_j\Bigr)
\\ & -(\al-1)\Bigl(\sum_{j=0}^k\sum_{a\in\ndset_j-\nd_j}\gm^j_a+\sum_{a\in\fullset-\nd_k}\tht_a
+\sum_{j=0}^k(K+1)2^j\cdot\kp_j\Bigr) 
\leq\al\nu.
\end{split}
\end{equation*}
The inequality follows because $\tht,\kp\geq 0$, 
%satisfies  \eqref{dualother}, \eqref{dualnoneg}, 
so \eqref{dualpos} holds, and we also have
$\sum_{a\in\fullset-\nd_k}\beta_a+\sum_{v\in V}\tht_v+\sum_{j=0}^k(K+1)2^j\cdot\kp_j\leq\nu$.
Hence, $\K(\al\nu)\neq\es$.

Finally, let $\UB$ be some upper bound on the optimal value of \eqref{cskolp} such that
$\log\UB$ is polynomially bounded; e.g., say $\UB=\sum_{v\in V}\sum_{j=0}^k\pi_v(2^j-1)$. For
any $\e>0$, using binary search in the interval $[0,\UB]$, we can find $\nu^*$ such that
the ellipsoid method when run with the above (approximate) separation oracle returns that
$\K(\al\nu^*)\neq\es$ and certifies that $\K(\nu^*-\e)=\es$. The former shows that
$\optcskolp\leq\al\nu^*$. The latter implies that, %$\optcskolp>\nu^*-\e$. Moreover, 
for $\nu^*-\e$, the inequalities of \eqref{cskodual} returned by the separation oracle
together with the inequality 
$\sum_{a\in\fullset-\nd_k}\beta_a+\sum_{v\in V}\tht_v+\sum_{j=0}^k(K+1)2^j\cdot\kp_j\leq\nu^*-\e$
yield a polynomial-size certificate of emptiness of $\K(\nu^*-\e)$. By LP duality (or
Farkas' lemma), this implies that if we restrict \eqref{cskolp} to use only the $x^a_\tau$
variables corresponding to the inequalities of type \eqref{dualko} returned during the
execution of the ellipsoid method, then the resulting polynomial-size LP has an optimal
solution $(\bx,\by)$ of objective value
at least $\nu^*-\e$. By taking $\e$ inverse exponential in the input size, this also
implies that that the objective value of $(\bx,\by)$ is at least
$\nu^*\geq\optcskolp/\al$. 
\end{proofof}

\begin{proofof}{Lemma~\ref{csolpsolve}}
The lower bound on \optcsolp follows from Theorem~\ref{cso-strucdthm}, since we can simply
set $x^a_\tau$ for $\tau=\spath_{a,\nxt(a)}$ for every $a\in\fullset-\nd_k$.

First, we argue that one can obtain $\eta$-approximation algorithm $\Alg$ for
\ptp-knapsack monotone-reward TSP, with $\eta=O(\al)$, using the $\al$-approximation
algorithm for deadline TSP. 
Friggstad and Swamy~\cite{FriggstadS21} showed that monotone-reward TSP
can be reduced to deadline TSP losing a $(1+\e)$-factor. 
As shown in the proof of Corollary~\ref{knapvrp-approx}, we can %obtain an approximation
reduce \ptp-knapsack monotone-reward TSP %from an approximation algorithm for
to monotone-reward TSP losing a small factor. Chaining the two yields essentially 
an $(\al+2)$-approximation algorithm for \ptp-knapsack monotone-reward TSP.

Now we utilize $\Alg$ to obtain an approximate separation oracle for the dual of \csolp,
and hence obtain a $\eta$-approximate solution to \csolp. This proceeds in much the same
way as in the proof of Lemma~\ref{cskolpsolve}.
Recall that for $a,b\in\fullset$, we say that $b\prec a$ if $b$ comes before $a$ in
$\fullset$. 
For $a\in\fullset-\nd_k$, define $\visit_a:=\sum_{b\prec a}\dist_b$.
Recall that %$\I_a$ is the collection of $a$-$\next(a)$ paths $\tau$ with 
%$\mu^j(\tau)\leq 2^j$, and 
for $\tau\in\I_a$, we define 
$\pi^{a,j}(\tau):=\sum_{v\in\tau-\nxt(a):d(\tau_{a,v})\leq\dist_a-d(v,\nxt(a))}\pi_v\bigl(\visit_a+d(\tau_{a,v})+2^j-1\bigr)$.
The dual of \csolp, denoted \csodual, looks very similar to \eqref{cskodual}: the only
change is that constraints \eqref{dualko} are replaced by
\begin{equation}
\pi^{a,j}(\tau)-
\sum_{v\in\tau-\nxt(a)}\bigl(\tht_v+\sum_{h=j}^k\mu^h_v\kp_h\bigr) \leq \beta_a \qquad
\forall j\in\dbrack{k},\,a\in\ndset_j-\nd_j,\,\tau\in\I_a. \label{csodualko}
\end{equation}
For every $j\in\dbrack{k}$, $a\in\ndset_j-\nd_k$, and every node $v\neq\nxt(a)$, 
define the following non-increasing reward function $\gm^{a,j}_v:\Z_+\mapsto\R$:
\begin{equation*}
\gm^{a,j}_v(t)=\begin{cases}
\pi_v\bigl(\visit_a+t+2^j-1\bigr)-\bigl(\tht_v+\sum_{h=j}^k\mu^h_v\kp_h\bigr) &
\text{if $t\leq\dist_a-d\bigl(v,\nxt(a)\bigr)$}; \\
-\bigl(\tht_v+\sum_{h=j}^k\mu^h_v\kp_h\bigr) & \text{otherwise}.
\end{cases}
\end{equation*}
Set $\gm^{a,j}_{\nxt(a)}(t)=0$ for all $t$.
Then separating over constraints \eqref{csodualko} amounts to deciding if the optimal
value of the knapsack monotone-reward TSP instance with these reward functions is at most
$\beta_a$, for every $j\in\dbrack{k}$, $a\in\ndset_j-\nd_k$.
It will be convenient to note that for any $\tht,\kp\geq 0$, 
we have 
\begin{equation}
\sum_{j=0}^k\sum_{a\in\ndset_j-\nd_j}\gm^{a,j}_a(0)+\sum_{a\in\fullset-\nd_k}\tht_a
+\sum_{j=0}^k(K+1)2^j\cdot\kp_j\geq 0. \label{csodualpos}
\end{equation}
This follows by adding the inequality 
$\gm^{a,j}_v(0)+\tht_v+\sum_{h=j}^k\mu^h_v\kp_h\geq 0$ (which follows from the definition of
$\gm^{a,j}_v$-functions) for all $j=0,1,\ldots,k$ and all $v\in\ndset_j-\nd_j$, and noting that 
$\sum_{h=0}^j\sum_{a\in\ndset_h-\nd_h}\mu^j_a\leq (K+1)2^j$ for all $j\in\dbrack{k}$;
the latter follows since we are assuming that we have the right guess for the portal
vertices $\fullset$ and the $\dist_{a}$ length bounds, satisfying
Theorem~\ref{cso-strucdthm}, and in particular property~\ref{cso-prefix}.

As before, define %$\K(\nu)$ 
\[
\K(\nu)\ :=\ \Bigl\{(\beta,\tht,\kp):\ \eqref{csodualko}, \tht,\kp\geq 0,\ 
\sum_{a\in\fullset-\nd_k}\beta_a+\sum_{v\in V}\tht_v+\sum_{j=0}^k(K+1)2^j\cdot\kp_j\leq\nu\Bigr\}
\]
to be the set of feasible solutions to \csodual of objective
value at most $\nu$. Clearly, the optimal value of \csodual and \csolp is
the smallest $\nu$ for which $\K(\nu)\neq\es$. We may focus on $\nu\geq 0$ since the
optimal value of \csolp is clearly nonnegative. 
%(Again, since we have the right
%guesses, \eqref{cskolp} is feasible, and hence has an optimal solution.)
For a given $\nu\geq 0$ and 
$(\beta,\tht,\kp)$, we argue that we can use $\Alg$ to either show that
$\K(\eta\nu)\neq\es$, or find a hyperplane separating $(\beta,\tht,\kp)$
from $\K(\nu)$. 

We first check that $\tht,\kp\geq 0$, and 
$\sum_{a\in\fullset-\nd_k}\beta_a+\sum_{v\in V}\tht_v+\sum_{j=0}^k(K+1)2^j\cdot\kp_j\leq\nu$,
and if not return the corresponding violated inequality as the
separating hyperplane. Next, for every 
$j\in\dbrack{k}$ and $a\in\ndset_j-\nd_j$, we run $\Alg$ on the knapsack monotone-reward
TSP-instance mentioned above. 
Although the $\gm^j_v$ functions may take negative values, since $\I_a$ is closed under
taking subpaths starting and ending at $a$, $\nxt(a)$, we can transform to an 
instance with nonnegative rewards as follows. Considering the direct $a$, $\nxt(a)$ path,
we must have $\beta_a\geq\gm^{a,j}_a(0)$; if not, we can return the corresponding inequality of
\eqref{csodualko} as a separating hyperplane. 
Now, we can consider the reward functions $\tgm^{a,j}_v(t)=\max\{0,\gm^{a,j}_v(t)\}$ for
all $t$ and all $v\notin\{a,\nxt(a)\}$, and $\tgm^{a,j}_a(t)=\tgm^{a,j}_{\nxt(a)}(t)=0$
for all $t$.
We check if $\Alg$ returns a path $\tau\in\I_a$ of reward larger than
$\beta_a-\gm^{a,j}_a(0)$. If so, then there is a subpath 
$\tau'\in\I_a$ of $\tau$ such that 
$\sum_{v\in\tau'-\nxt(a)}\gm^{a,j}_v\bigl(d(\tau'_{a,v})\bigr)>\beta_a$, and we
return the corresponding inequality of \eqref{csodualko} as the separating hyperplane. 

If, for every index $j$ and $a\in\ndset_j-\nd_j$, this does not happen, then we claim that 
$\K(\eta\nu)\neq\es$. To see this, for every $j\in\dbrack{k}$,
$a\in\ndset_j-\nd_j$, we know that every $\tau\in\I_a$
satisfies 
\[
\sum_{v\in\tau-\nxt(a)}\gm^{a,j}_v\bigl(d(\tau_{a,v})\bigr)\leq
\beta'_a:=\gm^{a,j}(0)+\eta\bigl(\beta_a-\gm^{a,j}(0)\bigr)
\]
since $\Alg$ is an $\eta$-approximation algorithm for knapsack monotone-reward TSP.
Therefore, $(\beta',\tht,\kp)$ is a feasible solution to \csodual. Its
objective value is
\begin{equation*}
\begin{split}
\eta\Bigl(\sum_{a\in\fullset-\nd_k}\beta_a&+\sum_{v\in V}\tht_v+\sum_{j=0}^k(K+1)2^j\cdot\kp_j\Bigr)
\\ & -(\eta-1)\Bigl(\sum_{j=0}^k\sum_{a\in\ndset_j-\nd_j}\gm^{a,j}_a(0)+\sum_{a\in\fullset-\nd_k}\tht_a
+\sum_{j=0}^k(K+1)2^j\cdot\kp_j\Bigr) 
\leq\eta\nu.
\end{split}
\end{equation*}
The inequality follows because $\tht,\kp\geq 0$, 
%satisfies  \eqref{dualother}, \eqref{dualnoneg}, 
so \eqref{csodualpos} holds, and we also have
$\sum_{a\in\fullset-\nd_k}\beta_a+\sum_{v\in V}\tht_v+\sum_{j=0}^k(K+1)2^j\cdot\kp_j\leq\nu$.
Hence, $\K(\eta\nu)\neq\es$.

The rest of the proof proceeds as in the proof of Lemma~\ref{cskolpsolve}.
\end{proofof}

%\newpage
\appendix \label{appstart}

%\medskip
\section{Adversarial orderings can be arbitrarily bad for correlated  knapsack} 
\label{append-badexample} \label{append-corrknap}
Consider an instance of correlated stochastic knapsack on the set of items $[n]$
with budget $W>2^{n+1}$. Let $\nsize[i]$ and $\nrewd[i]$ denote respectively the random
size and random reward of $i$, which follows the following distribution.
\begin{equation*}
(\nsize[i],\nrewd[i])=
\begin{cases}
\bigl(\nsize[i]^{(1)}:=W-2^{n-i+1}+1,\,\nrewd[i]^{(1)}:=1\bigr) & 
\quad \text{with probability $\frac{1}{n}$} \\
%\qquad
%(\nsize[i],\nrewd[i])=
\bigl(\nsize[i]^{(2)}:=2^{n-i},\,\nrewd[i]^{(2)}:=0\bigr) & 
\quad \text{with probability $1-\frac{1}{n}$.}
\end{cases}
\end{equation*}
%Suppose that $\nsize[i]$
%instantiates to $2^{n-i}$ with probability $1-1/n$ getting zero reward, and that $S_i$
%instantiates to $W-2^{n-i+1}+1$ with probability $1/n$ getting a reward of 1. 
Note that at most one item can obtain positive reward since $W/2<W-2^{n-i+1}+1$ for all
$i\in[n]$. 

Suppose that we are forced to process the items in the ordering $1,\ldots,n$
deciding at each step whether we attempt to insert the current item into the knapsack or
abandon it forever. Let $j$ be the first item that we choose to insert into the
knapsack. It instantiates to size $2^{n-j}$ with probability $1-1/n$. If this happens we
get zero reward. The residual budget becomes $W-2^{n-j}$, which is less than the
$\nsize[i]^{(1)}$-sizes of 
%associated sizes of the 
items $j+1,\ldots,n$ (which yield positive reward). Therefore, by
processing the items in this ordering the expected reward is at most $1/n$. On the other
hand, suppose we process the items in the reverse order $n,\ldots,1$. If we attempt to insert
items $n,n-1,\ldots,j$ and get no positive reward from any of them, the residual
budget is
%\begin{equation*}
$W-\sum_{k=j}^{n}2^{n-k}=W-2^{n-j+1}+1>\nsize[j-1]^{(1)}$,
%\end{equation*}
%which is larger than the possible size $W-2^{n-j+2}+1$ of $j-1$ (which corresponds to
%positive reward for $j-1$). 
which means that item $j-1$ can be inserted and would yield reward $1$ with probability
$\frac{1}{n}$. Thus, the probability that no item gives positive reward is
$(1-1/n)^n\leq e^{-1}$. Hence, the expected reward that we obtain by attempting to insert
items in the ordering $n,\ldots,1$ is at least $(1-e^{-1})$, which is $\Omega(n)$ times
larger than the expected reward from any policy that is forced to process the items in the
ordering $1,\ldots,n$. 

\end{document}